\newif\ifanonymous
\newif\ifdraft
\anonymoustrue
\draftfalse
\InputIfFileExists{localflags}{}

\PassOptionsToPackage{usenames,dvipsnames}{xcolor}
\documentclass[sigconf]{acmart}

\ifdraft
\fi

\usepackage{versions}
\usepackage{xr}

\excludeversion{conf}
\includeversion{full}

\usepackage[T1]{fontenc}
\usepackage{latexsym}
\usepackage{mathrsfs}
\usepackage{makecell}
\usepackage{graphics,subcaption}
\usepackage{listings}
\usepackage{todonotes}
\usepackage{regexpatch}
\usepackage{siunitx}
\usepackage{soul}
\usepackage{array,graphicx}
\usepackage{booktabs}
\usepackage{xspace}
\usepackage{geometry}
\usepackage[most]{tcolorbox}
\usepackage{listings}
\usepackage{tikz-cd}
\usepackage{pbox}
\usepackage{url}

\copyrightyear{2023} 
\acmYear{2023} 
\setcopyright{acmlicensed}\acmConference[CCS '23]{Proceedings of the 2023 ACM SIGSAC Conference on Computer and Communications Security}{November 26--30, 2023}{Copenhagen, Denmark}
\acmBooktitle{Proceedings of the 2023 ACM SIGSAC Conference on Computer and Communications Security (CCS '23), November 26--30, 2023, Copenhagen, Denmark}
\acmPrice{15.00}
\acmDOI{10.1145/3576915.3623090}
\acmISBN{979-8-4007-0050-7/23/11}

\usepackage{balance} 
\usepackage{multirow}
\usepackage{stackengine}
\usepackage{booktabs}
\usepackage{stmaryrd}
\usepackage{color}
\usepackage{graphics}
\usepackage{multicol}
\usepackage{hyperref}
\usepackage{stackengine}
\usepackage{mathtools,eqparbox}
\usepackage{ebproof}
\usepackage{scalerel}
\usepackage{adjustbox}
\usepackage{thm-restate}
\usepackage{graphicx}
\usepackage{pifont}
\usepackage{enumitem}
\usepackage{algorithm2e}
\SetKw{KwBy}{by}

\newcommand*\OK{\ding{51}}
\newcommand*\NO{\ding{55}}

\usepackage{tikz}
\usetikzlibrary{decorations.pathreplacing,angles,quotes}
\usetikzlibrary{shapes.geometric, arrows}
\usetikzlibrary{positioning,arrows.meta}
\definecolor{lightblue}{rgb}{0.68, 0.85, 0.9}
\tikzset{
	block/.style={
		draw, 
		rounded corners,
		text width=15em,
		text centered,
		minimum height=5ex,
		fill=lightblue
	},
	line/.style={
		draw, 
		-Stealth
	},
}
\tikzset{ 
	roundnode/.style={circle, draw=green!60, fill=green!5, minimum size=5ex},
}
\tikzset{
squarednode/.style={rectangle, draw=red!60, fill=red!5, minimum size=5ex},
}

\tikzset{
	block1/.style={
		draw, 
		rounded corners,
		text width=16.5em,
		text centered,
		minimum height=5ex,
		fill=lightblue
	},
	line/.style={
		draw, 
		-Stealth
	},
}
\tikzstyle{process} = [rectangle, minimum width=2cm, minimum height=1cm, text centered, text width=2cm, draw=black, fill=white]
\tikzstyle{process1} = [rectangle, minimum width=1cm, minimum height=1cm, text centered, text width=1cm, draw=black, fill=white]
\tikzstyle{arrow} = [thick,->,>=stealth]

\lstdefinelanguage
[x64]{Assembler}     
{morekeywords={cbnz,ldr, cmp, ldp, eq, ge, mul, nop, add, lsr, eor, and, M, Z, N, V, ldtm, mov, ret,b, bl, ble, bne,%
		udiv, bics, ldrb, ror, sxtw, for, do, end,  procedure, or, adrp,  %
		xzr, wzr,x0,X1,x1,X2,x2,X3,x3,w4,x4,x5,x6,x7,x8,x9, w0, w1, w2, w3,w4,w5,w6,w7,w8,w9,%
		x10,x11,x12,x13,x14,x15,x16,x17,x18,x19,%
		x20,x21,x22,x23,x24,x25,x26,x27,x28,x29,%
		x30,x31,sp},
	alsoletter={x\#}}[strings,comments,keywords] %

\definecolor{mGreen}{rgb}{0,0.6,0}
\definecolor{mGray}{rgb}{0.5,0.5,0.5}
\definecolor{mPurple}{rgb}{0.58,0,0.82}
\definecolor{backgroundColour}{rgb}{0.95,0.95,0.92}

\makeatletter
\def\lst@numbersymbol{}
\lst@Key{numbersymbol}{}{\def\lst@numbersymbol{#1}}
\lst@Key{numbers}{none}{%
	\let\lst@PlaceNumber\@empty
	\lstKV@SwitchCases{#1}%
	{none:\\%
		left:\def\lst@PlaceNumber{\llap{\normalfont
				\lst@numberstyle{\thelstnumber\lst@numbersymbol}\kern\lst@numbersep}}\\%
		right:\def\lst@PlaceNumber{\rlap{\normalfont
				\kern\linewidth \kern\lst@numbersep
				\lst@numberstyle{\lst@numbersymbol\thelstnumber}}}%
	}{\PackageError{Listings}{Numbers #1 unknown}\@ehc}}
\makeatother
\makeatletter
\renewcommand*\thelstnumber{\the\numexpr 4*\c@lstnumber\relax}
\lstdefinestyle{asmstyle}{
	language=[x64]Assembler,
	basicstyle=\fontsize{7.5}{9}\ttfamily,
	keywordstyle=\bfseries\color{darkgray},
	breaklines=true,
	mathescape=true,
	keepspaces=true,
	showspaces=false,
	showstringspaces=false,    
	xleftmargin={0.75cm},
	firstnumber = 0,
	numbers=left,%
	numberstyle={\footnotesize \color{mGreen}},
	numbersep=5pt, 
}

\lstdefinestyle{birstyle}{
	language=[x64]Assembler,
	basicstyle=\fontsize{7.5}{9}\ttfamily,
	keywordstyle=\bfseries\color{mGray},
	breaklines=true,
	mathescape=true,
	keepspaces=true,
	showspaces=false,
	showstringspaces=false,    
	xleftmargin={0.75cm},
	morekeywords={Address, Imm64, jmp, var, Const, Type_Imm, BLE_Exp, Exp_Den, assign, halt}, 
}

\lstdefinestyle{cstyle}{
	commentstyle=\color{mGreen},
	keywordstyle=\bfseries\color{mGreen},
	numberstyle=\tiny\color{mGray},
	stringstyle=\color{mPurple},
	breakatwhitespace=false,         
	breaklines=true,                 
	captionpos=b,                    
	keepspaces=true,                 
	showspaces=false,    
	numbers=none,          
	showstringspaces=false,
	showtabs=false,                  
	tabsize=2,
	xleftmargin={0.2cm},
	emph={int,char,double,float,unsigned, const, long},
	emphstyle={\bfseries\color{purple}},
	language=C
}

\lstdefinestyle{mlstyle}{
	language=caml,
	columns=[c]fixed,
	basicstyle=\small\ttfamily,
	keywordstyle=\bfseries,
	upquote=true,
	breaklines=true,
	showstringspaces=false,
	stringstyle=\color{blue},
	xleftmargin={0.1cm},
	aboveskip=1ex,
	literate={'"'}{\textquotesingle "\textquotesingle}3,
	morekeywords={new, out, event, assume}
}
\makeatletter
\DeclareRobustCommand{\shortto}{%
	\mathpalette\short@to\relax%
}

\newcommand{\short@to}[2]{%
	\mkern2mu
	\clipbox{{.5\width} 0 0 0}{$\m@th#1\vphantom{+}{\rightarrow}$}%
}
\makeatother

\makeatletter
\DeclareRobustCommand{\shortdash}{%
	\mathpalette\short@dash\relax%
}

\newcommand{\short@dash}[2]{%
	\mkern2mu
	\clipbox{0 0 {.6\width} 0}{$\m@th#1\vphantom{+}{\rightarrow}$}%
}
\makeatother

\lstdefinelanguage{mlang} {morekeywords={system,type,var,process,do,observable,configuration,if,then,else,array,init,while,function,repeat,fun,for,return,ret}, 
	sensitive=false,
	morestring=[b]", }

\let\origthelstnumber\thelstnumber
\makeatletter
\newcommand*\Suppressnumber{%
	\lst@AddToHook{OnNewLine}{%
		\let\thelstnumber\relax%
	}%
}

\newcommand*\Reactivatenumber[1]{%
	\setcounter{lstnumber}{\numexpr#1-1\relax}
	\lst@AddToHook{OnNewLine}{%
		\let\thelstnumber\origthelstnumber%
		\refstepcounter{lstnumber}
	}%
}

\newtheorem{property}{Property}

\let\orgautoref\autoref
\renewcommand{\autoref}
        {\def\equationautorefname{Eq.}%
        	\def\figureautorefname{Fig.}%
        	\def\subfigureautorefname{Fig.}%
         \def\Itemautorefname{Item}%
         \def\tableautorefname{Table}%
         \def\algorithmautorefname{Algorithm}%
         \def\paragraphautorefname{Paragraph}%
         \def\sectionautorefname{Sec.}%
         \def\subsectionautorefname{Sec.}%
         \def\subsubsectionautorefname{Sec.}%
         \def\chapterautorefname{Chapter}%
         \def\partautorefname{Part}%
         \def\goalautorefname{Goal}%
         \def\reqautorefname{Req.}%
         \def\adviceautorefname{Rule}%
         \def\parameterautorefname{Param.}%
         \def\definitionautorefname{Def.}%
         \def\definitionsautorefname{Def.}%
         \def\propertyautorefname{Property}%
         \def\lemmaautorefname{Lemma}%
         \def\theoremautorefname{Thm.}%
         \orgautoref}

\renewcommand\sectionautorefname{Section}

\newcommand{\ledot}{\mathrel{\ooalign{\hss\raise.200ex\hbox{$\cdot$}\hss\cr$\le$}}}
\newcommand{\gedot}{\mathrel{\ooalign{\hss\raise.200ex\hbox{$\cdot$}\hss\cr$\ge$}}}

\newcommand{\framework}{{\sc CryptoBap}}
\newcommand{\true}{\mathit{true}}

\newcommand{\tuple}[1]{( #1 )}
\newcommand{\naturalnum}{\mathscr{N}}
\newcommand{\filter}[1]{\downharpoonright_{#1}}
\renewcommand{\iff}{\Leftrightarrow}
\newcommand{\tree}{\sbirc{T}}
\newcommand{\node}{\sbirc{node}}
\newcommand{\branchingnode}{\sbirc{Branch}}
\newcommand{\leafnode}{\sbirc{Leaf}}
\newcommand{\nodecond}{\sbirc{\gamma}}
\newcommand{\nodeevent}{\sbirc{ev}}
\newcommand{\type}{\imlc{t}}
\newcommand{\secparam}{n}
\newcommand{\mixedexec}[2]{#1\{#2\}}
\newcommand{\executions}{R}

\newlength\shlength
\newcommand\xshlongvec[2][0]{\setlength\shlength{#1pt}%
  \stackengine{-5.6pt}{$#2$}{\smash{$\kern\shlength%
    \stackengine{7.55pt}{$\mathchar"017E$}%
      {\rule{\widthof{$#2$}}{.57pt}\kern.4pt}{O}{r}{F}{F}{L}\kern-\shlength$}}%
      {O}{c}{F}{T}{S}}
      
\newcommand{\memsymb}{{Mem}}
\newcommand{\heap}{{heap}}
\newcommand{\update}{mstore}
\newcommand{\memload}{mload}

\newcommand{\birsymb}{\birc{{BIR}}}
\newcommand{\sbirsymb}{\sbirc{{SBIR}}}
\newcommand{\imlsymb}{\imlc{{IML}}}
\newcommand{\statesymb}{{s}}
\newcommand{\envsymb}{\eta}
\newcommand{\bnfconcat}{\!::\!}
\newcommand{\bnfsep}{\;|\;}
\newcommand{\bnfdef}{\;\;:=\;\;}
\newcommand{\bnfc}[1]{\operatorname{\textnormal{\textbf{#1}}}}
\newcommand{\domain}[1]{\mathit{dom}(#1)}
\newcommand{\imagesym}{\mathit{range}}
\newcommand{\image}[1]{\imagesym(#1)}
\newcommand{\var}[1]{\birc{#1}}
\newcommand{\symvar}[1]{\sbirc{#1}}
\newcommand{\imlvar}[1]{\imlc{#1}}
\newcommand{\randommem}{\var{RM}}
\newcommand{\randommemidx}[1][k]{\var{r}_{\var{#1}}}

\newcommand{\cryptoOp}{\birc{Op}}
\newcommand{\cryptoop}{\birc{op}}
\newcommand{\Loop}[1]{\birc{{loop}}(#1)}
\newcommand{\Input}[1]{\birc{{In}}(#1)}
\newcommand{\Output}[1]{\birc{{Out}}(#1)}
\newcommand{\crypto}[1]{\birc{{Cr}}(#1)}
\newcommand{\event}[1]{\birc{{Ev}}(#1)}
\newcommand{\freshv}[1]{\birc{{Fr}}(#1)}
\newcommand{\attacker}{\mathcal{A}}
\newcommand{\oraclesym}{L}
\newcommand{\oracle}[1]{\oraclesym(#1)}

\newcommand{\reg}[1]{\birc{r_{#1}}}
\newcommand{\simrel}[3]{#1 \sim_{_{#2}} #3}
\newcommand{\simreltraces}[4]{#1 \sim_{_{#2,#3}} #4}

\newcommand{\mi}[1]{\ensuremath{\mathit{#1}}}

\newcommand{\mt}[1]{\ensuremath{\texttt{#1}}}

\newcommand{\mf}[1]{\ensuremath{\mathbf{#1}}}

\newcommand{\ms}[1]{\ensuremath{\mathsf{#1}}}

\newcommand{\neutcol}[0]{black}
\newcommand{\bircol}[0]{black}
\newcommand{\sbircol}[0]{RedOrange}
\newcommand{\imlcol}[0]{RoyalBlue}

\newcommand{\col}[2]{\ensuremath{{\color{#1}{#2}}}}

\newcommand{\birc}[1]{\mf{\col{\bircol}{#1}}}
\newcommand{\sbirc}[1]{\ms{\col{\sbircol}{#1}}}
\newcommand{\imlc}[1]{\mi{\col{\imlcol}{#1}}}

\newcommand{\bl}[1]{\mt{\col{\neutcol }{#1}}}

\newcommand{\rmcolor}[1]{{\col{black}{#1}}\!}

\newcommand{\biridx}{\birc{b}}
\newcommand{\birstates}{\birc{S}}
\newcommand{\birstate}[1]{\birc{{\statesymb_{#1}}}}
\newcommand{\birenv}{\birc{\envsymb}}
\newcommand{\birpc}{\birc{{pc}}}
\newcommand{\birvar}{\birc{{Bvar}}}
\newcommand{\birval}{\birc{{Bval}}}
\newcommand{\birexp}{\birc{{Bexp}}}
\newcommand{\birprog}{\birc{{P}}}

\newcommand{\eventsym}{a}
\newcommand{\birevent}{\birc{\eventsym}}
\newcommand{\birevents}{\birc{{E}}}
\newcommand{\birprogset}{\birc{\birprog_1 , ... , \birprog_m}}

\newcommand{\birprobinitstates}[1]{\birc{{S_{{init}}^{b,n}}}}
\newcommand{\birprobinitstate}[1]{\birc{{\statesymb_{{init}}^{b,n}}}}
\newcommand{\birrand}{\birc{\Re}}

\newcommand{\sbirrand}{\sbirc{\Re}}

\newcommand{\xrightarrowsgl}[2][]{\birc{%
  \xrightarrow[]{\bl{#2}}\mathrel{}_\bl{#1}
}}
\newcommand{\starxrightarrowsgl}[2][]{\birc{%
		\xrightarrow[]{\bl{#2}}\mathrel{}^{\bl{*}}_\bl{#1}
}}
\newcommand{\plusxrightarrowsgl}[2][]{\birc{%
		\xrightarrow[]{\bl{#2}}\mathrel{}^{\bl{+}}_\bl{#1}
}}
\newcommand{\nxrightarrowsgl}[2][]{\birc{%
		\xrightarrow[]{\bl{#2}}\mathrel{}^{\bl{n}}_\bl{#1}
}}

\newcommand{\birtrans}[4]{#3 \xrightarrowsgl{#1}_{#2} #4}
\newcommand{\birtransstar}[4]{#3 \starxrightarrowsgl{#1}_{#2} #4}
\newcommand{\birtransmulti}[4]{#3 \plusxrightarrowsgl{#1}_{#2} #4}
\newcommand{\birtransmultin}[4]{#3 \nxrightarrowsgl{#1}_{#2} #4}

\newcommand{\birexecutions}{\birc{\mathcal{R}^{b}}}

\newcommand{\birpr}{\mathit{{pr}}}

\newcommand{\sbiridx}{\sbirc{s}}
\newcommand{\sbirstates}{\sbirc{{S}}}
\newcommand{\sbirstate}[1]{\sbirc{{\statesymb_{#1}}}}
\newcommand{\sbirenv}{\sbirc{\envsymb}}
\newcommand{\sbirpcond}{\sbirc{\phi}}
\newcommand{\sbirvals}{\sbirc{{SE}}}
\newcommand{\sbirvaluation}{\mathit{H}}
\newcommand{\sbirvaluations}{\mathcal{H}}
\newcommand{\sbirevent}{\sbirc{\eventsym}}
\newcommand{\sbirevents}{\sbirc{{E}}}

\newcommand{\defeq}{\mathrel{\stackrel{\makebox[0pt]{\mbox{\normalfont\tiny def}}}{=}}}
\newcommand{\lbls}{\birc{L}}
\newcommand{\entr}{\xi}
\newcommand{\pcincsym}{\mathit{ret}}
\newcommand{\pcinc}[1]{\pcincsym(#1)}
\newcommand{\transition}[2]{{#1} \birc{\rightarrow} {#2}}

\newcommand{\transitionsf}[4]{{#1} \birc{\rightarrow}^{#2}_{#3} {#4}}

\newcommand{\pcexitsym}{\mathit{exit}}
\newcommand{\pcexit}[1]{\pcexitsym(#1)}
\newcommand{\loopprocesssym}{\mathit{processLoop}}
\newcommand{\loopprocess}[1]{\loopprocesssym(#1)}
\newcommand{\loopbodysym}{\mathit{LoopProc}}
\newcommand{\loopbody}[1]{\loopbodysym(#1)}

\newcommand{\xrightarrowdbl}[2][]{\sbirc{%
  \xrightarrow[]{\bl{#2}}\mathrel{\mkern-14mu}\rightarrow_\bl{#1}
}}
\newcommand{\plusxrightarrowdbl}[2][]{\sbirc{%
		\xrightarrow[]{\bl{#2}}\mathrel{\mkern-14mu}\rightarrow^{\bl{+}}_\bl{#1}
}}
\newcommand{\nxrightarrowdbl}[2][]{\sbirc{%
		\xrightarrow[]{\bl{#2}}\mathrel{\mkern-14mu}\rightarrow^{\bl{n}}_\bl{#1}
}}
\newcommand{\sbirtrans}[4]{#3 \xrightarrowdbl{#1}_{#2} #4}
\newcommand{\sbirtransmulti}[4]{#3 \plusxrightarrowdbl{#1}_{#2} #4}
\newcommand{\sbirtransmultin}[4]{#3 \nxrightarrowdbl{#1}_{#2} #4}

\newcommand{\syminput}[1]{\sbirc{{In}}(#1)}
\newcommand{\symoutput}[1]{\sbirc{{Out}}(#1)}
\newcommand{\symcrypto}[1]{\sbirc{{Cr}}(#1)}
\newcommand{\symevent}[1]{\sbirc{Ev}(#1)}
\newcommand{\symfreshv}[1]{\sbirc{{Fr}}(#1)}
\newcommand{\symloop}[1]{\sbirc{{loop}}(#1)}

\newcommand{\sbirexecutions}{\sbirc{\mathcal{R}^{s}}}

\newcommand{\sbirprobinitstates}[1]{\sbirc{{S_{{init}}^{s,n}}}}
\newcommand{\sbirprobinitstate}[1]{\sbirc{{\statesymb_{{init}}^{s,n}}}}
\newcommand{\advmem}{\memsymb_{\birc{\mathcal{A}}}}

\newcommand{\sbirtoiml}[1]{\llparenthesis #1 \rrparenthesis}
\newcommand{\imlidx}{\imlc{\iota}}

\newcommand{\imlstates}{\imlc{{S}}}
\newcommand{\imlstate}[1]{\imlc{\statesymb_{#1}}}

\newcommand{\imlenv}{\imlc{\envsymb}}
\newcommand{\imlexp}{\imlc{{Iexp}}}
\newcommand{\imlvarspace}{\imlc{{Ivar}}}
\newcommand{\imlvals}{\imlc{{BS}}}

\newcommand{\imlprocessp}{\imlc{{P}}}
\newcommand{\imlprocessq}{\imlc{{Q}}}
\newcommand{\imlmultisetprocs}{\imlc{\mathcal{Q}}}
\newcommand{\imlevent}{\imlc{a}}
\newcommand{\imleval}[1]{\llbracket #1 \rrbracket}
\newcommand{\xrightarrowtpl}[2][]{\imlc{%
  \xrightarrow[]{\bl{#2}}\mathrel{\mkern-14mu}\rightarrow\mathrel{\mkern-14mu}\rightarrow_{\bl{#1}}
}}

\newcommand{\mixxrightarrowtpl}[2][]{%
		\imlc{\xrightarrow[]{\bl{#2}}\mathrel{\mkern-14mu}\rightarrow\mathrel{\mkern-14mu}\rightarrow^{\bl{+}}_{\bl{#1}}}}

\newcommand{\imltrans}[4]{#3 \xrightarrowtpl{#1}_{\imlc{#2}} #4}
\newcommand{\imltransmulti}[4]{#3 \mixxrightarrowtpl{#1}_{\imlc{#2}} #4}
\newcommand{\imlexecution}{\imlc{{R}}}
\newcommand{\imlexecutions}{\imlc{\mathcal{R}^{\iota}}}
\newcommand{\imltrace}{\imlc{t}}
\newcommand{\imltraces}{\imlc{{T}}}
\newcommand{\imlfunc}{\imlc{{op}}}

\newcommand{\imlchannel}{\imlc{{c}}}

\newcommand{\simreltra}[4]{#1 \sim_{_{#2 , #3}} #4}
\newcommand{\simreltratraces}[5]{#1 \sim_{_{#2 , #3, #4}} #5}
\newcommand{\imlev}[1]{\imlc{ev}(#1)}
\newcommand{\imlevsym}{\imlc{ev}}
\newcommand{\imlfreshev}[1]{\imlc{{fr}}(#1)}

\newcommand{\probdist}[1]{\frac{1}{2^{#1}}}
\newcommand{\imlprog}{\imlc{ {I}}}

\newcommand{\imlpr}{\mathit{{pr}}}

\newcommand{\traceproperty}{\psi}
\newcommand{\tracepropertyneg}{\psi_\neg}
\newcommand{\insecurity}[2]{\mathbf{insec} ( #1, #2)}

\newcommand{\nextimlenv}{\imlc{\envsymb}'}

\newcommand{\mixedbiridx}{{\scaleto{\imlidx\times\biridx}{4pt}}}
\newcommand{\mixedbirexecutions}{\mathcal{R}^{\mixedbiridx}}
\newcommand{\mixedbirexecution}{R^{\mixedbiridx}}
\newcommand{\mixedbirtraces}{{T}^{\mixedbiridx}}
\newcommand{\mixedbirtrace}{t^{\mixedbiridx}}

\newcommand{\mixedbirprog}{\mixedexec{\mixediml}{\birprog}}
\newcommand{\mixedbirprogidx}[1]{\mixedexec{\mixediml}{#1}}
\newcommand{\mixedbirprogs}{\mixedexec{\mixediml}{\birprogset}}

\newcommand{\mixedsbiridx}{{\scaleto{\imlidx\times\sbiridx}{4pt}}}
\newcommand{\mixedsbirexecutions}{\mathcal{R}^{\mixedsbiridx}}
\newcommand{\mixedsbirexecution}{R^{\mixedsbiridx}}
\newcommand{\mixedsbirtraces}{{T}^{\mixedsbiridx}}
\newcommand{\mixedsbirtrace}{t^{\mixedsbiridx}}

\newcommand{\miximltrans}[5]{#4 \xrightarrowtpl{#1}_{{\imlc{#2}},#3} #5}
\newcommand{\miximltransmulti}[5]{#4 \mixxrightarrowtpl{#1}_{{\imlc{#2}},#3} #5}

\newcommand{\mixediml}{\imlc{I}}

\newcommand{\imlsbirstates}{{S^{\mixedsbiridx}}}
\newcommand{\miximlsbirevent}{a^{\mixedsbiridx}}

\newcommand{\imlbirstates}{{S^{\mixedbiridx}}}
\newcommand{\miximlbirevent}{a^{\mixedbiridx}}

\newcommand*{\suchthat}{\;%
  \ifnum\currentgrouptype=16\middle|\else%
     \iftoggle{WithinBracMacro}{\middle|}{|}%
  \fi%
  \;}%
\newcommand{\sbirtracesproof}{%
  \mixedsbirtraces(\mixedbirprog, \sbirenv_{\sbirc{0}}[ \randommem \mapsto \randomsymvals_{\sbirc{k}} , \randommemidx \mapsto 0])}
\newcommand{\imltracesproof}{%
\imltraces(\imltransprog,n)}
\newcommand{\imlsystemproof}{%
    \imltransprog}
\newcommand{\birtracesproof}{%
    \mixedbirtraces(\mixedbirprog, \birenv_{\birc{0}}[ \randommem \mapsto \randomvals_{\imlc{k}} , \randommemidx \mapsto 0])}

\newcommand{\randomvals}{\imlvar{rm}}
\newcommand{\randomsymvals}{\symvar{rm}}

\newcommand{\rng}[1]{\mathit{rng}(#1)}
\newcommand{\rngsym}{\mathit{rng}}

\newcommand{\trace}[1]{\mathit{tr}(#1)}

\newcommand{\eventtrace}{\mathit{t}}

\newcommand{\birsimrel}[3]{#1  \sim_{_{#2}}{} #3}

\newcommand{\colb}{\birc{B}}
\newcommand{\colsb}{\sbirc{SB}}
\newcommand{\imlbir}{\imlsymb$_\text{\colb}$}
\newcommand{\imlsbir}{\imlsymb$_\text{\colsb}$}

\newcommand{\initialimlsbir}{%
\mixedexec{\mixediml}{\sbirstate{0}} = \tuple{True, \sbirenv_{\sbirc{0}}[ \randommem \mapsto \randomsymvals_{\sbirc{k}} , \randommemidx \mapsto 0], \birpc_{\birc{0}}}}

\newcommand{\initialiml}{%
\imlstate{0} = (\imlenv_{\imlc{0}}, \imlprocessq^{\mathit{full}})}

\newcommand{\initialimlbir}{%
\mixedexec{\mixediml}{\birstate{0}} = \tuple{\birenv_{\birc{0}}[ \randommem \mapsto \randomvals_{\imlc{k}} , \randommemidx \mapsto 0], \birpc_{\birc{0}}}}

\newcommand{\imltransprog}{\imlprog \{ \sbirtoiml{\birprog} \} }
\newcommand{\imltransprogidx}[1]{\imlprog \{ \sbirtoiml{#1} \} }
\newcommand{\sbirtoimlprogset}{\sbirtoiml{\birprog_{\birc{1}}} , ... , \sbirtoiml{\birprog_{\birc{m}}}}
\newcommand{\imltransprogs}{\imlprog \{ \sbirtoimlprogset \} }

\newcommand{\looppath}{\sbirc{\pi}}

\newcommand{\itrsbirenv}[1]{\sbirc{\envsymb}_{_{#1}}}
\newcommand{\itrsbirpcond}[1]{\sbirc{\phi}_{_{#1}}}

\newcommand{\mixmultisetprocsimlsbir}{{\mathcal{Q}}^\mixedsbiridx}
\newcommand{\mixmultisetprocsimlbir}{{\mathcal{Q}}^\mixedbiridx}
\begin{document}
\title{\framework{}: A Binary Analysis Platform for Cryptographic Protocols}
\author{Faezeh Nasrabadi}
\affiliation{%
 \institution{CISPA Helmholtz Center for Information Security}
 \city{}
 \country{}
 \orcid{0009-0005-3659-7755}
}\email{faezeh.nasrabadi@cispa.de}

\author{Robert Künnemann}
\affiliation{%
 \institution{CISPA Helmholtz Center for Information Security}
 \city{}
 \country{}
 \orcid{0000-0003-0822-9283}
}\email{robert.kuennemann@cispa.de}

\author{Hamed Nemati}
\affiliation{%
  \institution{CISPA Helmholtz Center for Information Security}
 \city{}
 \country{}
 \orcid{0000-0001-9251-3679}
}\email{hamed.nemati@cispa.de}


\begin{abstract}
  We introduce~\framework{}, a platform to verify weak secrecy and authentication for the (ARMv8 and RISC-V) machine code of cryptographic protocols. We achieve this by first transpiling the binary of protocols into an intermediate representation and then performing a crypto-aware symbolic execution to automatically extract a model of the protocol that represents all its execution paths. Our symbolic execution  resolves indirect jumps and supports bounded loops using the loop-summarization technique, which we fully automate. 
  The extracted model is then translated into models amenable to automated verification via ProVerif and CryptoVerif using a third-party toolchain.
  We prove the soundness of the proposed approach and used \framework{} to verify multiple case studies ranging from toy examples to real-world protocols, TinySSH, an implementation of SSH, and WireGuard, a modern VPN protocol.\\

  \emph{This paper uses  \birc{co}\sbirc{lo}\imlc{rs} to distinguish between different abstraction layers in our modeling and verification~\cite{patrignani2020should}}.

\end{abstract}

\maketitle

\section{Introduction}

Cryptographic protocols are a vital part of end-user security on the Internet. Therefore, devising techniques to obtain high-assurance guarantees about their correctness and security is highly desirable. 
Nevertheless, despite the simplicity of such protocols, their design and implementation are notoriously error-prone, and there have been many attacks targeting either their design (e.g., man-in-the-middle attacks like the triple-handshake attack on TLS~\cite{bhargavantriple2014}) or implementation (e.g., Heartbleed, \texttt{CVE-2014-0160}). 

Formal methods suggest a rigorous foundation to find bugs and provide guarantees about the correctness and security of cryptographic protocols.
Rigorous techniques that have been used so far to reason about such protocols belong to two main lines of research. One is based on the \emph{Dolev-Yao model}, a symbolic model of cryptography~\cite{Dolev1981}, while the other, the \emph{computational approach}~\cite{Goldwasser1984}, is closer to reality 
and gives stronger and more realistic security assurance.

In recent years, significant progress has been made using these techniques to derive rigorous guarantees based on either abstract protocol specifications or
the concrete implementations in high-level programming languages such as \emph{C} or \emph{F\#}~\cite{aizatulinComputationalVerificationProtocol2012,aizatulinExtractingVerifyingCryptographic2011,DBLP:conf/csfw/ChakiD09,DBLP:conf/vmcai/Goubault-LarrecqP05,DBLP:conf/ccs/BackesMU10,DBLP:conf/ccs/BhargavanFCZ08}.
Despite this, there still exists a large gap between the correctness of cryptographic protocols' models or high-level implementations and their object code (i.e., machine code generated by compilers) that ultimately will execute on hardware.
This gap may result in a lack of security, even
when correctness proofs are developed for a model or a high-level implementation of the protocol. 
A major cause for this gap is the fact that program behavior at the source level
can diverge from its actual behavior when executed on hardware, e.g., due to compiler-introduced bugs~\cite{Xu2023SilentBM}. This is because enabling aggressive compiler optimizations can lead to missing source-level checks like code intended to detect integer overflows~\cite{10.1145/2517349.2522728} or null pointers~\cite{Xu2023SilentBM}. Compilers can invalidate code for secret scrubbing~\cite{secretscrubbing} or even turn constant-time code into a nonconstant-time binary~\cite{DBLP:conf/eurosp/SimonCA18}.
An example is depicted in~\autoref{fig:motivatingeg} where the compiler ({\small GCC ARM V\texttt{11.2.1}}) removes the \texttt{memset} function used to erase the secret key from memory. If such a function is part of a protocol specification, its removal can potentially leak the secret key. Alas, even verifying compilers like CompCert~\cite{compcert} are not a cure-all---verifying all optimization stages that developers want to enjoy in practice is tough.

\begin{figure}[t!]
		\includegraphics[width=1\linewidth]{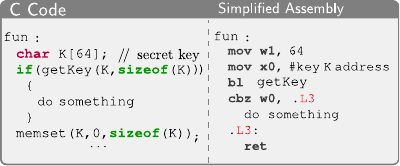}
		\caption{An example of how compilers can invalidate code for secret scrubbing}
		\label{fig:motivatingeg}
\end{figure}

\emph{Protocol verification} obtains privacy and authenticity guarantees for
an abstract model of the protocol that emphasizes concurrency
and communication, typically abstracting cryptographic primitives with a fixed
set of computation rules for the attacker, known as the Dolev-Yao
model~\cite{Dolev1981}. 
On the other hand, \emph{program verification} obtains functional guarantees but also confidentiality in a program model that is typically sequential and
focused on a single machine. The attacker is not bound
by computation rules or runtime restrictions.
Several works have adopted features from one domain in the other---we
discuss them in detail in~\autoref{sec:relatedwork}---but this comes at
the cost of complex, inflexible execution models and high development
cost. We advocate for \emph{composing} these techniques, thereby bridging the gap between tools and verification technologies that can thus continue to evolve
independently. 

Our goal in this paper is to extend the verification of security protocols to their machine code. 
This eliminates the need for trusting compilers and provides a higher assurance about the security and correctness of protocols.
To achieve this, we extend HolBA~\cite{DBLP:journals/scp/LindnerGM19} for the verification of RISC-V and ARM binaries with a method for \emph{symbolic execution} that handles interactions with an arbitrary attacker and trusted cryptographic code. 
Our symbolic execution resolves indirect jumps and supports (compile-time) bounded loops using the summarization technique~\cite{abstractingpc2012}, which we \emph{fully automate}.
We also devise a sound translation from symbolic execution to an intermediate layer that is amenable to automated verification with the protocol verifiers ProVerif~\cite{DBLP:conf/cis/WenFMZWT16,blanchet:protoProlog:2001} and CryptoVerif~\cite{Blanchet2008}.~\autoref{fig:differences} depicts the pipeline of the~\framework{}.

Implementation-level vulnerabilities like the stack-based buffer overflow in Sami FTP Server 2.0.1 are then covered by the symbolic execution, while protocol-level vulnerabilities like the triple-handshake attack on TLS~\cite{bhargavantriple2014} are detected by those protocol verifiers.
This is guaranteed by our soundness results and has been evidenced by the discovery of two flaws in CSur and NSL.
If there is an error in the implementation of the protocol specification, it could either lead to a model that deviates from the intended behavior, or an error during the extraction process.
For instance, if a cryptographic value is erroneously copied within the program, the symbolic execution precisely tracks the memory and detects any subsequent library call operating on the wrongly formatted data. As a result, abstraction would fail, indicating an implementation error; the failure is in the sense that the abstract operations do not apply to incorrectly encoded cryptographic data.
\begin{figure}[t!]
	\includegraphics[width=1\linewidth]{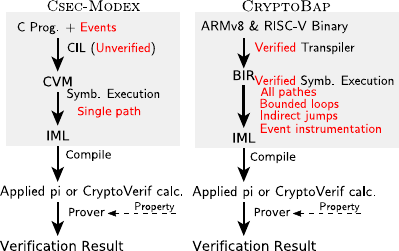}
    \caption{\textsc{CryptoBap} vs. \textsc{Csec-modex}}
	\label{fig:differences}
\end{figure}

Our method is inspired by and builds on a line of work by Aizatulin et al. \cite{aizatulinComputationalVerificationProtocol2012,aizatulinExtractingVerifyingCryptographic2011};~\autoref{fig:differences} highlights differences between the two approaches (see Sec.~\ref{sec:relatedwork} for more details on differences of the two approaches.)
We adopt their process calculus \emph{intermediate model language} ($ \imlsymb $\footnote{We depict $ \birsymb $ in \birc{black, \ bold \ roman}, $ \sbirsymb $ in \sbirc{RedOrange, \ sans \ serif} and $ \imlsymb $ in \imlc{RoyalBlue, \ text \ italic}. Elements common to all languages are typeset in black, italic.}) and translation from $ \imlsymb $ to ProVerif and CryptoVerif.
However, we extract the $\imlsymb$ model from the machine code of protocols, providing more reliable security guarantees that are independent of the compiler used to generate the object code of protocols.
We demonstrate the effectiveness of our approach by covering the case studies of Aizatulin et al.~\cite{aizatulinComputationalVerificationProtocol2012}, including the CSur case study (which they could not verify due to limitations of their framework in handling network messages as \emph{C} structures), as well as an implementation of SSH called TinySSH~\cite{tinyssh}. Moreover, we have verified WireGuard~\cite{wiregaurd},
a modern VPN protocol integrated into Linux, by automatically extracting its ProVerif model.

Our threat model includes a set of functions whose input and output
behavior are controlled by the attacker. This can represent a network
attacker (when these functions are syscalls for network I/O) or
a VM running in parallel (when these functions are hypercalls in a security hypervisor, e.g.,~\cite{prosper:dam2013,sel4}). The execution platform is trusted to implement the
machine-code semantics correctly and includes a set of trusted
cryptographic functions. 
These are likewise assumed to be correct.
Also, when outputting to ProVerif, the attacker is assumed to be
Dolev-Yao, i.e., cryptography is perfect.

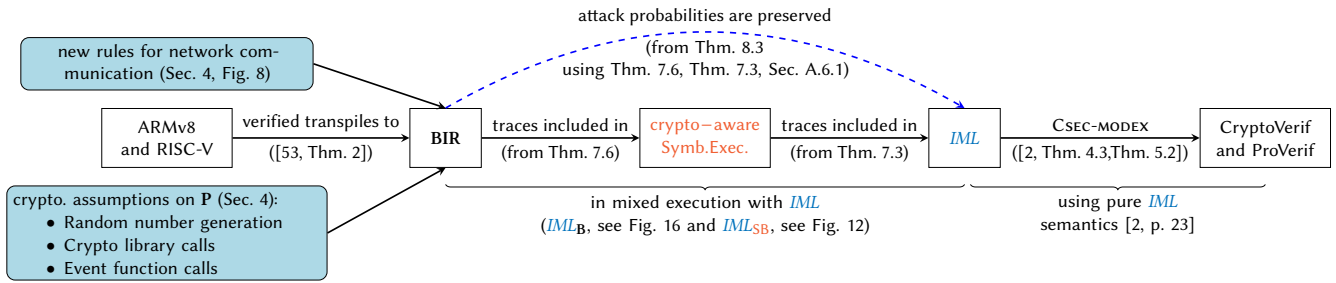
\begin{figure*}[t]
  \centering
  \resizebox{1\textwidth}{!}{
  \begin{tikzpicture}[node distance=3cm,font = {\normalsize\sffamily}]
\node [block] (NetCom) {new rules for
	network communication (\autoref{sec:cryptoawarebir}, \autoref{birnetcom})};
\node (ARM) [process, below = 0.3cm of NetCom] {ARMv8 and RISC-V};
\node (BIR) [process1,  right=3cm of ARM] {\hyperref[ssec:HolBA-framework]{$\birsymb$}};
\node (SBIR) [process, right of=BIR, xshift=1.4cm] {\sbirc{crypto\!-\!aware}\\\hyperref[sec:cryptoawaresbir]{\sbirc{Symb. Exec.}}};
\node (IML) [process1, right of=SBIR, xshift=1.4cm] {\hyperref[ssec:iml]{$ \imlsymb $}};
\node (Csec) [process, right of=IML, xshift=2.1cm] {CryptoVerif and ProVerif};
 \node [block1, below = 0.3cm of ARM] (Cry-assum) {
         \parbox{18em}{
         crypto. assumptions on $\birprog$ (\autoref{sec:cryptoawarebir}):
         \begin{itemize}[noitemsep,nosep]
             \item Random number generation 
            \item Crypto library calls 
            \item Event function calls 
         \end{itemize}
}};

\draw [arrow] (Cry-assum.east) -- (BIR.south);

\draw [arrow] (NetCom.east) -- (BIR.north);

\draw [arrow] (ARM) -- node[above,pos=0.5] {verified transpiles to} node[below,pos=0.5]{(\cite[Thm.~2]{DBLP:journals/scp/LindnerGM19})} (BIR);

\draw [arrow] (BIR) -- node[above] {traces included in} node[below]{(from~\autoref{thm:bir:traceeq})} (SBIR);

\draw [arrow] (SBIR) -- node[above] {traces included in} node[below]{(from~\autoref{thm:iml:traceeq})} (IML);

\draw [arrow] (IML) -- (Csec) node[above,pos=0.5,text width=5cm,align=center] {\textsc{Csec-modex}} node[below,pos=0.5,align=center]{(\cite[Thm.~4.3,Thm.~5.2]{aizatulin2015verifying})};

\draw[->,>=stealth, blue, thick,dashed] (BIR.north) to [bend left] node[midway,above,black]{attack probabilities are preserved} node[midway, below,black,text width=6cm,align=center] {(from~\autoref{thm:soundness}\\using~\autoref{thm:bir:traceeq},~\autoref{thm:iml:traceeq},~\autoref{lem:thm4lems})} (IML.north) ;

\draw[black, decoration={brace,mirror,raise=8pt},decorate] (BIR.south) -- node[below=10pt,text width=10cm,align=center,black] {in mixed execution with $ \imlsymb $ \\ (\imlbir{}, see \autoref{fig:mixedbiriml} and \imlsbir{}, see \autoref{fig:mixedsymbiml})} (IML.south);

\draw[black,decoration={brace,mirror,raise=8pt},decorate ] (IML.south)  ++ (0.1,0) -- node[below=10pt,text width=4cm,align=center,black] {using pure $ \imlsymb $ \\ semantics \cite[p.~23]{aizatulin2015verifying}} (Csec.south);
\end{tikzpicture}
  }
\caption{Organization of the~\framework{} approach.}
\label{fig:overview}
\end{figure*} 

\begin{figure*}[t]
	\centering
	\includegraphics[width=1\linewidth]{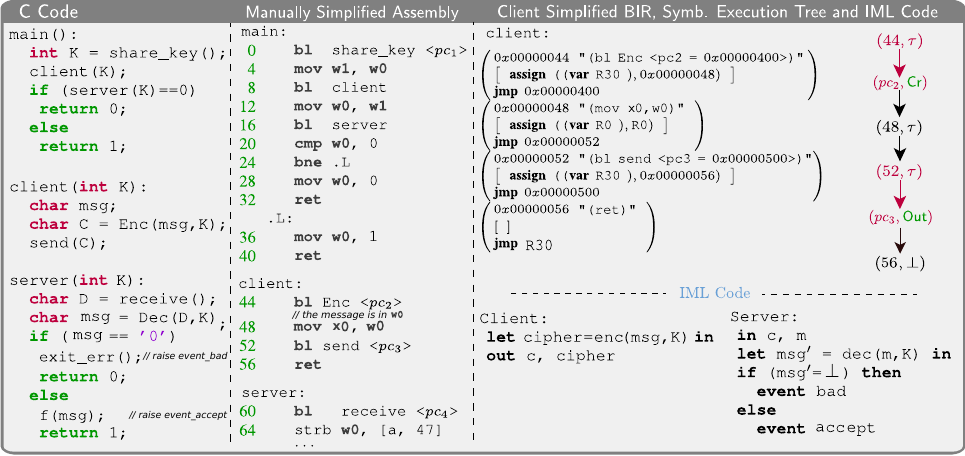}
	\caption{Running example. Function addresses are chosen randomly. Line numbers and $pc_i$ are addresses of instructions and functions in the memory. 
		}
	\label{fig:runningexample}
\end{figure*}

\paragraph{\bf Outline of the \framework{'s} approach.}
As in Fig.~\ref{fig:differences},~\framework{} (source code is available at~\texttt{\small \url{https://github.com/FMSecure/CryptoBAP}}) takes as input:
\begin{itemize} 
\item the protocol participants' binary; 
\item the symbolic model of cryptographic functions;
\item and, the \emph{security property} we verify for event traces of executions, formally defined in~\autoref{sec:secproperties}. 
\end{itemize}

The properties we consider are \emph{safety properties over event traces}; currently, \framework{} supports verification of \emph{authentication} and \emph{weak secrecy}~\cite{bloch2011physical}. 
To verify these properties, we transpile the participants' binary into the $ \birsymb $ representation, the internal language of the HolBA framework (see~\autoref{ssec:HolBA-framework}).  
Our transpiler is formally verified and guarantees to preserve the semantics of the machine code. However, $\birsymb$ as defined in~\cite{DBLP:journals/scp/LindnerGM19} is not suitable for reasoning about the security of cryptographic protocols. We address this by extending $\birsymb$ in \autoref{sec:cryptoawarebir}, e.g., to support network communication, random number generation, etc. 

To export the property checking to off-the-shelf verifiers, we extract a model of protocol participants that these verifiers can process. To automate this, we symbolically execute the $\birsymb$ translation to instrument on-the-fly the $\birsymb$ code with events that flag completion of certain operations and occurrence of errors, and build its execution tree that contains all program execution paths. We translate the resulting execution tree to obtain equivalent programs in $ \imlsymb $ (see~\autoref{ssec:Csec-Modex:iml} and~\autoref{sec:sbirtoiml}). 
The extracted $ \imlsymb $ model is then passed to \textsc{Csec-modex}~\cite{csecmodex,aizatulinExtractingVerifyingCryptographic2011}, which applies algebraic rewriting to convert the model into the input language of ProVerif and CryptoVerif. 

Akin to prior work~\cite{aizatulinComputationalVerificationProtocol2012,aizatulinExtractingVerifyingCryptographic2011,DBLP:conf/csfw/ChakiD09,DBLP:conf/vmcai/Goubault-LarrecqP05,DBLP:journals/entcs/Jurjens09,DBLP:journals/access/HeLCHWM20,DBLP:conf/ccs/BackesMU10,DBLP:journals/toplas/BhargavanFGT08,DBLP:conf/ccs/BhargavanFCZ08}, we trust the cryptographic primitives and abstract them with function symbols that are linked to a Dolev-Yao model when exporting to ProVerif, or to complexity-theoretic assumptions when exporting to CryptoVerif.
Analyzing these primitives' correctness is important, but it requires a different methodology, e.g., weakest precondition propagation. 
Projects like MSR's Project Everest~\cite{bhargavan2017everest} and the resulting HACL* library~\cite{zinzindohoue2017hacl}, CompCert in conjunction with FCF~\cite{beringer2015verified,ye2017verified} or VALE~\cite{bond2017vale} or synthesis approaches like Fiat-Crypto~\cite{erbsen2020simple} address this equally challenging problem.

To prove that the extracted $ \imlsymb $ model preserves the behavior of the actual binary and to relate back the verified properties to the protocol binary, we define a mixed execution semantics. The mixed execution enables protocol participants from different abstraction layers  to run in parallel and communicate (see~\autoref{sec:imlsoundness} and Fig.~\ref{sec:mixedbirtomixedsbir}).  
Such a mixed semantic also frees our symbolic execution from dealing with concurrency.
\autoref{fig:overview} shows the interconnection between different layers in our approach. 
To summarize:
\begin{itemize}
    \item We present \framework{} to automate the verification of cryptographic protocols' binary. Our framework explores all execution paths of protocols, resolves indirect jumps, and handles (compile-time) bounded loops automatically. 
    \item  We extend the vanilla symbolic execution engine in  HolBA to automate the model extraction of cryptographic protocols. The way we extend the engine is significant. While HolBA's vanilla symbolic execution considers the program a holistic entity and thus basically encodes the semantics of the language ($\birsymb$), our semantics regards only a part of the whole program and abstracts cryptographic libraries, attacker calls, and random number generation. 
    \item We formally verify the soundness of our approach and show that verified properties can be transferred back to the binary of analyzed protocols. 
    \item To evaluate \framework{}, we have successfully verified multiple case studies, ranging from toy examples, e.g., CSur, to TinySSH and WireGuard. 
\end{itemize}

\paragraph{\bf Running example.} 
Our running example, \autoref{fig:runningexample}, consists of a client and a server that use a symmetric-key encryption scheme to communicate securely.
This example shows a weak form of authentication, called aliveness~\cite{10.5555/794197.795075}: the server will accept the connection to the (single) client only if it can successfully decrypt the received message using the pre-shared key.
First, the client encrypts a message using the shared key, and sends it to the server. 
Second, the server receives the encrypted message at the other end and decrypts it using the same key. 
Depending on whether the decryption succeeds or fails, either \texttt{\small event\_accept}, to show acceptance of the connection with the client, or \texttt{\small event\_bad} will be released.

\section{Related Work}
\label{sec:relatedwork}

In the last decade, cryptographers started to employ and even develop theorem provers to develop verifiable proofs~\cite{haleviplausible2005}. 
This started with the CertiCrypt framework for Coq~\cite{barthe:formal:2009}, which subsequently developed into EasyCrypt~\cite{barthe:computer:aided:2011}. These tools support reasoning about probabilistic programs and classes of (e.g., poly-time restricted) adversaries via probabilistic and probabilistic relational Hoare logic. 
There are also embeddings of probabilistic reasoning like
FCF~\cite{symmetricEncPetcher:2015}, Verypto~\cite{Verypto} and
CryptHOL~\cite{crypthol}, which are easier to combine with
other techniques (this was
demonstrated~\cite{opensslHmacBeringer,symmetricEncPetcher:2015}
for C code), but they require tedious manual analysis and a deep
understanding of the underlying relational probabilistic logic.

These methods are typically used to verify cryptographic constructions. 
By contrast, complex protocols are
analyzed using symbolic models, where cryptographic
primitives like encryption, signatures, etc.\ are abstracted using
a term algebra and a set of reduction rules.
This makes the analysis of protocols that use such primitives
amenable to automation, e.g., using first-order SAT solving~\cite{DBLP:journals/scp/FraikinFS14},
Horn clause resolution~\cite{blanchet:protoProlog:2001} or constraint
solving~\cite{SMCB-csf12}.
To great success: within a decade, protocol verification tools went
from analyzing small academic protocols~\cite{proto:hornc:Blanchet,mitb} to 
fully-fledged TLS models~\cite{symb:tls:cremers,bhargavan:tls:2017}.

This degree of automation comes at the cost of abstraction, both in
terms of the computation environment and the cryptographic primitives.
The former is owed to the focus on protocol \emph{specifications}
rather than implementations. This makes sense, because often, the task at hand is to evaluate designs
or standards rather than specific implementations, which can be
incomplete or nonexistent.
There are efforts to translate implementations into the protocol
specifications, but they are limited to high-level languages such
as \emph{C}~\cite{aizatulinExtractingVerifyingCryptographic2011}. 
The same holds for verification tools that operate at the source-code
level~\cite{DBLP:journals/jcs/DupressoirGJN14,kusters2012framework}.

\begin{table*}[t]
   \adjustbox{varwidth=\linewidth,scale=1.1}{%
    \begin{tabular}{@{} cl*{5}c @{}}
        & \textbf{Papers} & {\shortstack[c]{Language}}
        & {Abstract Model} & {Model Type} & {Property} 
        & {Soundness}\\
        \cmidrule{2-7}
        &Aiazatulin et al.\cite{aizatulinComputationalVerificationProtocol2012,aizatulinExtractingVerifyingCryptographic2011}   & C & Applied-pi &  Symb. + Comp.  & Secrecy, Auth.  & \OK  \\
        & Chaki et al.\cite{DBLP:conf/csfw/ChakiD09}   & C & ASPIER  &  Symb. & Secrecy, Auth.  & \OK  \\    
        & Goubault-Larrecq et al.\cite{DBLP:conf/vmcai/Goubault-LarrecqP05}   & C & Horn clauses  &  Symb. & Secrecy, Info. flow  & \OK  \\   
        & Backes et al.\cite{DBLP:conf/ccs/BackesMU10}   & F\# & RFC  &  Symb. + Comp. & Safety properties  & \OK  \\      
        & Bhargavan et al.\cite{DBLP:journals/toplas/BhargavanFGT08,DBLP:conf/ccs/BhargavanFCZ08}   & F\# & Pi/CryptoVerif  & Symb. + Comp. &  Secrecy, Auth.  & \NO  \\  
        & J{\"u}rjens\cite{DBLP:journals/entcs/Jurjens09}   & Java & First-order Logic & Symb. &  Secrecy, Auth.  & \NO  \\           
        & HE et al.\cite{DBLP:journals/access/HeLCHWM20}   & Python-JS & Applied-pi & Symb. &  Secrecy, Auth.  & \NO  \\        
        & \textbf{This work}  & Binary & Applied-pi &  Symb. + Comp. & Secrecy, Auth.  & \OK  \\         
        \cmidrule[1pt]{2-7}
    \end{tabular}
   }
    \caption{Selected model extraction approaches; Symb = Symbolic, Comp = Computational, Auth = Authentication, JS = JavaScript.}
    \label{tab:modelextarctiontechns}
    \vspace{-1em}
\end{table*}

\subsection{Verified crypto protocols' implementation}
There have been efforts to verify (annotated) implementations of security protocols using: deductive verification\cite{DBLP:journals/jcs/DupressoirGJN14}, type checking~\cite{10.1145/1706299.1706350}, code generation~\cite{DBLP:conf/IEEEares/CadeB12}, and model extraction.
Existing work in model extraction mainly targets implementations in high-level languages like \emph{C}~\cite{DBLP:conf/csfw/ChakiD09,DBLP:conf/vmcai/Goubault-LarrecqP05}, \emph{F\#}~\cite{DBLP:conf/ccs/BackesMU10,bhargavan:tls:2017,DBLP:conf/wsfm/BhargavanFG06} and \emph{Java}~\cite{DBLP:journals/entcs/Jurjens09,o2008using}. 
However, due to the complexity of such languages, the existing works have to limit their scope. For example,~\cite{DBLP:conf/csfw/ChakiD09} does not model floating pointers and~\cite{DBLP:conf/vmcai/Goubault-LarrecqP05} omits explicit casts and negative array indexes.
\autoref{tab:modelextarctiontechns} compares selected works that verify the security of protocols via model extraction in some form.

There are also works that recover formats of protocol messages at the binary level~\cite{DBLP:journals/cn/CaballeroS13,DBLP:conf/ndss/LinJXZ08,DBLP:conf/ndss/WondracekCKK08,DBLP:journals/jifs/XiaoZL16}. However, their intent is different from ours. Existing approaches are mostly applied to malware binaries and use heuristics to gain insights into their operation. Thus, they are not concerned with soundness, and the inferred message formats can be wrong.

\paragraph{\bf Comparison to Aizatulin et al. approach}
Closely related to~\framework{} is Aizatulin's work~\cite{aizatulinExtractingVerifyingCryptographic2011,aizatulinComputationalVerificationProtocol2012}. 
Analogous to our work, they also proved the soundness of their approach, i.e., showed that all attacks present in the \emph{C} code are preserved in the extracted $\imlsymb$ models. 
The main difference between the two approaches
is the fact that we target protocols' binary. Moreover, compared to
their approach, which can handle only a single execution path, \framework{} handles \textbf{all execution paths} of protocols, including those that contain \textbf{conditionals}, \textbf{bounded loops}, and \textbf{indirect jumps}. 
More specifically, Aizatulin's approach restricts the input programs to programs that have no ``else'' branches and no loops.
Conceptually, one may see this as irrelevant because one might speculate that all paths outside the ``main'' part are not useful, i.e., they do not produce network output, or at least none that reveals cryptographic information (e.g., error codes). But this is still restrictive. First, many protocols are specified to produce network output in error cases, for example, \emph{decoy messages} in anonymity protocols or error messages that are encrypted. Moreover, seemingly regular protocols have ``else'' branches as part of their ``main'' message flow, once we look at them with the level of detail necessary to analyze implementations. For example, cypher suite negotiation in TLS must be formulated with multiple branches depending on the input message.

Also, their translation from \emph{C} to \emph{C virtual machine} (CVM)---the intermediate language used for extracting the $\imlsymb$ model of protocols---is not verified. This renders relating the verified properties to the \emph{C} implementation infeasible.
Similarly, they abstracted cryptographic libraries, attacker calls, and random number generation in their verification. However,
they do not formulate the requirements on the program under the analysis explicitly. CVM already assumed to include primitives for library calls, attacker calls, and random number generation---which are not primitives of the \emph{C} language. In this regard, CVM is fairly close to $\imlsymb$. Moreover, our translation into $\imlsymb$ is entirely different (a) because the input language has a more complex state and interaction with the other $\imlsymb$ entities and (b) because in contrast to Aizatulin's approach, we handle conditionals, hence our symbolic execution is not only a sequence of actions.
Finally, while they had to change the source code of protocols and add dummy functions to flag the occurrence of events, \framework{} automates releasing events during symbolic execution and minimizes user involvement.

\section{Background}
We next explain the preliminaries before presenting the details of our approach.

\subsection{HolBA framework \& Vanilla symbolic exec.}\label{ssec:HolBA-framework}
\framework{} relies on HolBA~\cite{DBLP:journals/scp/LindnerGM19} to transpile the binary of protocols to the $\birsymb$ representation.  $\birsymb$ is a simple and architecture-agnostic language used as the internal language of HolBA and is designed to simplify the binary analysis of programs and facilitate building analysis tools.
HolBA is proof-producing and ensures that the transpilation preserves the semantics of the binary.
\begin{figure}[t]
	\begin{minipage}[b]{0.45\linewidth}
		\input{gfx/syntax-BIR.tex}%
	\end{minipage}
\end{figure}
A $\birsymb$ program $\birprog$ includes a number of blocks, see~\autoref{fig:birSyntax}, each consisting of a tuple of a unique label---a string or an integer---and a few statements. 
Each label refers to a particular location in the  program and is often used as the target of jump instructions (i.e., \birc{jmp} or \birc{cjmp}).
$\birsymb$ expressions include constants, standard binary and unary operators (ranged over by $\birc{\Diamond_b}$ and $\birc{\Diamond_u}$) for finite integer arithmetic, memory operations, and conditionals.
Fig.~\ref{fig:runningexample} presents a $\birsymb$ snippet for the running example.

A $\birsymb$ state $\tuple{\birenv, \birpc} \in \birstates{}$ consists of an environment $\birenv: \birvar \mapsto \birval$
which maps variables, i.e., registers $\reg{i}$ and memory locations $\memsymb$, to values and a program counter $\birpc$ that holds the label of the executing $\birsymb$ block.
The relation $\birtrans{}{}{}{} \subseteq \birstates \times \birstates$ models the execution of a $\birsymb$ block.
The execution of ${n}$ steps is denoted by
$\birtransstar{}{}{}{}$ if $n\ge 0$
and
$\birtransmulti{}{}{}{}$ 
or
$\birtransmultin{}{}{}{}$,
if $n > 0$.

We build \framework{} on a proof-producing symbolic execution for $ \birsymb $~\cite{proofproducingsymbexecution} that
formalizes the symbolic generalization of $\birsymb$ (hereafter $ \sbirsymb $).
The symbolic semantics is bisimilar to the concrete one and
allows guiding the execution while maintaining a sound set of reachable states from an initial symbolic state (we call this the \emph{symbolic execution structure}).
To generalize from $\birsymb$ to $ \sbirsymb $, symbolic expressions $\sbirvals$ are defined that can be interpreted to $\birsymb$ values $\birval$ via an interpretation $\sbirvaluation : \sbirvals \to \birval$.
In addition to the symbolic environment $\sbirenv: \birvar \mapsto \sbirvals$, the $\sbirsymb$ state
$\tuple{\sbirpcond, \sbirenv, \birpc} \in \sbirstates{}$
also contains a path condition $\sbirpcond \in \sbirvals$ and a $\birpc$ that is kept concrete to obtain a concrete control flow.

Let $\sbirtrans{}{}{}{}: \sbirstates \times \sbirstates$ be the single-step transition relation of $\sbirsymb$ and $\sbirtransmultin{}{}{}{}$ (or  $\sbirtransmulti{}{}{}{}$) denote a multi-step symbolic transition.
We write $\transitionsf{\birstate{i}}{n}{\lbls}{\birstate{j}}$ to restrict the transition from $\birstate{i}$ to $\birstate{j}$ to the label set $\var{L}$. 
For the HolBA's vanilla symbolic execution, Lindner et al.~\cite{proofproducingsymbexecution} proved that a single $\sbirsymb$ execution step soundly matches a single $\birsymb$ execution step, characterized by the following simulation theorem:

\begin{property}\label{single-step-vanila}
For all $\birstate{i}, \sbirvaluation, \birstate{j}$, $\sbirstate{i}$ s.t. $\birsimrel{\birstate{i}}{\sbirvaluation}{\sbirstate{i}}$, if $\transition{\birstate{i}}{\birstate{j}}$ then there exist an $\sbirvaluation'$ and $\sbirstate{j} $ s.t. $\sbirvaluation\! \subseteq\! \sbirvaluation'$ and $\sbirtrans{}{}{\sbirstate{i}}{\sbirstate{j}}$ and $\birsimrel{\birstate{j}}{\sbirvaluation'}{\sbirstate{j}}$.
\end{property}

The simulation relation $\birsimrel{}{\sbirvaluation}{}$ asserts the consistency of corresponding $\birsymb$ and $\sbirsymb$ states, i.e., their program counters are equal, their environments are equal through the interpretation $\sbirvaluation$, and the evaluation of $\sbirpcond$ under $\sbirvaluation$ results in $\true$.
Then the soundness of the symbolic execution structure for multiple steps corresponds to the extension of \autoref{single-step-vanila} to a multi-step simulation theorem. 

\subsection{\textsc{Csec-modex} toolchain \& \imlsymb}\label{ssec:Csec-Modex:iml}
Aizatulin et al.~\cite{aizatulinExtractingVerifyingCryptographic2011} proposed an automated technique to verify the security of cryptographic protocols' \emph{C} implementation.
At a high level, \textsc{Csec-modex} takes as input the \emph{C} code of protocol participants together with a template file for the verifier (ProVerif or CryptoVerif). 
The toolchain extracts the $ \imlsymb $ model of the protocol, which is then converted into the verifier's input language. The template encodes assumptions about cryptographic primitives in the implementation,  the environment process which spawns the participants and generates shared cryptographic material, and a query for the security property that is checked for the implementation.

The intermediate model language, $ \imlsymb $, is a version of the applied-pi calculus extended with bitstring manipulation primitives. 
\begin{figure}[t]
	\begin{minipage}[b]{0.60\linewidth}
		\input{gfx/syntax-iml.tex}
	\end{minipage}
\end{figure}
In Fig.~\ref{fig:imlSyntax}, $\imlvals = \{0,1\}^*$ is the set of finite bitstrings, $\imlc{Ops}$ 
is the set of operations, including
cryptographic primitives, and $\imlfunc(\imlvar{e_1},\dots,\imlvar{e_m})$ denotes function application.
$ \imlsymb $ expressions are evaluated with respect to an environment
$\imlenv : \imlvarspace \mapsto \imlvals \cup \{\imlc{\bot}\}$
which maps variables to bitstrings or $ \imlc{\bot} $. 

$\imlvar{P}$ and $\imlvar{Q}$ represent input/output processes. 
An executing process is the basic unit of execution in $ \imlsymb $ and has the form
$ (\imlenv, \imlprocessp)$, 
where
$\imlprocessp$ is either an input or output process.
The input process $\imlc{0}$ does nothing.
In $ \imlsymb $, inputs and outputs are performed using $\imlc{in}$ and $\imlc{out}$ in which $\imlchannel$ denotes the channel name and $ \imlvar{e_1},\dots,\imlvar{e_m} $ indicate the protocol participants' identifier.
The construct $\imlc{new}\ \imlvar{x}:\type$ generates a uniform random number of 
type $\type$ and $\imlc{event}(\imlvar{d_1},\dots,\imlvar{d_m})$ is used to raise an event during the execution.

An $ \imlsymb $ state $ (\imlenv, \imlprocessp), \imlmultisetprocs \in \imlstates$ includes an output process $\imlprocessp$ and a multiset of executing input processes $\imlmultisetprocs$. 
The initial configuration of an input process $ \imlprocessq $ is defined as 
$(\emptyset, \imlc{out}(\imlchannel , \imlc{\varepsilon}); \imlc{0}), \mathit{reduce}{(\emptyset, \imlprocessq)} $
where $\mathit{reduce}$ represents a function that executes a sequence of processes inside $\imlprocessq$ (e.g., $\imlprocessq = \imlprocessp_{\imlc{1}};\imlprocessp_{\imlc{2}}; ...$) until an input process waiting for a message from channel $ \imlchannel $ is reached (see 2nd rule in~\autoref{fig:imltranssemantics}). 

We use $\imltrans{\imlevent}{p}{}{} \subseteq \imlstates \times \imlstates$ to denote the $ \imlsymb $ transition relation with the probability $\imlvar{p}$ and the event $\imlevent$. The event $\imlevent$ may be empty or include a single event of the form $ \imlvar{ev}(\imlvar{b_1}, ..., \imlvar{b_m}) $, where $\imlvar{ev}$ is an event symbol and $\imlvar{b_1}, ..., \imlvar{b_m}$ are bitstrings. 
Moreover, an $ \imlsymb $ trace is defined as $\imlexecution = \imltrans{\imlevent_\imlc1}{p_1}{\imlstate{1}}{} \dotsb \imltrans{\imlevent_{\imlc{n-1}}}{p_{n-1}}{}{\imlstate{n}}  \subseteq \imlexecutions(\imlmultisetprocs) $. 

We have borrowed  the semantics of the $ \imlsymb $ transition relations from~\cite[p.~23]{aizatulin2015verifying}. A few representative transition rules of $ \imlsymb $ for random number generation and sending a message on the channel $\imlchannel$ are presented in~\autoref{fig:imltranssemantics} and the rest explained in~\autoref{app:imlsemantics}.
In this figure $\mathit{truncate}$ cuts messages according to the provided length and $\mathit{maxlen}$ is the maximum size of the channel.
We extend the random number generation rule with an event $\imlc{fr}$ which represents the creation of a fresh bitstring $\imlvar{b}$.
This simplifies stating our invariants but is operationally the same.

\begin{figure}[t]
\hspace{-1cm}\adjustbox{varwidth=\linewidth,scale=0.92}{%
\[
\begin{array}{c}
\begin{prooftree}
		\hypo{\type = \texttt{fixed}_n\ \text{for some}\ n \in \naturalnum}
		\hypo{|\imlvar{b}| = n}
        \infer2[]{ \imltrans{\imlfreshev{\imlvar{b}}}{\probdist{n}}{(\imlenv,\! \imlc{new} \ \imlvar{x}\!:\!\type;\!\imlprocessp), \imlmultisetprocs}{\!\!(\imlenv[\imlvar{x} \mapsto \imlvar{b}],\! \imlprocessp), \imlmultisetprocs}}
\end{prooftree}\\[20pt]
\begin{prooftree}
		\hypo{ \begin{matrix}
		   \imleval{\imlvar{e}}_{\imlenv} = \imlvar{b} \neq \imlc{\bot}  \quad	\imlvar{b}' = \mathit{truncate}(\imlvar{b}, \mathit{maxlen}(\imlchannel)) \\[5pt]
		    \forall j \leq m : \imleval{\imlvar{e_j}}_{\imlenv} = \imlvar{b_j} \neq \imlc{\bot}  \quad \imlmultisetprocs' =  \mathit{reduce}(\{ (\imlenv, \imlprocessq) \}) \\[5pt]
		   \exists ! (\nextimlenv, \imlprocessq')\!\in\!\imlmultisetprocs	 : \imlprocessq' = \imlc{in}(\imlchannel[\imlvar{e{\rmcolor{'}}_1},\dots,\imlvar{e{\rmcolor{'}}_m}], \imlvar{x}');\imlprocessp' \! \land \! \forall j \leq m : \imleval{\imlvar{e_j}'}_{\nextimlenv} = \imlvar{b_j} \neq \imlc{\bot}
		\end{matrix}}			
          \infer1[]{ \imltrans{}{1}{(\imlenv, \imlc{out}(\imlchannel[\imlvar{e_1},\dots,\imlvar{e_m}], \imlvar{e});\imlprocessq), \imlmultisetprocs}{(\nextimlenv[\imlvar{x}' \mapsto \imlvar{b}'], \imlprocessp'), \imlmultisetprocs \uplus \imlmultisetprocs' \setminus \{ (\nextimlenv, \imlprocessq') \} }}
\end{prooftree}
\end{array}
\]
}
	\caption{The semantics of $ \imlsymb $~\cite[p.~23]{aizatulin2015verifying} transition relation.}
	\label{fig:imltranssemantics}
\end{figure}

\section{BIR with cryptography}
\label{sec:cryptoawarebir}
\begin{figure*}
	\centering
		\hspace{-1.5cm}\adjustbox{varwidth=\linewidth,scale=0.9}{%
		\centering
		\begin{prooftree}
			\hypo{ \begin{matrix}
					\birpc \in \entr_{\birc{\attacker_s}}
			\end{matrix}}
			\hypo{ \begin{matrix}
					\birpc' = \pcinc{\birenv, \birpc}
			\end{matrix}}
			\hypo{\begin{matrix}
					\var{a} = \birenv[\reg{0}]
			\end{matrix}}
			\hypo{\begin{matrix}
					\imlvar{e} = \memload(\birenv,\var{a})
			\end{matrix}}
			
			\infer4[$ \birc{\attacker_{s}} $]{\birprog \vdash \birtrans{\Output{\imlvar{e},(\var{e}_1,\dots,\var{e}_m)}}{}{\tuple{\birenv, \birpc},\imlchannel[\var{e}_1,\dots,\var{e}_m]}{\tuple{\birenv , \birpc'},\imlvar{e}\bnfconcat \imlchannel[\var{e}_1,\dots,\var{e}_m]} }
		\end{prooftree}
		\hspace{3mm}
		\begin{prooftree}
			\hypo{ \begin{matrix}
					\birpc \in \entr_{\birc{\attacker_r}}
			\end{matrix}}
			\hypo{ \begin{matrix}
					\birpc' = \pcinc{\birenv, \birpc}
			\end{matrix}}
			\hypo{\begin{matrix}
					(\birenv{'},\var{a}) = \update(\birenv, \heap_{\mathcal{A}}, \advmem, \imlvar{e}', 128)
			\end{matrix}}
			\infer3[$ \birc{\attacker_{r}} $]{\birprog \vdash \birtrans{\Input{\imlvar{e}',(\var{e}_1,\dots,\var{e}_m)}}{}{\tuple{\birenv, \birpc}, \imlvar{e}'\bnfconcat \imlchannel[\var{e}_1,\dots,\var{e}_m]}{\tuple{\birenv{'}[\reg{0} \mapsto \var{a}] , \birpc'},\imlchannel[\var{e}_1,\dots,\var{e}_m]} }
		\end{prooftree}
	}
	\caption{The semantics of $\birsymb$ network communication. $\birsymb$ uses $\imlsymb$ channels $\imlchannel$ for communication.} 
	\label{birnetcom}
\end{figure*}

$ \birsymb $, as described in~\cite{DBLP:journals/scp/LindnerGM19}, does not support ingredients required to reason about the security of cryptographic protocols.
To resolve these issues, we model random number generation and abstract network communications and formulate assumptions on state transformation in certain function calls on top of the existing $\birsymb$ semantics. Such an extension preserves the verified properties of $\birsymb$ and, thus, the soundness of binary transpilation.

Using the information in the (unstripped) binaries' header and preprocessing of lifted programs, we split the address space of $\birsymb$ into five label sets: 
$ \lbls = \lbls_{\birc{\mathcal{N}}} \uplus \lbls_{\cryptoOp} \uplus \lbls_{\birc{\mathcal{A}}} \uplus \lbls_{\birc{\mathcal{R}}} \uplus \lbls_{\birc{\mathcal{E}}}$.  The sets $\lbls_{\cryptoOp} = \bigcup_{\cryptoop\in \cryptoOp}$, $\lbls_{\birc{\mathcal{A}}}$, $\lbls_{\birc{\mathcal{R}}}$, $\lbls_{\birc{\mathcal{E}}}$
correspond to the \emph{cryptographic libraries}, \emph{attacker} calls, 
\emph{random number generation}, and \emph{event} functions. 
Addresses outside these label sets are classified as
normal execution points in $\lbls_\birc{\mathcal{N}}$. Moreover, a specific label set $\lbls_{\birc{\mathcal{L}}} \subset \lbls_\birc{\mathcal{N}}$ defines loop entry points.
For each label set, we axiomatize the expected behavior of the $\birsymb$ program by defining a number of assumptions. We also define a specific entry point for each function---denoted by $\entr_{\birc{\ell}}$ for $\birc{\ell}\in\{\lbls_{\birc{\mathcal{N}}}, \lbls_{\cryptoOp}, \lbls_{\birc{\mathcal{A}}}, \lbls_{\birc{\mathcal{R}}}, \lbls_{\birc{\mathcal{E}}}\}$---and ensure that function calls are done only through the specified entry points.

In the crypto-aware symbolic execution, these function calls will be
treated as atomic operations. We thus introduce some notation to
indicate with an event whenever a \emph{sequence} of steps passes via these special functions.

We extend the $\birsymb$ transition relation with events, $\birtrans{\var{a}}{}{}{} \subseteq \birstates \times \birevents \times \birstates$, where $\birevents$ is the set of observable events
plus the silent transition $\tau$.
Then, a multi-step $\birsymb$ transition  
$\birtransmulti{(\birevent_{\birc{1}},\dots, \birevent_{\birc{m}})}{}{\birstate{0}}{\birstate{n}}$ 
exists if 
$\birtransstar{\birevent_{\birc{1}}}{}{\birstate{0}}{\birstate{1}}
\ldots
\birtransstar{\birevent_{\birc{m}}}{}{\birstate{n-1}}{\birstate{n}}$ 
where
$\birtransstar{\birevent}{}{\birstate{i}}{\birstate{j}}$
if
$\birtransstar{\tau}{}{\birstate{i}}{} \birtrans{\birevent}{}{}{} \birtransstar{\tau}{}{}{\birstate{j}}$ is reminiscent to the big-step semantics. 
Also, we use $\birexecutions(\birprog, \birstate{0})$ to denote the set of execution traces of $\birprog$ starting from the initial state $\birstate{0}$.

We define $\pcincsym: \birstates \to \lbls$ to obtain the next execution point immediately reachable after returning from a call. 
$\update :(\birvar \mapsto \birval) \times \birvar \times 2^{\birvar} \times \imlvals \times \naturalnum \to (\birvar \mapsto \birval) \times \birval$ stores bitstrings into the memory in the $\birsymb$ environment.
Given $\heap \in \memsymb$ and $l \in \{1,8,16,32,64, 128\}$ and
$|\imlvar{b}|$ a multiple of $l$ that can be encoded in 128 bits, 
$\update(\birenv, \heap, \memsymb, \imlvar{b}, {l})$ 
stores $\imlvar{b}$ in ${l}$-bit chunks,
preceded by ${l}$ (encoded as a word)
in the memory $\memsymb\subseteq \birvar$, starting from the
pointer stored in $\heap$.
$\update$ returns a new environment and the address of the data within it, as indicated by $\heap$ in the \emph{previous} environment $\birenv$.
We also introduce notation for reading this bitstring. Let $\|$ be the byte concatenation operator. Then, $\memload(\birenv,\var{a})\defeq{} \left\|_{_{\var{i} = 1 \dots \birenv(\var{a})}}  \birenv(\var{a}+i) \right.$ for the given $\birenv$ and the address `$\var{a}$'.

\framework{} supports 
\emph{call-by-value},
\emph{call-by-reference},
and data passing via global variables (which TinySSH uses) call conventions.
For brevity, we focus on the call-by-value convention; the others follow a similar
pattern.

\paragraph{\bf Random number generation (RNG)}\label{rng-fun} 
$\birsymb$ is a deterministic language; as a result, we are unable to draw cryptographic keys without an external source of randomness that the attacker cannot predict. 
Thus, we allocate a memory region $\randommem$ in the initial state for storing  $k \in \naturalnum$ random values of size $l \in \naturalnum$, for the security parameter  $\secparam=lw$ that is a multiple of some supported word length $w$.
We assume $\randommem$ is an ordered list of $kl$ consecutive addresses. 
To track the number of words read from $\randommem$, we define a counter $\randommemidx$ and store it in the environment. 
Given an initial state 
$\birstate{0}$ and 
a random tape $\randomvals_{\imlc{k}} \in \{0,1\}^{kl\times w}$
the state
$\birstate{0}.\birenv[\randommem \mapsto \randomvals_{\imlc{k}} , \randommemidx \mapsto 0]$ is an instance of this initial state.
To extract a random number of size $l$ from $\randommem$, we define 
$\birrand: (\birvar \to \birval) \times \naturalnum \to \imlvals$
which returns a value from $\randommem$ yet unread:
$
\birrand (\birenv{}, n) \defeq\;
\text{let } \var{x} = \birenv{}[\randommemidx]\;\; \text{in}\;\;
\left\|_{_{i = 0 \dots l-1}}  \birenv{}[\randommem + \var{x} + i] \right.
$.
This construction is reminiscent of probabilistic Turing machines, only that the random number generator is finite due to the finite-memory restriction of $\birsymb$'s memory.

    We call a function from $\lbls_{\birc{\mathcal{R}}}$ with $ \var{l}_{\birc{\mathcal{R}}} \in  \lbls_{\birc{\mathcal{R}}}$ one of its entry points an RNG function, if for any 
    entering state
    $(\birenv_{\birc{0}}, \var{l}_{\birc{\mathcal{R}}})$
    for which $\birrand(\birenv_{\birc{0}}, \var{n})$ is defined,
    and
    execution point 
    $(\birenv{''}, \var{l})$
    after returning from RNG, i.e., $ \var{l} = \pcinc{\birenv_{\birc{0}}, \var{l}_{\birc{\mathcal{R}}}}$,
    the output register holds the address of a copy of the random value
        and $\randommemidx$ is updated, i.e.,
        $\birenv{''} = \birenv{'}[\reg{0} \mapsto \var{a}; \randommemidx \mapsto \birenv_{\birc{0}}[\randommemidx] + {l}]$
        with
        $(\birenv{'},\var{a}) = \update(\birenv_{\birc{0}}, \heap, \memsymb, \imlvar{x_i},128)$ s.t. $ \imlvar{x_i} = \birrand(\birenv_{\birc{0}}, n) $.
We denote RNG steps as $\birtransstar{\freshv{\imlvar{x_i}}}{\lbls_{\birc{\mathcal{R}}}} {(\birenv_{\birc{0}}, \var{l}_{\birc{\mathcal{R}}})} {(\birenv{''}, \var{l})}$, with $ i = \left\lfloor \frac{\birenv_{\birc{0}}[\randommemidx]}{{l}}  \right\rfloor + 1$ being a counter for  the number of times the RNG function was called.

\paragraph{\bf Network communication.}
Protocols relate events on different participants. Therefore, a setting where multiple parties run in parallel is essential to analyze protocols' correctness.
\autoref{sec:soundness} introduces a mixed execution, in which $\birsymb$ programs run in parallel with $ \imlsymb $ processes.
The latter model protocol participants for which we do not have a $\birsymb$ implementation, but also the adversary.

Our $\birsymb$ programs rely on an $ \imlsymb $ channel for communication that has the form $\imlchannel[\var{e_1},\dots,\var{e_m}]$, where $ \var{e_1},\dots,\var{e_m} $ are expressions which identify communicating parties and their channel $\imlchannel{}$. 
To send a message $\imlvar{e}$ ($\entr_{\birc{\attacker_s}} $ in Fig.~\ref{birnetcom}), we fetch the value of $\imlvar{e}$ from the memory address $\birenv{}[\reg{0}]$, put it on the channel $\imlvar{e}\bnfconcat \imlchannel[\var{e_1},\dots,\var{e_m}] $, and release the $\Output{\imlvar{e}, (\var{e_1},\dots,\var{e_m})}$ event. 

To receive a message $\imlvar{e}'\bnfconcat \imlchannel[\var{e_1},\dots,\var{e_m}] $, represented by 
$\entr_{\birc{\attacker_r}}$ in Fig.~\ref{birnetcom},
we store it in a buffer that is only accessible to libraries and return (via register $\reg{0}$) the address, i.e., `$\var{a}$', of the memory region where the message $\imlvar{e}'$ is stored. 
Passing the address via $\reg{0}$ is just one way to model the send and receive functions that also accommodates passing the buffer address by reference. \framework{} models these functions according to the implementation.

\paragraph{\bf Crypto library.}\label{lib-fun}
We establish a set of concrete assumptions on the way
crypto libraries operate. That is, a crypto-library call, like $\cryptoop$, computes
the correct result, never invokes another function, and only changes its own memory, i.e., $\memsymb_{\cryptoOp}$. We denote library steps with 
    $\birtransstar{\crypto{\imlvar{v}}}{\lbls_{\birc{\cryptoOp}}}
    {(\birenv_{\birc{0}}, \var{l}_{\cryptoop})}
    {(\birenv{''}, \var{l})}$, and expect \emph{transitions using labels outside $\lbls_{\cryptoOp}$ do not change the memory of library calls}.
    
    We call
    $\lbls_{\cryptoop} \subseteq \bigcup_{\cryptoop\in \cryptoOp} $
    the library implementation of $\cryptoop$ (with arity $m$)
    and
    $\var{l}_{\cryptoop}\in \entr_{\cryptoOp}$
    one of its entry points, if 
    for any entering state
    $(\birenv_{\birc{0}}, \var{l}_{\cryptoop})$ 
    and
    the return state
    $(\birenv{''}, \var{l})$,
the 
    function result 
    $\imlvar{v} = \cryptoop(\imlvar{b_1},\ldots,\imlvar{b_{m}})$
    for $\imlvar{b_i} = \memload(\birenv,\reg{i})$, $i\in\{1,\ldots,m\}$,
    is stored in a heap and its address is put 
    into $\reg{0}$:
    $ \birenv{''} = \birenv{'}[\reg{0} \mapsto \var{a}] $ where
    $(\birenv{'},\var{a})=\update(\birenv_{\birc{0}}, \heap_{\cryptoOp}, \memsymb_{\cryptoOp}, \imlvar{v}, 128)$.

\paragraph{\bf Event functions.}\label{ev-fun}
Event functions identify specific steps in our program that
we want to argue about. 
For example, when a protocol ends with the establishment of a key, that key is used to transmit some data. 
We want to show that, whenever this step is
reached, it is authenticated, i.e., the purported communication
partner has requested the execution of this step (e.g., \texttt{\small f(msg)} in~\autoref{fig:runningexample}). What happens
in this step is not important for us, only that it is reached.
We hence assume, for simplicity, that such functionality is
replaced by stand-ins we call event functions. These only raise
a visible event, but do not alter the memory.
We denote the transition corresponding to an event function call with
        $\birtransstar{\event{\imlvar{b_1},\ldots,\imlvar{b_{m}}}}{\lbls_{\birc{\mathcal{E}}}}
        {(\birenv_{\birc{0}},\! \var{l}_{\birc{\mathcal{E}}})\!}
        {\!(\birenv_{\birc{0}},\! \var{l})}$
        where
        $\var{l}_{\birc{\mathcal{E}}}\in \lbls_{\birc{\mathcal{E}}}$ is the entry point,
        $\imlvar{b_i}\!=\!\memload(\birenv,\reg{i})$ for $i\in\{1,\ldots,m\}$ are event parameters, and $\var{l} \in \lbls_{\birc{\mathcal{N}}}$.

\begin{figure*}[t]
\adjustbox{varwidth=\linewidth,scale=0.9}{%
\[
\begin{array}{c}
\begin{prooftree}
	\hypo{\birpc \in \entr_{\birc{\attacker_s}}}
	\hypo{\birpc' = \pcinc{\sbirpcond, \sbirenv, \birpc}}
	\hypo{\var{a} = \sbirenv[\reg{0}]}
	\hypo{\symvar{e} = \memload(\sbirenv,\var{a})}
	\infer4[$ \attacker_{s} $]{ \birprog \vdash\sbirtrans{\symoutput{\symvar{e},(\symvar{e_1},\dots,\symvar{e_m})}}{}{\tuple{\sbirpcond, \sbirenv, \birpc},\imlchannel[\symvar{e_1},\dots,\symvar{e_m}]}{\tuple{\sbirpcond, {\sbirenv}, \birpc'},\symvar{e}\bnfconcat \imlchannel[\symvar{e_1},\dots,\symvar{e_m}]} }
\end{prooftree}
\hspace{4mm}
\begin{prooftree}
	\hypo{\birpc \in \entr_{\birc{\mathcal{E}}}}
	\hypo{\birpc' = \pcinc{\sbirpcond, \sbirenv, \birpc}}
	\hypo{(\symvar{d_1},\dots,\symvar{d_m}) \notin \image{\sbirenv}}
	\infer3[event]{ \birprog\vdash \sbirtrans{\symevent{\symvar{d_1},\dots,\symvar{d_m}}}{}{\tuple{\sbirpcond, \sbirenv, \birpc}}{\tuple{\sbirpcond, \sbirenv, \birpc'}} }
\end{prooftree}
 \vspace{2mm}
 \\[15pt] 
 \begin{prooftree}
	\hypo{\birpc \in \entr_{\birc{\mathcal{L}}}}
	\hypo{\birpc' = \pcexit{\birpc}}
	\hypo{\symvar{t} \notin \image{\sbirenv}}
	\hypo{({\sbirpcond}', {\sbirenv}')	= \loopprocess{\sbirpcond, \sbirenv, \birpc}}
	\infer4[loop]{ \birprog\vdash \sbirtransmulti{\symloop{\symvar{t}}}{}{\tuple{\sbirpcond, \sbirenv, \birpc}}{\tuple{{\sbirpcond}', {\sbirenv}', \birpc'}} }
\end{prooftree}
 \\[15pt] 
 \vspace{3mm}
\begin{prooftree}	
	\hypo{ \begin{matrix}
		\birpc \in \entr_{\birc{\mathcal{R}}}\\[5pt] 
			\birpc' = \pcinc{\sbirpcond, \sbirenv, \birpc}
		\end{matrix}}
		\hypo{ \begin{matrix}
		i = \left\lfloor \frac{\sbirenv[\randommemidx]}{l}\right\rfloor + 1 \\[5pt] 
			\symvar{x_i} \notin \image{\sbirenv}
		\end{matrix}}
	\hypo{ \begin{matrix}
			\symvar{x_i} = \sbirrand(\sbirenv, {n})\\
			(\sbirenv{'},\var{a}) = \update(\sbirenv, \heap, \memsymb, \symvar{x_i} ,128)
		\end{matrix}}
			\hypo{\sbirenv{''} = {\sbirenv{'}}[\reg{0} \mapsto \var{a}; \randommemidx \mapsto \sbirenv[\randommemidx] + {l}]}       
	\infer4[RNG(${n}$)]{
		\birprog\vdash \sbirtrans{\symfreshv{\symvar{x_i}}}{}{\tuple{\sbirpcond, \sbirenv, \birpc}}{\tuple{\sbirpcond, 	\sbirenv{''}, \birpc'}}
	}
\end{prooftree}
 \\[15pt] 
  \vspace{2mm}
\begin{prooftree}
	\hypo{\birpc \in \entr_{\attacker_r}}
	\hypo{\birpc' = \pcinc{\sbirpcond, \sbirenv, \birpc}}
	\hypo{\symvar{e} \notin \image{\sbirenv}}
\hypo{(\sbirenv{'},\var{a}) = \update(\sbirenv, \heap_{\birc{\mathcal{A}}}, \advmem,\symvar{e}, 128)}
	\infer4[$ \attacker_{r} $]{ \birprog \vdash\sbirtrans{\syminput{\symvar{e},(\symvar{e_1},\dots,\symvar{e_m})}}{}{\tuple{\sbirpcond, \sbirenv, \birpc}, \symvar{e}\bnfconcat\imlchannel[\symvar{e_1},\dots,\symvar{e_m}]}{\tuple{\sbirpcond, \sbirenv{'}[\reg{0} \mapsto \var{a}], \birpc'}, \imlchannel[\symvar{e_1},\dots,\symvar{e_m}]} }
\end{prooftree}    
 \\[15pt] 
\begin{prooftree}
		\hypo{\birpc \in \entr_{\birc{Op}}}
		\hypo{\birpc' = \pcinc{\sbirpcond, \sbirenv, \birpc}}
		\hypo{\symvar{v} \notin \image{\sbirenv}}
		\hypo{\symvar{v} = \oracle{\birpc, \sbirenv}}
	\hypo{(\sbirenv{'},\var{a})=\update(\sbirenv, \heap_{\birc{Op}}, \memsymb_{\birc{Op}}, \symvar{v}, 128)}
	\infer5[library]{  \birprog\vdash\sbirtrans{\symcrypto{\symvar{v}}}{}{\tuple{\sbirpcond, \sbirenv, \birpc}}{\tuple{\sbirpcond, \sbirenv{'}[\reg{0} \mapsto \var{a}], \birpc'}} }
\end{prooftree}
\end{array}
\]
}
\caption{Crypto-aware symbolic execution semantics. Here, $ \symvar{e}, \symvar{d}, \symvar{v}, \symvar{x_i} \in \sbirvals$.}
\label{sbirfuncalls}
\end{figure*}

\section{Crypto-aware Symbolic Execution}
\label{sec:cryptoawaresbir}

We have \emph{significantly} extended HolBA's vanilla symbolic execution~\cite{proofproducingsymbexecution} to handle network communication, calls to crypto primitives and event functions, and random number generation, which are essential to reason about protocols' security.
We defined the rules for our symbolic execution in~\autoref{sbirfuncalls}.
In this figure, $\imagesym(f)$ returns the image of $f$.
For library calls, we define an \emph{oracle} $\oraclesym : \entr_{\cryptoOp} \times (\birvar \mapsto \sbirvals) \to \sbirvals$ to compute the result of the invoked function w.r.t. the current $\birpc$ and symbolic environment.
For the symbolic execution, we initialize the memory region to store random numbers $\randommem$ with symbolic values. Thus, $\sbirrand$ signifies the symbolic lifting of $\birrand$, and RNG generates a fresh symbolic expression to represent the extracted value. 

Similar to $\birsymb$ transitions, we extend the symbolic transition relation of $\sbirsymb$ with events, i.e.,
$\sbirtrans{\symvar{a}}{}{}{} \subseteq \sbirstates \times \sbirevents \times \sbirstates$, and use $\sbirexecutions(\birprog, \sbirstate{0})$ to denote the set of symbolic traces of $\birprog$ starting at $\sbirstate{0}$.

\paragraph{\bf Bounded Loops.}\label{bo-loop}
 Loops can na{\"i}vely be handled by \emph{unrolling}. This, however, is inefficient in most cases and can quickly result in a path explosion. To avoid this, we summarize loops following Strej\v{c}ek~\cite{abstractingpc2012}.
The algorithm summarizes the loops' effect on program variables and path conditions to compute a necessary condition on the loop's inputs to reach a specific execution point in the program.
The summary is computed in terms of a tuple of \emph{iterated symbolic state} and \emph{looping condition}.
The iterated symbolic state computes for each variable modified within the loop its symbolic value based on the initial value of the program's variables and \emph{path counters}. Each path counter indicates the number of iterations of a specific path within the loop leading from the loop entry point to itself.
For each path in the loop, a path condition is computed, and the conjunction of all such conditions is the looping condition. 
 
We have \emph{automated} the loop summarization process in our symbolic execution. In~\autoref{sbirfuncalls}, the function $ \loopprocesssym : \sbirstates \to \sbirvals \times (\birvar \mapsto \sbirvals) $ represents our implementation to summarize loops' effect. It takes as input the symbolic state of the loop entry point and reflects the effect of the loop body in its exit state (computed by $\pcexitsym: \entr_{\birc{\mathcal{L}}} \to \lbls$). 
The rule also raises the event $\symloop{\symvar{t}}$ with $\symvar{t}$ being the number of loop iterations.

Loops in protocol implementations are often not bounded; typically, each session runs in a $ while(true)\{..\} $ loop until the server is externally terminated. 
However, the semantics of $ \imlsymb $, like most cryptographic standard models, assumes a bound on the protocol. 
Thus, we need to assume that such loops are externally terminated after some polynomial time in the security parameter. 
This is captured by our automated loop summarization and by translation to the replication operator.

\paragraph{\bf Indirect Jumps.}\label{ind-jmp}
If during symbolic execution of the code, we encounter an indirect jump, e.g., $ \birc{jmp} \ \sbirc{e} $, we evaluate $ \sbirc{e} $ w.r.t. the current state to get an expression $\sbirc{e}'$; we then query the SMT solver for a satisfiable assignment to $\sbirc{tgt} = \sbirc{e}' \birc{\wedge} \sbirpcond$, assuming that $ \sbirc{tgt} $ does not occur in $ \sbirc{e}' $ and $\sbirpcond$. The solver returns one possible target, say $ \sbirc{t} $. We repeat this procedure, each time asking the solver to exclude found targets, until the query becomes unsatisfiable. This technique was sufficient for our experiments; however, for more complex cases, some optimizations would be required, e.g., considering only a subset of possible targets instead of enumerating all.

We symbolically execute $\birsymb$ programs to instrument them with events that facilitate clear observation of implementation behavior and to obtain their execution tree, which is later used to obtain the corresponding $\imlsymb$ model.
A node in this tree is either a branching node 
$\branchingnode(\nodecond, \tree_{\sbirc{1}}, \tree_{\sbirc{2}})$ with the condition $\nodecond$ and sub-trees $\tree_{\sbirc{i}}$, or an event node $\node(\birpc, \nodeevent)$ with $\birpc$ specifying where the event occurred.
We add a $\birc{halt}$ statement at the end of each complete path, i.e., leaves are due to $\birc{halt}$ statements with $\sbirc\bot$ as the event.
An edge connects two nodes iff they are in the transition relation.

The tree is constructed from a $\birsymb$ program and an initial
symbolic state as follows: the root is the initial state. For any
node, including the root, the crypto-aware symbolic execution gives
us up to two successors states.
If the node represents a branching statement, we obtain two successor states. We store the statement's condition in a branching node and proceed to translate the two successor states into subtrees.
If the node represented any other statement, there can only be one or
no successor state, and we store an event node with, or respectively without,
a successor tree. 
Since we abstract function calls and loops, we safely assume that each node in the tree can be uniquely identified by the $\birpc$ of its statement. We define the selection operator  "[]" to extract the node for a given program counter, e.g., $\tree[\birpc]$ will return a node  indexed by $\birpc$.

\autoref{fig:runningexample} shows a fragment of the symbolic execution tree for the client of our running example. Note that, each function call is depicted with two nodes: the first node loads the address of the callee into $\birpc$, and the second node is the actual call, represented as an atomic transition.

\section{Model Extraction}
\label{sec:sbirtoiml}
\begin{figure}[t]
\adjustbox{varwidth=\linewidth,scale=0.95}{%
\begin{equation*} 
\begin{split}
  \tree& =  \node\bnfconcat\tree' \hspace{3cm} \text{Event tree} \\
  & \sbirtoiml{\leafnode\bnfconcat\tree'} \mapsto\ \imlc{0};\sbirtoiml{\tree'}\\
  & \sbirtoiml{\node\ \bnfconcat\tree'} \bnfdef \hspace{2.3cm} \text{Events nodes} \\
  &
 \begin{array}{lll}
  \sbirtoiml{(\birpc,\sbirc\tau)\bnfconcat\tree'}  & \mapsto & \sbirtoiml{\tree'} \\
  \sbirtoiml{(\birpc,\symevent{\symvar{d_1},\dots,\symvar{d_m}})\bnfconcat\tree'} & \mapsto & \imlc{event} (\symvar{d_1},\dots,\symvar{d_m});\sbirtoiml{\tree'}\\
  \sbirtoiml{(\birpc,\syminput{\symvar{v},(\symvar{e_1},\dots,\symvar{e_m})})\bnfconcat\tree'} & \mapsto & \imlc{in}(\imlchannel[\symvar{e_1},\dots,\symvar{e_m}], \symvar{v});\sbirtoiml{\tree'}\\
  \sbirtoiml{(\birpc,\symoutput{\symvar{e},(\symvar{e_1},\dots,\symvar{e_m})})\bnfconcat\tree'} & \mapsto & \imlc{out}(\imlchannel[\symvar{e_1},\dots,\symvar{e_m}],\symvar{e});\sbirtoiml{\tree'}\\
  \sbirtoiml{(\birpc,\symcrypto{\symvar{v}})\bnfconcat\tree'} & \mapsto & \imlc{let}\ \imlvar{x} \ \imlc= \ \symvar{v}\ \imlc{in}\ \sbirtoiml{\tree'} \text{ where} \\ &&  \text{$\imlvar{x}$ is fresh}\\
  \sbirtoiml{(\birpc,\symfreshv{\symvar{x_i}})\bnfconcat\tree'} & \mapsto & \imlc{new}\ \symvar{x_i};\ \sbirtoiml{\tree'}\\
  \sbirtoiml{(\birpc,\symloop{\symvar{t}})\bnfconcat\tree'} & \mapsto & \imlc{!}^{\symvar{t}\leq \imlc{m}} \loopbody{\birpc};\ \sbirtoiml{\tree'}\\
\end{array}\\
  & \sbirtoiml{(\birpc,\branchingnode(\nodecond, \tree_\sbirc{1}, \tree_{\sbirc{2}}))} \mapsto  \imlc{if}\ \sbirtoiml{\nodecond}\ \imlc{then}\ \sbirtoiml{\tree_{\sbirc{1}}}\ \imlc{else}\ \sbirtoiml{\tree_{\sbirc{2}}} \\
  & \sbirtoiml{\nodecond} \bnfdef \hspace{3cm} \text{Expressions}\\
  &
\begin{array}{lll}
  \;\;\sbirtoiml{\var{b} \in \birval}  & \mapsto & \sbirtoiml{\var{b}}\in\imlvals \\
  \;\;\sbirtoiml{\var{x} \in \birvar}  & \mapsto & \sbirtoiml{\var{x}}\in \imlvarspace \\
  \;\;\sbirtoiml{\nodecond_{\sbirc{1}} \birc{\Diamond_b} \nodecond_{\sbirc{2}}}  & \mapsto & \sbirtoiml{\nodecond_{\sbirc{1}}}\sbirtoiml{\birc{\Diamond_b}}\sbirtoiml{\nodecond_{\sbirc{2}}} \hspace{0.25cm} \text{Binary Ops.}\\
  \sbirtoiml{\birc{\Diamond_b}} &\mapsto & \left\{\begin{array}{ll}
                                        \imlc{\wedge} & \birc{AND}    \\
                                        \imlc{\vee}   & \birc{OR}     \\                                        
                                         \imlc{=}    & \birc{Equal}   \\
                                        \imlc{+}  & \birc{Plus}      \\
                                        \dots&
                                        \end{array}
                                 \right.\\
  \sbirtoiml{\birc{\Diamond_u} \nodecond'}          & \mapsto & \sbirtoiml{\birc{\Diamond_u}}\sbirtoiml{\nodecond'} \hspace{0.6cm} \text{Unary Ops.}\\
  \sbirtoiml{\birc{\Diamond_u}} &\mapsto & \left\{\begin{array}{ll}
                                        \imlc{\neg}            & \birc{Not} \\
                                        \imlc{\bot}            & \text{otherwise}
                                        \end{array}
                                 \right.\\
  \sbirtoiml{\sbirc{f}(\symvar{e_1},\dots,\symvar{e_m})}  & \mapsto & \sbirtoiml{\sbirc{f}}(\sbirtoiml{\symvar{e_1}},\dots, \sbirtoiml{\symvar{e_m}}) 
\end{array} \\
\end{split}
\end{equation*}
}
\caption{Rules for the translation of the symbolic execution tree $\tree$ to $ \imlsymb $ model.}
\label{fig:sbirtoiml}
\end{figure}

We now proceed to explain how to automatically extract the $\imlsymb$ model from protocols' $\birsymb$ representation.
Our model extraction approach relies on translating the symbolic execution tree $\tree$ of the protocol under adversarial semantics into its corresponding $ \imlsymb $ model. 
We translate $\tree$ into an executing process $\imlprocessq^{\mathit{full}}$ 
according to the rules depicted in Fig.~\ref{fig:sbirtoiml}, where $\sbirtoiml{\tree}$ represents the compiled process, i.e. $\imlprocessq^{\mathit{full}} = \sbirtoiml{\tree}$.
Since $\tree$ contains all possible execution of protocols and their interactions with the crypto primitives and the attacker, the extracted model includes all behaviors of the protocol at its binary representation (i.e., all attacks present at the binary level are preserved in the extracted model).

Our translation converts leaf nodes into a $ \imlc{nil} $ process $\imlc{0}$. For internal nodes, we translate the event stored in each node into its $ \imlsymb $ counterpart.
Loops are modeled using the replication operator of $ \imlsymb $; $\loopbodysym :  \birc{\entr_{\mathcal{\mathcal{L}}}} \to \imlprocessp$  converts the loop body into its corresponding $ \imlsymb $ process using the defined rules. 
Notice that we do not translate $\sbirc\tau$ events. Fig.~\ref{fig:sbirtoiml} also presents our rules to translate symbolic $\birsymb$ expressions. 
Intuitively, the symbolic execution is used to symbolically compute
the effects of such transitions, while the protocol model only
contains the interactions with the network.
Our rules to translate expressions are standard. 
The only interesting one is the translation of the function application, which is used, e.g., to translate memory $ \birc{load} $/$ \birc{store} $ and bitwise operations. 
For example, this rule translates a memory load operation $\birc{load}(\var{mem}, \var{pa}, l)$, for $l \in \{1,8,16,32,64,128\}$, into $\imlc{read}(\symvar{x_1}, l)$, where $\symvar{x_1}$ is the fresh name chosen for the symbolic value in  $\var{mem}$ at the address $\var{pa}$. 

\autoref{fig:runningexample} presents the $ \imlsymb $ model of the running example. In this model, \texttt{c} is the input and output channel, \texttt{\small bad} is the event that we release if the decryption is not successful, and \texttt{enc} (\texttt{dec}) is the encryption (resp., the decryption).

\subsection{\textsc{CryptoBap} vs. \textsc{Csec-modex} $\imlsymb$ models}

\label{sec:diffimlmodels}

Our derived $ \imlsymb $ models are simpler than \textsc{Csec-modex} without losing accuracy.
To demonstrate this, we use the simple XOR case study from \textsc{Csec-modex} set of case studies.
Simple XOR implements a protocol in which the one-time pad includes both protocol parties. 
The methodology we employ to derive the $\imlsymb$ model from binary code differs significantly from that used in \textsc{Csec-modex} which leads to much simpler models. 
The most glaring difference between \textsc{CryptoBap} and \textsc{Csec-modex} is that our analysis produces a symbolic tree that is translated into an $\imlsymb$ process with conditionals and replication (see~\autoref{fig:sbirtoiml}) instead of a single path that is translated into a linear $\imlsymb$ process. 
Even for the linear subprocesses, our toolchain produces processes that are shorter and often human-readable (at least for small case studies, e.g., XOR). 
This is because the \textsc{Csec-modex} translates each symbolic CVM process (their intermediate representation of C) into $\imlsymb$ and performs the most simplification steps concerning the bitwise operations (concatenation, extraction, etc.) later at the translation step into CryptoVerif and ProVerif. 
For instance, in \textsc{Csec-modex}'s $ \imlsymb $ model for the client side of simple XOR case study, shown in~\autoref{fig:imlmodels}, \verb|(1)^[u,1]| $|$ \verb|nonce1| is simplified to \verb|conc1(nonce1)| in its CryptoVerif model.
Instead, we leverage the support for simplification of these operations at symbolic execution time, because 
(i) support there is much more mature 
(ii) will benefit from future development 
(iii) simplified constraints can also be used for path elimination and simplification in follow-up states.
In addition to this, our $ \imlsymb $ model for the client side of the simple XOR case study is more concise since several assumptions made in \textsc{Csec-modex}'s $ \imlsymb $ model were unnecessary for verifying this particular case study.
\autoref{fig:imlmodels} presents the different $ \imlsymb $ models produced by \textsc{CryptoBap} and \textsc{Csec-modex} for the client side of simple XOR case study.

 \begin{figure*}
 \centering
 \adjustbox{varwidth=\linewidth,scale=0.95}{%
 \begin{tcolorbox}[
   titlerule=3pt, boxsep=0pt, colframe=gray, halign=left, valign=center,
   left=0pt, right=0pt, top=0pt, bottom=0pt,  lower separated=true, sidebyside,
    sidebyside gap=7mm, 
   title={\hspace{0.7em}\textsc{Csec-modex}'s IML model \hspace{5.6cm}\textsc{CryptoBap}'s IML model}
 ]
 \begin{lstlisting}[style=mlstyle, escapeinside={(*@}{@*)},numbers=none]
 assume argv0 = argv0 in
(*@\centerline{\raisebox{-5pt}[0pt][0pt]{$\vdots$}}@*)
 new var1: fixed_20;
 let nonce1 = var1 in
 assume Defined(pad) in
 assume len(pad) = 21 in
 let xor1 = XOR((1)^[u,1]|nonce1, pad) in
(*@\centerline{\raisebox{-5pt}[0pt][0pt]{$\vdots$}}@*)
 let msg1 = xor1 in
 out(c, msg1);
   \end{lstlisting}
 \tcblower
   \begin{lstlisting}[style=mlstyle, escapeinside={(*@}{@*)},numbers=none]
new OTP_48: fixed_64; 
let Conc1_66 = conc1(OTP_48) in
let XOR_70 = exclusive_or(Conc1_66,pad) in
out(c, XOR_70);
 \end{lstlisting}
 \end{tcolorbox}
 }
\caption{$ \imlsymb $ models produced by \textsc{CryptoBap} vs. \textsc{Csec-modex} for the client side of simple XOR case study}
\label{fig:imlmodels}
\end{figure*}

\section{Soundness of \textsc{CryptoBAP}'s Approach}
\label{sec:soundness}

The extracted $\imlsymb$ model should preserve the $\birsymb$ program’s behaviors to ensure that we can transfer the verified properties back to the binary of protocols.

    \subsection{Soundness of translation into $ \imlsymb $}\label{sec:imlsoundness}

To show that our extracted $ \imlsymb $ model preserves the semantics of the protocols' binary, we need to prove that our translation from a crypto-aware symbolic execution tree into an $ \imlsymb $ process is sound, i.e., \emph{for each path in the symbolic execution tree there is an equivalent $ \imlsymb $ execution trace.}

Our symbolic execution supports communication with the attacker, which, like honest protocol parties given by specification, is represented as an $ \imlsymb $ process. Thus, we need to prove soundness in the context of an $ \imlsymb $ process, i.e., that
each execution trace obtained by symbolically executing the $\birsymb$ program in parallel with an $ \imlsymb $ attacker has an equivalent $ \imlsymb $ trace where the translated processes run in parallel with the same attacker.
Our strategy to prove this is to construct an $ \imlsymb $-$ \sbirsymb $, \imlsbir{}, mixed execution semantics to facilitate the communication of $\birsymb$ programs and the $ \imlsymb $ attacker.
\imlsbir{} is generic and considers $\birsymb$ programs and $ \imlsymb $ processes as independent entities running in parallel and communicating through a channel.

\begin{figure*}[t]
\adjustbox{varwidth=\linewidth,scale=0.95}{%
\[
\begin{array}{c}
 \begin{prooftree}

 			   	\hypo{ \birpc \in \lbls_{\birc{\mathcal{N}}} \uplus \entr_{\birc{Op}} \uplus \entr_{\birc{\mathcal{R}}} \uplus \entr_{\birc{\mathcal{E}}}}
 			    	\hypo{\sbirtrans{\sbirevent}{}{\tuple{ \sbirpcond, {\sbirenv}, \birpc}}{\tuple{ \sbirpcond', {\sbirenv'}, \birpc'}} }
 	\infer2[normal]{\birprog \vdash \miximltrans{\sbirevent}{1}{\sbirvaluation}{\tuple{ \sbirpcond, {\sbirenv}, \birpc}, \mixmultisetprocsimlsbir}{\tuple{ \sbirpcond', {\sbirenv}', \birpc'}, \mixmultisetprocsimlsbir}}
 \end{prooftree}
 \\[25pt] 
   \vspace{4mm}
  \begin{prooftree}
 	\hypo{\begin{matrix}
 			\imleval{\imlvar{e}}_{\imlenv} = \imlvar{b} \neq \imlc\bot  \quad \imlvar{b}' = \mathit{truncate}(\imlvar{b}, \mathit{maxlen}(\imlchannel)) \quad \forall j \leq m : \imleval{\imlvar{e_j}}_{\imlenv} = \imlvar{b_j} \neq \imlc\bot  \quad 
 			{\mixmultisetprocsimlsbir}' =  \mathit{reduce}(\{ (\imlenv,   \imlprocessq) \}) \\[5pt]
 			\exists !(\tuple{\sbirpcond,\! \sbirenv,\! \birpc}, \imlchannel[\symvar{e_1},\dots,\symvar{e_m}])\! \in\! \mixmultisetprocsimlsbir \ : \ \birpc \in \entr_{\birc{\attacker_r}} \land	\forall j\! \leq\! m \ : \ \sbirvaluation(\sbirenv[\symvar{e_j}])\! =\! \imlvar{b_j} \neq \imlc\bot \quad  \birpc' = \pcinc{\sbirpcond, \sbirenv, \birpc} \\[5pt]
 			\symvar{e} \!\notin\! \image{\sbirenv} \quad \sbirvaluation(\symvar{e}) \!=\! \imlvar{b}'  \quad
 			(\sbirenv{'},\var{a}) \ = \ \update(\sbirenv, \heap_{\birc{\mathcal{A}}}, \advmem, \symvar{e}, 128) \quad
 			{\mixmultisetprocsimlsbir}''\ = \ \{ (\tuple{\sbirpcond, \sbirenv, \birpc}, \imlchannel[\symvar{e_1},\dots,\symvar{e_m}]) \} \\
 	\end{matrix}}
 	\infer1[$ \imlc{I}to\sbirc{SB}$]{\birprog \vdash \miximltrans{\syminput{\symvar{e},(\symvar{e_1},\dots,\symvar{e_m})}}{1}{\sbirvaluation}{(\imlenv,   \imlc{out}(\imlchannel[\imlvar{e_1},\dots,\imlvar{e_m}], \imlvar{e});\imlprocessq), \mixmultisetprocsimlsbir}{(\tuple{\sbirpcond, \sbirenv{'}[\reg{0} \mapsto \var{a}], \birpc'}, \imlchannel[\symvar{e_1},\dots,\symvar{e_m}]), \mixmultisetprocsimlsbir \uplus {\mixmultisetprocsimlsbir}' \setminus {\mixmultisetprocsimlsbir}''}}
 \end{prooftree}
  \\[25pt] 
  \vspace{3mm}
 \begin{prooftree}
 		\hypo{ \begin{matrix}
 				\birpc \in \entr_{\birc{\attacker_s}}  \quad  \birpc' = \pcinc{\sbirpcond, \sbirenv, \birpc} 	\quad  {\mixmultisetprocsimlsbir}' =  \{ (\sbirpcond, \sbirenv, \birpc') \} \quad \var{a} = \sbirenv[\reg{0}] \quad
 		        \symvar{e} = \memload(\sbirenv,\var{a}) \quad
 		        \sbirvaluation(\symvar{e}) = \imlvar{b} \neq \imlc\bot \quad \imlvar{b}' = \mathit{truncate}(\imlvar{b}, \mathit{maxlen}(\imlchannel)) \\[5pt]
 		         \forall j \leq m : \sbirvaluation(\sbirenv[\symvar{e_j}]) = \imlvar{b_j} \neq \imlc\bot \quad 
 				\exists ! (\imlenv,   \imlprocessq) \in \mixmultisetprocsimlsbir	 : \imlprocessq = \imlc{in}(\imlchannel[\imlvar{e_1},\dots,\imlvar{e_m}], \imlvar{x});\imlprocessp   \ \land	\ \forall j \leq m : \imleval{\imlvar{e_j}}_{\imlenv} = \imlvar{b_j} \neq \imlc\bot	 \quad  {\mixmultisetprocsimlsbir}'' = \{ (\imlenv,   \imlprocessq) \}
 		\end{matrix}}
 		\infer1[$ \sbirc{SB}to\imlc{I}$]{\birprog \vdash \miximltrans{\symoutput{\symvar{e},(\symvar{e_1},\dots,\symvar{e_m})}}{1}{\sbirvaluation}{(\tuple{\sbirpcond, \sbirenv, \birpc}, \imlchannel[\symvar{e_1},\dots,\symvar{e_m}]), \mixmultisetprocsimlsbir}{(\imlenv[\imlvar{x} \mapsto \imlvar{b}'],   \imlprocessp), \mixmultisetprocsimlsbir  \ \uplus   {\mixmultisetprocsimlsbir}'  \setminus {\mixmultisetprocsimlsbir}''}}
 \end{prooftree}
 \\[25pt] 
\begin{prooftree}
	\hypo{ \begin{matrix}
			\forall j \leq m : \imleval{\imlvar{y_j}}_{\imlenv} = \imlvar{b_j} \neq \imlc\bot  \quad
			(\symvar{y_1},\dots,\symvar{y_m}) \notin \image{\sbirenv_{\sbirc{0}}} \quad \forall j \leq m : \sbirvaluation(\symvar{y_j}) = \imlvar{b_j} \neq \imlc\bot \quad
			\forall j \leq m : (\sbirenv_{\sbirc{j+1}} , \var{a_i})= \update(\sbirenv_{\sbirc{j}}, \heap, \memsymb, \symvar{y}_{\sbirc{j}},128) 
	\end{matrix}}
	\infer1[run]{\birprog \vdash \miximltrans{}{1}{\sbirvaluation}{(\imlenv,   \imlc{run}( \birpc, (\imlvar{y_1},\dots,\imlvar{y_m}))), \mixmultisetprocsimlsbir}{\tuple{\sbirpcond, \sbirenv_{\sbirc{m+1}}[\reg{0} \mapsto \var{a_{1}},\dots,\reg{m} \mapsto \var{a_{m}}], \birpc}, \mixmultisetprocsimlsbir }}
\end{prooftree}
\end{array}
 \]
 }
 	\caption{The mixed semantics of $\sbirsymb$ and $ \imlsymb $ shown by \imlsbir{}.}
 	\label{fig:mixedsymbiml}
 \end{figure*}
 
The $ \imlsymb $ process already describes the parallel execution of parties and how
they share secrets. We only need to integrate $\birsymb$ into this
framework. Therefore, we extend $ \imlsymb $ with a construct $\imlc{run}( \birpc, (\imlvar{y_1},\dots,\imlvar{y_m}))$ to initialize $\birsymb$ symbolic memory and transfer control to the $\birsymb$ program specified by the $\birpc$. 
To share the secrets, we generate fresh symbolic values $ \symvar{y_1},\dots,\symvar{y_m} $ and store them in the environment $ \sbirenv $ of the $\birsymb$ program.

In the following, we use $ \mixedbirprog $ as a pair of an $ \imlsymb $ process $ \imlprog $ extended with the $\imlc{run}$ construct and a $\birsymb$ program $ \birprog $ that defines the entry points therein.
Slightly misusing notation, $\mixedexec{\mixediml}{\sbirstate{i}}$ also denotes states of the mixed semantics.

~\autoref{fig:mixedsymbiml} shows the operational semantics of \imlsbir{} which combines $ \imlsymb $ input and output processes~\cite[p.~23]{aizatulin2015verifying} with the transition relations of $\sbirsymb$ in~\autoref{sbirfuncalls}. In the figure, 
rules $\imlc{I}to\sbirc{SB}$ and $ \sbirc{SB}to\imlc{I}$ define the communication between the symbolic $\birsymb$ program and the $ \imlsymb $ process. Using these rules a protocol participant can receive a sent message if its channel identifiers have the same evaluation as the channel identifiers of the sender. 
To send a message, i.e., when $\birpc$ is in the label set $\entr_{\birc{\attacker_s}}$, we first fetch the symbolic value $\symvar{e}$ from the memory location $\reg{0}$, truncate the interpretation of the message according to the maximum length of the $ \imlsymb $ channel $\imlchannel$,\footnote{%
A requirement from \textsc{Csec-modex}'s correctness proof for the translation to ProVerif and CryptoVerif. As the attacker's polynomial bound is chosen after the process, the attacker could send a large message that the process runs out of time reading it. }
and then place it in the $ \imlsymb $ environment $\imlenv$.
When $\birpc$ is in $\entr_{\birc{\attacker_r}}$ and the $\birsymb$ program receives input from the $ \imlsymb $ channel $\imlchannel$, we receive the truncated bitstring $  \imlvar{b}' $ from an $ \imlsymb $ state and generate a fresh symbolic value $ {\symvar{e}} $ such that the interpretation of $ {\symvar{e}} $ is equal to bitstring $  \imlvar{b}' $. 
Then, we store the symbolic value $ {\symvar{e}} $ in the memory and return its address in $\reg{0}$.

We use the standard notion of \emph{trace inclusion} to show the translations' soundness (see~\autoref{thm:iml:traceeq}), i.e., the set of \imlsbir{} execution traces is a subset of the $ \imlsymb $ execution traces. To prove this formally, we define a simulation relation $\simreltra {}{\sbirvaluation}{\sbirtoiml{.}}{} \subseteq  \imlsbirstates \times \imlstates$ between states/events of these two abstraction layers and show that it is preserved by the single-step executions. The simulation relation, $\simreltra{\mixedexec{\mixediml}{\sbirstate{i}}}{\sbirvaluation}{\sbirtoiml{.}}{\imlstate{i}} $, checks if (i) the $ \imlsymb $ output process in the given $ \imlsymb $ state is the correct translation of the symbolic state in $\tree$ according to the rules in~\autoref{fig:sbirtoiml}, i.e.,
$
\imlstate{i}.\imlprocessp = \sbirtoiml{\tree[\sbirstate{i}.\birpc]}
$, and (ii) the environments of the two abstractions are related through the interpretation $\sbirvaluation$, i.e., for all
$
\symvar{x} \in \domain{\mixedexec{\mixediml}{\sbirstate{i}}.\sbirenv}
$ there are
$
\imlvar{x} \in \domain{\imlstate{i}.\imlenv}$ and an
$  
\sbirvaluation
$ s.t.
$
\sbirvaluation(\mixedexec{\mixediml}{\sbirstate{i}}.\sbirenv[\symvar{x}])=\!\imlstate{i}.\imlenv[\imlvar{x}]
$.
Lemma~\ref{lem:iml:stateeventeq:step0} shows that the initial states of \imlsbir{} and the derived $ \imlsymb $ process are in the relation.  

\begin{restatable}{lemma}{mixsbirstateeventeqstepzero}\label{lem:iml:stateeventeq:step0}
	For a symbolic execution tree $ \tree $ of the $\birsymb$ program $ \birprog $,
	an $ \imlsymb $ process $\imlprog$
        and any $ k \in \naturalnum$ the size of the random memory,
        let
        $\initialimlsbir $
        be an initial symbolic state in \imlsbir and 
        $ \initialiml $
        the corresponding initial $ \imlsymb $ state. Then,
        for all $\sbirvaluation$:
        $\simreltra{\mixedexec{\mixediml}{\sbirstate{0}}}{\sbirvaluation}{\sbirtoiml{.}}{\imlstate{0}}.$
\end{restatable}

Next, we show that single-step transitions preserve the simulation relation. 

\begin{restatable}[State/Event Equivalence]{lemma}{mixsbirstateeventeqstepone}\label{lem:iml:stateeventeq}
	Let $\birprog$ be a $\birsymb$ program and
	$\imlprog$ be an $ \imlsymb $ process, then,
	for all $\imlstate{i}$, $\mixedexec{\mixediml}{\sbirstate{i}}$, $\mixedexec{\mixediml}{\sbirstate{j}}$ and $\sbirvaluation$ s.t. $\simreltra{\mixedexec{\mixediml}{\sbirstate{i}}}{\sbirvaluation}{\sbirtoiml{.}}{\imlstate{i}}$ and $\miximltransmulti{\miximlsbirevent}{p}{\sbirvaluation}{\mixedexec{\mixediml}{\sbirstate{i}}}{\mixedexec{\mixediml}{\sbirstate{j}}}$,
there exist an $\sbirvaluation'$ and $\imlstate{j}$ s.t. $\sbirvaluation \subseteq \sbirvaluation'$, $\imltrans{\imlevent}{p}{\imlstate{i}}{\imlstate{j}}$ and $\simreltra{\mixedexec{\mixediml}{\sbirstate{j}}}{\sbirvaluation'}{\sbirtoiml{.}}{\imlstate{j}}$ and if $\imlevent \neq \bot$ then $\miximlsbirevent =_{\sbirvaluation'} \imlevent$.
\end{restatable}

We then show the translation's soundness by extending the simulation relation to execution traces, i.e., $\simreltratraces {}{\sbirvaluation}{\sbirtoiml{.}}{k}{} \subseteq  \mixedsbirexecutions \times \imlexecutions$,
w.r.t an upper bound\footnote{This bound $k$ is needed because of $\birsymb$'s finite memory model.} $ k \in \naturalnum $ on the number of RNG steps of the execution
$\rngsym : \executions \to \naturalnum $. That is, 
$
\simreltratraces{\mixedsbirexecution}{\sbirvaluation}{\sbirtoiml{.}}{k}{\imlexecution}
$ holds, iff, 
$
\rng{\mixedsbirexecution} \leq k
$ and for all
$
\mixedexec{\mixediml}{\sbirstate{}}$ and $ \miximlsbirevent \in \mixedsbirexecution
$ there exist
$
\imlstate{}, \imlevent \in \imlexecution
$, and
$
\sbirvaluation
$ s.t.\;
$
\simreltra{\mixedexec{\!\mixediml}{\sbirstate{}}}{\sbirvaluation}{\sbirtoiml{.}}{\imlstate{}}
$ and
$
\miximlsbirevent\! =_{\sbirvaluation}\! \imlevent
$.

Finally, we show that executions of the mixed $ \imlsymb $ and symbolic execution and $ \imlsymb $ preserve the simulation relation. Note that, in the following, we assume a single $\birsymb$ program that implements different protocol participants with distinct sets of program counters. The results can be extended to multiple $\birsymb$ programs, as presented in Appendix~\ref{multi-programs}.

\begin{theorem}[Trace Inclusion]
	\label{thm:iml:traceeq}
	Let $\birprog$ be a $\birsymb$ program,
	$\imlprog$ be an $ \imlsymb $ process, and
        $ k \in \naturalnum $ is any upper bound on the number of RNG steps, then,
        for all mixed $ \imlsymb $ and symbolic execution traces
        $\mixedsbirexecution\!\in\!\mixedsbirexecutions(\mixedbirprog,$ $ \sbirenv_{\sbirc{0}}[ \randommem \mapsto \randomsymvals_{\sbirc{k}}  , \randommemidx \mapsto 0])$
        s.t. 
        $ \rng{\mixedsbirexecution}\! \leq\! k $,
        there are an $ \imlsymb $ trace 
        $\imlexecution\! \in\! \imlexecutions(\imltransprog) $
        and an $\sbirvaluation$ s.t. $\simreltratraces{\mixedsbirexecution}{\sbirvaluation\!}{\sbirtoiml{.}}{k}{\!\imlexecution}$.
\end{theorem}

 \begin{proof}
	The goal is to show that for all \imlsbir{} traces, there is an equivalent $ \imlsymb $ trace that are in the simulation relation through the interpretation $\sbirvaluation$.
	We prove the theorem by induction on the length of the execution traces: 
 
	\begin{itemize}
		\item \textbf{Base case}. Follows from Lem.~\ref{lem:iml:stateeventeq:step0}.
		\item \textbf{Inductive case}. Follows from Lem.~\ref{lem:iml:stateeventeq}.
	\end{itemize}
	We present the proof of lemmas~\ref{lem:iml:stateeventeq:step0} and~\ref{lem:iml:stateeventeq} in Appendix~\ref{sound-iml-trans}.
\end{proof}

\autoref{thm:iml:traceeq} is the first step to relating the
properties we verify for the $ \imlsymb $ model to the actual binary of the
protocol. We showed that the
$ \imlsymb $ model resulting from translation covers all behaviors in the \imlsbir{} semantics. Recall that we have to talk about behavioral properties in the mixed
semantics (as opposed to the pure $\birsymb$ semantics) as protocol
properties typically concern more than one party.
Next, we show that these symbolic behaviors cover all
concrete behaviors.

    \subsection{Soundness of Symbolic Execution}
\label{sec:mixedbirtomixedsbir}
To ensure that the extracted $ \imlsymb $ model preserves the semantics of the protocol's binary, we have to prove further that our symbolic execution is behaviorally equivalent to the transpiled $\birsymb$ code. 
To show this, we construct a mixed $ \imlsymb $-$\birsymb$ execution semantics, hereafter \imlbir{}, that allows the $\birsymb$ program to communicate with the same $ \imlsymb $ attacker at the \imlsbir{} level. 
The \imlbir{} execution semantics is presented in Appendix~\ref{app:imlbirsemantics}-Fig~\ref{fig:mixedbiriml}---the rules are similar and have the same meaning as those defined for \imlsbir{}.

Our proof strategy to show the behavioral equivalence of \imlbir{} and \imlsbir{} is similar to our technique to prove the soundness of the $ \imlsymb $ translation. That is, we first show the state/event equivalence between the two abstractions and then use this to prove the trace inclusion of \imlbir{} in \imlsbir{}.

We show the state/event equivalence by extending the simulation relation of~\autoref{single-step-vanila} to a relation 
$\simrel {}{\sbirvaluation}{} \subseteq  \imlbirstates \times \imlsbirstates$ between \imlbir{} and \imlsbir{}. The relation 
$
\simrel{\mixedexec{\mixediml}{\birstate{i}}}{\sbirvaluation}{\mixedexec{\mixediml}{\sbirstate{i}}}
$ checks that 
$
\mixedexec{\mixediml}{\sbirstate{i}}.\birpc = \mixedexec{\mixediml}{\birstate{i}}.\birpc
$, and for all
$
\var{x} \in \domain{\mixedexec{\mixediml}{\birstate{i}}.\birenv}
$  there exist
$
\symvar{x} \in \domain{\mixedexec{\mixediml}{\sbirstate{i}}.\sbirenv}
$ and an interpretation
$
\sbirvaluation
$ s.t.
$\sbirvaluation(\mixedexec{\mixediml}{\sbirstate{i}}.\sbirenv[\symvar{x}])\!=\!\mixedexec{\mixediml}{\birstate{i}}.\birenv[\var{x}]
$. The first step in showing this simulation relation between the two layers is to prove that the initial states are in the relation using Lem.~\ref{lem:bir:stateeventeq:step0}:
 
\begin{restatable}{lemma}{mixbirstateeventeqstepzero}\label{lem:bir:stateeventeq:step0}
 	For a $\birsymb$ program $\birprog$,
 	an $ \imlsymb $ process $\imlprog$ and
 	any upper bound $ k \in \naturalnum $ on the number of RNG steps, 
 	let  $  \initialimlbir $ be an initial $\birsymb$ state in \imlbir{} and 
 	$  \initialimlsbir$ be the corresponding initial state in \imlsbir{}.
 	Then, $ \simrel{\mixedexec{\mixediml}{\birstate{0}}}{\sbirvaluation}{\mixedexec{\mixediml}{\sbirstate{0}}}$ for all  $\sbirvaluation$.
 \end{restatable}

 We then prove that the single-step transitions of  \imlbir{} and \imlsbir{}  preserve the simulation relation using Lem.~\ref{lem:bir:stateeventeq}.
 
\begin{restatable}[State/Event Equivalence]{lemma}{mixbirstateeventeqstepone}\label{lem:bir:stateeventeq}
	Let $\birprog$ be a $\birsymb$ program and
	$\imlprog$ be an $ \imlsymb $ process, then,
	for all $\mixedexec{\mixediml}{\sbirstate{i}}$, $\mixedexec{\mixediml}{\birstate{i}}$, $\mixedexec{\mixediml}{\birstate{j}}$ and $\sbirvaluation$ s.t. $ \simrel{\mixedexec{\mixediml}{\birstate{i}}}{\sbirvaluation}{\mixedexec{\mixediml}{\sbirstate{i}}} $ and $\imltransmulti{\miximlbirevent}{p}{\mixedexec{\mixediml}{\birstate{i}}}{\mixedexec{\mixediml}{\birstate{j}}}$, 	there exist an $\sbirvaluation'$ and $\mixedexec{\mixediml}{\sbirstate{j}}$ s.t. $\sbirvaluation \subseteq \sbirvaluation'$, $\miximltransmulti{\miximlsbirevent}{p}{\sbirvaluation'}{\mixedexec{\mixediml}{\sbirstate{i}}}{\mixedexec{\mixediml}{\sbirstate{j}}}$, $ \simrel{\mixedexec{\mixediml}{\birstate{j}}}{\sbirvaluation'}{\mixedexec{\mixediml}{\sbirstate{j}}} $ and  $\miximlbirevent =_{\sbirvaluation'} \miximlsbirevent$.

\end{restatable}

We show the behavioral equivalence between the two layers by extending the simulation relation to execution traces
$\simreltraces {}{\sbirvaluation}{k}{} \subseteq  \mixedbirexecutions \times \mixedsbirexecutions$ w.r.t an upper bound $ k \in \naturalnum $ on the number of RNG steps. That is,
$
\simreltraces{\mixedbirexecution\!}{\!\sbirvaluation}{k}{\mixedsbirexecution}
$ holds, iff,
$
\rng{\mixedbirexecution} \leq k
$, and for all
$
\mixedexec{\mixediml}{\birstate{}}, \miximlbirevent \in \mixedbirexecution
$ there exist
$
\mixedexec{\mixediml}{\sbirstate{}}, \miximlsbirevent \in \mixedsbirexecution
$ and an
$
\sbirvaluation
$ s.t.
$
\simrel{\mixedexec{\mixediml}{\birstate{}}}{\sbirvaluation}{\mixedexec{\mixediml}{\sbirstate{}}} 
$ and
$ 
\miximlbirevent\!=_{\sbirvaluation}\! \miximlsbirevent
$.

\begin{theorem}[Trace Inclusion]
	\label{thm:bir:traceeq}
	Let $\birprog$ be a $\birsymb$ program,
	$\imlprog$ be an $ \imlsymb $ process, and
	$ k \in \naturalnum $ is any upper bound on RNG steps, then, 
	for all \imlbir{} traces $\mixedbirexecution\in\mixedbirexecutions(\mixedbirprog, \birenv_{\birc{0}}[ \randommem \mapsto \randomvals_{\imlc{k}} , \randommemidx \mapsto 0])$ s.t. $ \rng{\mixedbirexecution} \leq k $, there are an \imlsbir{} trace $\mixedsbirexecution\in\mixedsbirexecutions(\mixedbirprog, \sbirenv_{\sbirc{0}}[ \randommem \mapsto \randomsymvals_{\sbirc{k}}  , \randommemidx \mapsto 0])$ and an $\sbirvaluation$ s.t. $	\simreltraces{\mixedbirexecution}{\sbirvaluation}{k}{\mixedsbirexecution}$.
\end{theorem}

\begin{proof}
	\autoref{thm:bir:traceeq} shows that for all \imlbir{} traces, there is an equivalent \imlsbir{} trace through a properly chosen interpretation $\sbirvaluation$.
	We prove \autoref{thm:bir:traceeq} by induction on the length of the traces. 
	\begin{itemize}
		\item \textbf{Base case}. The base case can be proved using Lem.~\ref{lem:bir:stateeventeq:step0}.
		\item \textbf{Inductive case}. The inductive step can be proved using Lem.~\ref{lem:bir:stateeventeq}.
	\end{itemize}
	Appendix~\ref{mixedbirtomixedsbirsoundness} presents the proof of lemmas~\ref{lem:bir:stateeventeq:step0} and~\ref{lem:bir:stateeventeq}.
\end{proof}

\autoref{thm:bir:traceeq}
shows that, for an appropriately chosen interpretation and random memory, 
symbolic and concrete executions of a $\birsymb$ program are behaviorally equivalent.
This holds in the mixed $ \imlsymb $-($ \sbirc{S} $)$ \birsymb $ semantics, i.e.,
when coupled with the same $ \imlsymb $ attacker and protocol partners.

\section{Security properties}\label{sec:secproperties}

From the simulation results between concrete $\birsymb$, symbolic $\birsymb$ and extracted $ \imlsymb $, we will now conclude our target result, which argues that probabilistic security results translate across these levels of abstraction. The security properties we consider, i.e., authentication and weak secrecy, are safety properties over event traces.
Specifically, we consider a security property $ \traceproperty $ as a polynomially decidable prefix-closed set of event traces.

\paragraph{Example.}
    For SSH, we show authentication
    between the events
    $ \mathit{Acpt_S}(PK_S, PK_C) $
    (in the server model derived from the TinySSH binary)
    and
    $ \mathit{Acpt_C}(PK_S, PK_C) $ (in the client model implemented based on the SSH specification)
    where $ PK_S $ and $ PK_C $ are the server's public key and the client's public key, respectively. 
\[
\begin{array}{l}
	\mathbf{Auth} =	\{ \eventtrace  \in \traceproperty \ | \ \forall i \in \naturalnum : \eventtrace[i] = \mathit{Accept_S}(PK_S, PK_C)\\ 
	\hspace{1mm}\begin{array}{l}	
		\;\;\;\;\;\;\;\;\;\;\implies \left( \exists j \in \naturalnum: j < i \ \land \ \eventtrace[j] = \mathit{Accept_C}(PK_S, PK_C)\!\right)
	\end{array}\!\}
\end{array}
\]

We quantify the probability of a protocol remaining secure by
considering the complementary probability: the sum of the
probabilities of each violation. To avoid double counting, we only sum
over the set of shortest violating prefixes, i.e.,
$\tracepropertyneg = \{ \eventtrace \notin \traceproperty | \forall
\eventtrace' . \text{$\eventtrace'$ is prefix of $\eventtrace$}\!\implies\! \eventtrace'\!\in\!\traceproperty \}$.
As security properties are prefix-closed, this captures the
probability of a violation.
The system we analyze consists of the protocol implementations in $\birsymb$. 
 Say $ {T^{\alpha}} $ denotes a set of event traces obtained from the respective set of execution traces $\mathcal{R}^{\alpha}\filter{\mathit{events}} $, $\mathit{pr}$ is a probability distribution function that computes the probability of an event trace and $\imlvals_{{n}}^{{k}}$ is a set of bit strings for generating $k$ random numbers of length $n$,  then:

 \begin{definition}[$\birsymb$ insecurity]\label{def:BIR-insec}
    For a $\birsymb$ program $\birprog$,
    an $ \imlsymb $ process $\imlprog$,
    a security parameter $\secparam\in\naturalnum$,
    and 
    $k\in \naturalnum$ the size of $\birsymb$'s random memory,
    the 
    insecurity of $\mixedbirprog$ w.r.t. $\traceproperty$ is:
    $
	\insecurity{\mixedbirprog,n,k}{\traceproperty} = 
        2^{-n \cdot k} \cdot \sum_{\substack{
         \randomvals_{\imlc{k}} \in \imlvals_{{n}}^{{k}} \\
\mixedbirtrace \in \mixedbirtraces_{ \randomvals_{\imlc{k}}} \cap  \tracepropertyneg }}
    \birpr(\mixedbirtrace)
    $
    where $\mixedbirtraces_{ \randomvals_{\imlc{k}}}
    =  \mixedbirtraces(\mixedbirprog, \birenv_{\birc{0}}[ \randommem \mapsto  \randomvals_{\imlc{k}}  , \randommemidx \mapsto 0])$.
\end{definition}

After translating $\birprog$ into  $\sbirtoiml{\birprog}$, we define insecurity in terms of $ \imlsymb $'s probabilistic semantics.

\begin{definition}[$ \imlsymb $ insecurity]\label{def:IML-insec}
    The insecurity of an $ \imlsymb $ process $\imlprog$ w.r.t a
    trace property $\traceproperty$
    and a security parameter $n$ is 
    $
    \insecurity{\imlprog, n}{ \traceproperty} = 
    \sum_{\imltrace \in \imltraces(\imlprog,n)\cap \tracepropertyneg } \imlpr(\imltrace)
    $
    where 
$\imlpr(
\imltrans{}{p_1}{\imlstate{0}}{\imlstate{1}}
\cdots
\imltrans{}{p_n}{\imlstate{n-1}}{\imlstate{n}}
) = \prod\limits_{1 \leq i \leq n} \imlc{p_i} $. 
\end{definition} 

Note that definitions \ref{def:BIR-insec} and \ref{def:IML-insec}
coincide on \imlbir{} processes 
that do not contain the
$\imlc{run}$-construct, as in this case, the \textsc{RNG} rule (like
any other $\birsymb$ rule) can never be applied and thus $k$ be chosen to be
$0$. This applies to the \imlbir{} processes resulting from our
translation.
\autoref{thm:soundness} shows the translation is sound. Note that $\imlprog$ contains $\birsymb$ programs (via the $\imlc{run}$ construct), but also $ \imlsymb $ processes that represent communication partners and the network attacker.
\begin{restatable}[Translation preserves attacks]{theorem}{thmsoundness}\label{thm:soundness}
    Given a $\birsymb$ program $\birprog$,
    an $ \imlsymb $ process $\imlprog$,
    a security parameter $n \in \naturalnum$,
    a trace property $ \traceproperty$
    and
    an upper bound $k\in \naturalnum$
    on the number of \textsc{RNG} steps in $\mixedbirtraces(\mixedbirprog, \birenv_{\birc{0}}[ \randommem \mapsto  \randomvals_{\imlc{k}}  , \randommemidx \mapsto 0])$, 
    we get that
    \begin{center}
    $
	\insecurity{\mixedbirprog,n,k}{\traceproperty} \le
    \insecurity{\imltransprog, n}{ \traceproperty}
    $.
        \end{center}
\end{restatable}

    Via~\cite[Thm.~4.3,Thm.~5.2]{aizatulin2015verifying} we obtain a bound for $ \insecurity{\imltransprog, n}{ \traceproperty}$ from either of the backends, ProVerif or CryptoVerif.
	In cryptography, probability bounds are expressed as asymptotic functions in the security parameter. 
	CryptoVerif provides a symbolic expression of such a probability bound and, furthermore, proves that the bound is negligible, i.e., it decreases faster than the inverse of any polynomial. 
	On the other hand, ProVerif only confirms the existence of a negligible bound. 
	In both cases, the existence of this negligible upper bound ensures $\insecurity{\mixedbirprog,n,k}{\traceproperty}$ is negligible.
 
We  present the proof of~\autoref{thm:soundness} in Appendix~\ref{sec:properties}. 
	Based on definitions~\ref{def:BIR-insec} and~\ref{def:IML-insec}, we calculate the probability distribution in both $ \imlsymb $ and \imlbir{}. 
	While the probability of all transitions except for random number generation is 1, we need to demonstrate other requirements, such as extra randomness, injective event trace inclusion, etc.
	To this end, we present lemmas in~\autoref{lem:thm4lems} that show these requirements, which is necessary to prove~\autoref{thm:soundness}.

\begin{table*}
\centering
\resizebox{2\columnwidth}{!}{
\begin{tabular}{ll|c|c|c|c|c|c|c|ccc|c|c}
\multicolumn{2}{c|}{\multirow{2}{*}{Protocol}} & \multirow{2}{*}{\begin{tabular}[c]{@{}c@{}}ARM\\ Loc\end{tabular}} & \multirow{2}{*}{\begin{tabular}[c]{@{}c@{}}Verified\\  Code Size\end{tabular}} & \multirow{2}{*}{\begin{tabular}[c]{@{}c@{}}Feasible\\ Path\end{tabular}} & \multirow{2}{*}{\begin{tabular}[c]{@{}c@{}}Infeasible\\ Path\end{tabular}} & \multirow{2}{*}{\begin{tabular}[c]{@{}c@{}}$\imlsymb$\\ Loc\end{tabular}} & \multirow{2}{*}{\begin{tabular}[c]{@{}c@{}}CryptoVerif\\ (CV) Loc\end{tabular}} & \multirow{2}{*}{\begin{tabular}[c]{@{}c@{}}ProVerif\\ (PV) Loc\end{tabular}} & \multicolumn{3}{c|}{Time (Second)} & \multirow{2}{*}{Verified in} & \multirow{2}{*}{Primitives} \\
\multicolumn{2}{c|}{} &  &  &  &  &  &  &  & \multicolumn{1}{c|}{$\imlsymb$} & \multicolumn{1}{c|}{CV} & PV &  &  \\ \hline
\multicolumn{2}{l|}{RPC} & 1.8K & 0.659K & 8 & 178 & 23 & 236 & 102 & \multicolumn{1}{c|}{10} & \multicolumn{1}{c|}{0.035} & 0.012 & CV \& PV & UF-CMA MAC \\
\multicolumn{2}{l|}{RPC-enc} & 53K & 0.294K & 28 & 348 & 53 & 313 & 118 & \multicolumn{1}{c|}{9} & \multicolumn{1}{c|}{0.073} & 0.047 & CV \& PV & IND-CPA INT-CTXT AE \\
\multicolumn{2}{l|}{CSur} & 0.7K & 0.382K & 11 & 237 & 29 & 277 & 177 & \multicolumn{1}{c|}{6} & \multicolumn{1}{c|}{0.656} & 0.035 & flaw & IND-CCA2 PKE \\
\multicolumn{2}{l|}{NSL} & 2.8K & 0.595K & 23 & 455 & 56 & 296 & 204 & \multicolumn{1}{c|}{35} & \multicolumn{1}{c|}{2.740} & 0.052 & flaw, verified & IND-CCA2 PKE \\
\multicolumn{2}{l|}{Simple MAC} & 1.5K & 0.294K & 13 & 149 & 29 & 207 & 101 & \multicolumn{1}{c|}{8} & \multicolumn{1}{c|}{0.047} & 0.033 & CV \& PV & UF-CMA MAC \\
\multicolumn{2}{l|}{Simple XOR} & 2.4K & 0.100K & 2 & 50 & 7 & 141 & -- & \multicolumn{1}{c|}{4} & \multicolumn{1}{c|}{0.40} & -- & CV & XOR \\
\multicolumn{2}{l|}{TinySSH} & 18K & 0.476K & 136 & 1079 & 87 & 286 & 190 & \multicolumn{1}{c|}{55} & \multicolumn{1}{c|}{0.077} & 0.079 & CV \& PV & CRHF-CDH$^{\bf \dagger}$ \& UF-CMA SIGN \\
\multicolumn{1}{c|}{\multirow{2}{*}{WG}} & Initiator & \multirow{2}{*}{27K} & \multirow{2}{*}{1.323K} & 68 & 1482 & 181 & \multirow{2}{*}{--} & \multirow{2}{*}{222} & \multicolumn{1}{c|}{52} & \multicolumn{1}{c|}{\multirow{2}{*}{--}} & \multirow{2}{*}{59.646} & \multirow{2}{*}{PV} & \multirow{2}{*}{DH-X25519$^{\bf \ddagger}$ \& ROM-hash$^{\star}$} \\
\multicolumn{1}{c|}{} & Responder &  &  & 153 & 1389 & 464 &  &  & \multicolumn{1}{c|}{47} & \multicolumn{1}{c|}{} &  &  & 
\end{tabular}%
}
\caption{Case studies. ARM assembly includes crypto code and the code that is used in preprocessing, e.g., to compute the control flow of the program required in loop summarization. Primitives are standard, except Collision-resistant hash based on computational Diffie-Hellman ($\dagger$), Curve25519~\cite{langley2016elliptic} ($ \ddagger $) and hash function in the Random Oracle Model~\cite{bellare1993random} ($ \star $). Also, WG = WireGuard.}
\label{fig:testcases}
\vspace{-1em}
\end{table*}

\section{Evaluation}
\label{sec:evaluation}
We have implemented \framework{} on the HOL4 theorem prover~\cite{hol4} using its metalanguage SML. 
\framework{} relies on HolBA's semantics-preserving transpiler and symbolic execution~\cite{DBLP:journals/scp/LindnerGM19,proofproducingsymbexecution}. We significantly extended the HolBA vanilla symbolic execution to handle crypto primitives, communication with the attacker, indirect jumps, and loops, which are essential to verify the security of protocols. 
We also adapted the \textsc{Csec-modex}'s pipeline  to process \framework{}-generated $ \imlsymb $. \autoref{fig:testcases} shows the list of protocol implementations that we used in our evaluation.

\paragraph{\bf Simple case studies.}
Apart from the XOR case study discussed in~\autoref{sec:diffimlmodels}, we also verified other case studies of~\cite{aizatulinComputationalVerificationProtocol2012,aizatulinExtractingVerifyingCryptographic2011} to evaluate \framework{}.
The only exception to~\cite{aizatulinComputationalVerificationProtocol2012,aizatulinExtractingVerifyingCryptographic2011} is \emph{smart metering protocol} which is not open source.
For other cases, we obtained the same result. Additionally, in contrast to~\cite{aizatulinExtractingVerifyingCryptographic2011}, which could not handle CSur, we successfully verified this case study.

 RPC implements the MAC-based remote procedure call protocol~\cite{bengtson2011refinement}. We verified the client request and the server response authenticity under the MAC unforgeability assumption against chosen-message attacks with symbolic and computational guarantees. RPC-enc is an implementation of the RPC protocol that uses authenticated encryption. We also verified the secrecy of the payloads  (which is not protected by the MAC-based RPC) with an assumption that authenticated encryption is indistinguishable against the chosen-plaintext attack and provides ciphertext integrity.
CSur is the Needham-Schroeder public-key authentication protocol~\cite{needhamUsingEncryptionAuthentication1978}. We verified the secrecy and authentication properties for the CSur binary. Our analysis confirmed that CSur is vulnerable to attack in~\cite{loweAttackNeedhamSchroederPublickey1995} and leaks protocol parties’ nonces.
Similar to~\cite{aizatulinComputationalVerificationProtocol2012}, we also removed the assumption (i.e., all cryptographic material plus nonce are tagged) used in~\cite{aizatulinExtractingVerifyingCryptographic2011} for the Needham-Schroeder-Lowe (NSL) case study to obtain the computational soundness result. We confirmed the flaw in the protocol discovered in~\cite{aizatulinComputationalVerificationProtocol2012}: if the nonce of the second party is not tagged and is sent separately, it can be (mis)used as the first protocol message.

Simple MAC implements the first half of the RPC protocol in which a single payload is concatenated with its MAC~\cite{bengtson2011refinement}. We verified the payload authenticity under the unforgeability of MAC against the chosen-message attack assumption. 
For simple XOR case study, we verified the secrecy of the payload with CryptoVerif.
We did not attempt verification with ProVerif, as the analysis of theories with XOR requires extra effort~\cite{kusters2008reducing}, while CryptoVerif's attacker model is strictly stronger.

\paragraph{\bf Verification of TinySSH and WireGuard.} We further evaluated \framework{} by verifying TinySSH and WireGuard protocols.
TinySSH is a minimalistic SSH server that implements a subset of SSHv2
features and ships with its own crypto library.
 To formulate
authentication properties, which ought to hold for any communication
partner to the TinySSH server that
conforms to the SSH protocol specification, we 
modeled the client side of the SSH protocol in $ \imlsymb $; agents at the other end of the communication line are manually developed.
We manually specified the cryptographic assumptions about the primitives used by the TinySSH implementation in CryptoVerif and ProVerif templates.
We verified mutual authentication with ProVerif and CryptoVerif.

WireGuard implements virtual private networks akin to IPSec and OpenVPN. It is quite recent and was incorporated into the Linux kernel (stable) in March 2020. We have automatically extracted, \emph{for the first time}, the WireGuard's ProVerif model from its linux implementation binary---all other existing models are hand-written and derived from the specification~\cite{donenfeld2017formal,lipp2019mechanised,kobeissi2019noise}.
Existing WireGuard's symbolic models often utilize pattern matching to verify authentication.
This involves constructing the entire message from the initial stage and subsequently comparing it to the received message from the network, which requires further justification.
In contrast, our extracted model from the implementation closely represents the actual behavior of WireGuard, as it deconstructs the network message using ChaCha20Poly1305 Decryption for the purpose of authenticating ChaCha20Poly1305 Aauthenticated Encryption with Associated Data (AEAD).
As a result, our model obviates the need for additional caution when verifying the authentication property. 
We model the handshake and first transport message, after which the key exchange is concluded, and prove the protocol participants mutually agree on the resulting keys.

	TinySSH and WireGuard are case studies with indirect jumps in their binary, and TinySSH is the only one that features multiple sessions. The other case studies from~\cite{aizatulinComputationalVerificationProtocol2012,aizatulinExtractingVerifyingCryptographic2011} do not contain loops as~\cite{aizatulinComputationalVerificationProtocol2012,aizatulinExtractingVerifyingCryptographic2011} did not support them. 
	We handle the loops  inside the implementation of TinySSH using the summarization technique~\autoref{sec:cryptoawaresbir}, which is automated by \framework{}---these loops handle, e.g., reading key files of variable size.
	TinySSH spawns a new server for each incoming TCP connection by using \texttt{inetd}, which we correctly cover in our template files. 
	We hence handle multiple sessions, including potential replay attacks, correctly.

\section{Concluding Remarks}
We introduced \framework{} to analyze the binary of cryptographic protocols. \framework{} enables sound verification of authentication and weak secrecy for protocols' machine code using the model extraction technique.
We have automated extracting models from the binary of protocols and formally proved that the extracted $\imlsymb$ model preserves the protocol's behavior at its machine-code level. Additionally, we showed that the verified properties for the $\imlsymb$ model could be transferred back up to a mixed execution semantics, where the $\birsymb$ translation of the protocol
can be coupled and communicate with an $\imlsymb$ attacker.

The problem of deciding secrecy is generally undecidable~\cite{mitchell1999undecidability} even in dedicated protocol specification languages whose semantics are designed for verification. Hence the verification tool at the backend must necessarily sacrifice automation or completeness. 
In practice, ProVerif sacrifices completeness but verifies many typical protocols.
It is as easy to construct programs that are not recognized as secure by ProVerif as it is to write programs that purposefully break our abstractions in symbolic execution (e.g., by using a randomly generated key to guide the control flow). 
Lacking a syntactic criterion of what protocol implementations should or should not do, we have evaluated the automation and completeness of \framework{} via case studies while ensuring soundness through proofs.

The current implementation of \framework{} has a few limitations. 
At the protocol verification level, our limits are currently those of the protocol verification backend. These are constantly improving and we are currently working on targeting Tamarin~\cite{meierTAMARINProverSymbolic2013} as well.
An example of our limitations at the code handling level is \textbf{infinite loops}.
The soundness of handling loops using the replication operator $\imlc{!^m P}$ can 
hold as long as enough randomness can be stored in the memory, which is finite. 
Hence, the number of RNG steps is limited by $ \birsymb $'s finite memory in combination with its deterministic semantics. 
Thus, $\birsymb$ needs to be extended with external non-determinism to handle infinite loops.
Moreover, similar to other related works, we trust the implementation of cryptographic primitives. This means if there is any bug in these primitives that violates the protocols' security, these violations will not be detected.
The exclusion of these primitives from our verification is because their verification requires a different methodology (e.g., weakest precondition propagation). This is different from the goal we set ourselves in the paper. Also, while the \framework{'s} approach is generic, at present, we can analyze only ARM and RISC-V binaries. To extend our analysis to other architecture, the only part that needs to be extended is the HolBA transpiler.
Finally, we have tried to reduce manual effort, yet, there still remains a necessity for the user to initialize and steer the verification process. Mainly, the user specifies (i) The code fragments they want to verify: the code-under-analysis to the lifter, and the function names that are trusted (libraries) or untrusted (attacker/network) to the symbolic execution engine; (ii) the symbolic model of cryptographic functions, and (iii) the security assumptions on the cryptographic primitives and security queries in the verifiers' template files.

In the future, we plan to mechanize our proofs in a proof assistant such as HOL4~\cite{hol4}. Moreover, to better cover top-level loops, it is necessary to handle non-monotonous global states, i.e., state changes that alter how a protocol reacts to future messages. 
This problem was recognized in protocol verification, for instance, in the more recent Tamarin verifier~\cite{meierTAMARINProverSymbolic2013} or ProVerif's global state extensions~\cite{chevalLittleMoreConversation2018}.
It would thus be promising to skip Aizatulin's toolchain altogether and target one of these tools.
Tamarin's multiset-rewrite calculus is well-suited for modeling state machines. Moreover, one could essentially reuse the state-access axioms from SAPIC~\cite{kremerAutomatedAnalysisSecurity} to handle persistent data stored in the heap.

\section*{Acknowledgments}
This work was supported in part by a gift from Intel; and by the German Federal Ministry of Education and Research (BMBF) (FKZ: 13N1S0762). We thank the anonymous reviewers for their valuable feedback during the review process. We also thank Andreas Lindner for helping with the \framework{} implementation.

\bibliographystyle{ACM-Reference-Format}
\balance
\bibliography{refs}

\begin{full}
\newpage
\appendix
\section{Appendix}
\subsection{\imlsymb{} Semantics}
\label{app:imlsemantics}
\autoref{fig:imlsemantics} presents the $ \imlsymb $ semantics.
The inference rules for random number generation and sending a message on the channel $\imlchannel$ are explained in~\autoref{fig:imltranssemantics}.
The construct $\imlc{event}(\imlvar{d_1},\dots,\imlvar{d_m})$ is used to raise an event $ \imlev{\imlvar{b_1}, \dots , \imlvar{b_m}} $ during the execution where $ \imlevsym $ is an event symbol and $ \imlvar{b_1}, \dots , \imlvar{b_m} $ are bitstrings stored for expressions $ \imlvar{d_1},\dots,\imlvar{d_m} $ in the $ \imlsymb $ environment $ \imlenv $.
The $ \imlc{if} \ \imlvar{e} \ \imlc{then} \ \imlprocessp \ \imlc{else} \ \imlprocessp' $ construct reduces to $ \imlprocessp $ if the expression $ \imlvar{e} $ evaluates to \emph{true} with respect to the $ \imlsymb $ environment $ \imlenv $, otherwise reduces to $ \imlprocessp' $.
The construct $ \imlc{assume} \ \imlvar{e};\imlprocessp $ reduces to $ \imlprocessp $ only in the case that the expression $ \imlvar{e} $ evaluates to \emph{true} with respect to the $ \imlsymb $ environment $ \imlenv $.
To execute the construct $ \imlc{let} \ \imlvar{x = e} \ \imlc{in} \ \imlprocessp $, the bitstring $ \imlvar{b} $ stored for the expression $ \imlvar{e} $ in the $ \imlsymb $ environment $ \imlenv $ is fetched. Subsequently, the value of the variable $ \imlvar{x} $ is updated with the bitstring $ \imlvar{b} $ in the environment $ \imlenv $ and the construct reduces to $ \imlprocessp $.

\begin{figure}[t]
	\adjustbox{varwidth=\linewidth,scale=0.88}{%
			\centering
\begin{prooftree}
		 		\hypo{ \begin{matrix}
		 				\imleval{\imlvar{e}}_{\imlenv} = \text{true}
		 		\end{matrix}}
		 		\infer1[]{ \imltrans{}{1}{(\imlenv , \imlc{assume} \ \imlvar{e};\imlprocessp), \imlmultisetprocs}{(\imlenv , \imlprocessp), \imlmultisetprocs} }
		 	\end{prooftree}
		 	 \\[10pt] 
		 	\begin{prooftree}
		 		\hypo{ \begin{matrix}
		 				\imleval{\imlvar{e}}_{\imlenv} = \imlvar{b} \in \{ \imlvals \cup \{\bot\} \}
		 		\end{matrix}}
		 		\infer1[]{ \imltrans{}{1}{(\imlenv ,\imlc{let} \ \imlvar{x = e} \ \imlc{in} \ \imlprocessp), \imlmultisetprocs}{(\imlenv[\imlvar{x} \mapsto \imlvar{b}] , \imlprocessp), \imlmultisetprocs} }
		 	\end{prooftree}
		 	\\[10pt] 
		 		\begin{prooftree}
		 		\hypo{ \begin{matrix}
		 				\imleval{\imlvar{e}}_{\imlenv} = \text{true}
		 		\end{matrix}}
		 		\infer1[]{ \imltrans{}{1}{(\imlenv , \imlc{if} \ \imlvar{e} \ \imlc{then} \ \imlprocessp \ \imlc{else} \ \imlprocessp'), \imlmultisetprocs}{(\imlenv , \imlprocessp), \imlmultisetprocs} }
		 	\end{prooftree}
		 	\\[10pt] 
		 	\begin{prooftree}
		 		\hypo{ \begin{matrix}
		 				\imleval{\imlvar{e}}_{\imlenv} = \text{false}
		 		\end{matrix}}
		 		\infer1[]{ \imltrans{}{1}{(\imlenv ,\imlc{if} \ \imlvar{e} \ \imlc{then} \ \imlprocessp \ \imlc{else} \ \imlprocessp'), \imlmultisetprocs}{(\imlenv , \imlprocessp'), \imlmultisetprocs} }
		 	\end{prooftree}
\\[10pt] 
\begin{prooftree}
	\hypo{\type = \texttt{fixed}_n\ \text{for some}\ n \in \naturalnum}
	\hypo{|\imlvar{b}| = n}
	\infer2[]{ \imltrans{\imlfreshev{\imlvar{b}}}{\probdist{n}}{(\imlenv,  \imlc{new} \ \imlvar{x} : \type; \imlprocessp), \imlmultisetprocs}{ (\imlenv[\imlvar{x} \mapsto \imlvar{b}],  \imlprocessp), \imlmultisetprocs}}
\end{prooftree}
\\[10pt] 
 \begin{prooftree}
 		\hypo{\forall j \leq m : \imleval{\imlvar{d_j}}_{\imlenv} = \imlvar{b_j} \neq \bot}
 		\infer1[]{ \imltrans{\imlev{\imlvar{b_1}, \dots , \imlvar{b_m}}}{1}{(\imlenv,\imlc{event}(\imlvar{d_1},\dots,\imlvar{d_m}); \imlprocessp), \imlmultisetprocs}{ (\imlenv, \imlprocessp), \imlmultisetprocs}}
 \end{prooftree}
\\[10pt] 
\begin{prooftree}
	\hypo{ \begin{matrix}
			\imleval{\imlvar{e}}_{\imlenv} = \imlvar{b} \neq \imlc{\bot}  \quad	\imlvar{b}' = \mathit{truncate}(\imlvar{b}, \mathit{maxlen}(\imlchannel)) \\[5pt]
			\forall j \leq m : \imleval{\imlvar{e_j}}_{\imlenv} = \imlvar{b_j} \neq \imlc{\bot}  \quad \imlmultisetprocs' =  \mathit{reduce}(\{ (\imlenv, \imlprocessq) \}) \\[5pt]
			\exists ! (\nextimlenv, \imlprocessq') \in \imlmultisetprocs	 : \imlprocessq' = \imlc{in}(\imlchannel[\imlvar{e{\rmcolor{'}}_1},\dots,\imlvar{e{\rmcolor{'}}_m}], \imlvar{x}');\imlprocessp'   \land   \forall j \leq m : \imleval{\imlvar{e_j}'}_{\nextimlenv} = \imlvar{b_j} \neq \imlc{\bot}
	\end{matrix}}			
	\infer1[]{ \imltrans{}{1}{(\imlenv, \imlc{out}(\imlchannel[\imlvar{e_1},\dots,\imlvar{e_m}], \imlvar{e});\imlprocessq), \imlmultisetprocs}{(\nextimlenv[\imlvar{x}' \mapsto \imlvar{b}'], \imlprocessp'), \imlmultisetprocs \uplus \imlmultisetprocs' \setminus \{ (\nextimlenv, \imlprocessq') \} }}
\end{prooftree}
	}	
	\caption{The semantics of $ \imlsymb $~\cite[p.~23]{aizatulin2015verifying}.}
	\label{fig:imlsemantics}
	\vspace{-1em}
\end{figure}

\begin{figure}[t]
	\begin{minipage}[b]{0.45\linewidth}
		\centering
		\adjustbox{varwidth=\linewidth,scale=0.9}{%
			\begin{tcolorbox}[
				titlerule=1pt, boxsep=0pt, colframe=gray, halign=left, valign=center,
				left=0pt, right=0pt, top=0pt, bottom=0pt,  lower separated=false,
				title={\hspace{0.7em}Simplified C Code}
				]
				\begin{lstlisting}[style=cstyle]
do { 
 k = receive();
 if (!k) assert(false); 
} while (!newkey(k));\end{lstlisting}
			\end{tcolorbox}
		}
		\caption{A simplified loop from the TinySSH binary.}
		\label{fig:loop-C}
	\end{minipage}
	\hfill
	\begin{minipage}[b]{0.5\linewidth}
		\adjustbox{varwidth=\linewidth,left,scale=0.75}{%
			\tikzset{every picture/.style={line width=0.75pt}} %
		\begin{tikzpicture}[draw, minimum size=5ex, x=0.75pt,y=0.75pt,yscale=-1,xscale=1]
			
			\node[circle,draw](a){$ a $} [grow = up]
			child{node[circle,draw] (c){$ c $}}
			child{
				node[circle,draw] (b){$ b $}
				child {
					node[circle,draw] (d){$ d $}
					child {
						node[circle,draw] (f){$ f $}
						child {node[rectangle,draw] (g){end}}
					}
					child {node[circle,draw] (e){$ e $}	
					}
				}
			};
			
			\node[rectangle,draw,above=0.5cm of a] (start) {start};
			\draw[->,>=angle 60,thick] (start) -- node[anchor=west, xshift=0.1cm] {$\var{k}$} (a);
			\draw[->,>=stealth, thick] (e.west) to [bend right=80,looseness=1.5] (a.west) ;
			\draw[->,>=angle 60,thick] (a) -- node[anchor=east] {$ ! \bnfc{newkey} (\var{k}) $} (b);
			\draw[->,>=angle 60,thick] (a) -- node[anchor=west] {$ \bnfc{newkey} (\var{k}) $} (c);
			\draw[->,>=angle 60,thick] (b) -- node[anchor=east] {$ \var{k} = \bnfc{receive}() $} (d);
			\draw[->,>=angle 60,thick] (d) -- node[anchor=east, xshift=-0.8cm] {$\var{k} \neq 0$} (f);
			\draw[->,>=angle 60,thick] (d) -- node[anchor=west, xshift=0.8cm] {$\var{k} = 0$} (e);
			\draw[->,>=angle 60,thick] (f) -- node[anchor=west] {$ \bnfc{assert}(false) $} (g);
		\end{tikzpicture}
	}
		\caption{Flowgraph of the example loop.}
		\label{fig:loop-tree}
	\end{minipage}
\end{figure}

\subsection{Loop Summarization Example}
\label{loopsummary}
Our symbolic execution engine handles \emph{bounded loops} using the summarization technique of Strej\v{c}ek~\cite{abstractingpc2012}.
We show by an example how this method is used in our symbolic execution to summarize loops.
\autoref{fig:loop-C} is (the simplified version of) a loop from the TinySSH's implementation, and \autoref{fig:loop-tree} indicates its flowgraph.
Our example loop consists of nodes $ \{ a, b, d,  e \} $, with $ a $ being the entry node of the loop, and a single path $ \looppath_{\sbirc{1}} = abdea $.
The program continues to execute from node $ c $ if a new key is received.

The loop’s effect can be described by an iterated symbolic environment $ \itrsbirenv{\vec {\sbirc{t}}}$ for $ \vec {\sbirc{t}} $ iterations of the loop.
The $ \vec {\sbirc{t}} $ is a vector of path counters, which in our example, consists of a single counter $ t_1 $ that is assigned to $ \looppath_{\sbirc{1}} $.
The only variable which is modified within the loop is $ \var{k} $. In our  example, the value of $\var{k}$ is independent of the number of loop iterations and the initial values of program variables,  and it can be described by a symbolic value $ \itrsbirenv{\vec {\sbirc{t}}}[\var{k}] = \symvar{k}_{_{\vec {\sbirc{t}}}}$.

The looping condition describes the necessary condition to keep looping along the acyclic paths inside a given loop. For our example, the looping condition, $\itrsbirpcond{\vec {\sbirc{t}}} = \sbirpcond_{\sbirc{1}}$, consists of a single formula $\sbirpcond_{\sbirc{1}}$ which describes the required condition to loop along the acyclic path $\looppath_{\sbirc{1}}$.
The formula $ \sbirpcond_{\sbirc{1}} $ is based on two assertions, namely, $ ! \bnfc{newkey} (\var{k}) $ and $ \var{k} \neq 0 $.
Let's say that in the symbolic environment, the stored values of $ \var{k} $ in $ (t'_1 - 1)$-st and $ t'_1$-st iteration of the path $ \looppath_{\sbirc{1}} $  are $ \symvar{k}_{\sbirc{t'_1 - 1}} $ and $ \symvar{k}_{\sbirc{t'_1}} $, respectively.
Then, to get $ t_1 $ iterations of $ \looppath_{\sbirc{1}} $, the formula $ \sbirpcond_{\sbirc{1}} $ says that for each $ t’_1 $ satisfying $ 0 < t'_1 \leq t_1 $, $ \symvar{k}_{\sbirc{t'_1 - 1}}$ is not a new key, $ \symvar{k}_{\sbirc{t'_1}}$ is received from the channel, and it is not equal to zero:

\begin{math}
	\begin{array}{l}
		\begin{array}{l}
			\itrsbirpcond{\vec {\sbirc{t}}} \equiv \sbirpcond_{\sbirc{1}}
		\end{array}	\\[5pt]
\begin{array}{l}
	\sbirpcond_{\sbirc{1}} \equiv \forall  t'_1 . \
	( 0 < t'_1 \leq t_1 \to \\ \hspace{1.3cm} ( \ ! \bnfc{newkey} (\symvar{k}_{\sbirc{t'_1 - 1}}) \land \symvar{k}_{\sbirc{t'_1}} = \bnfc{receive}() \ \land \ \symvar{k}_{\sbirc{t'_1}} \neq 0  ) ) 
\end{array}	\\[5pt]
\end{array}
\end{math}\\

The loop summary is characterized by $ \tuple{ \itrsbirenv{\vec {\sbirc{t}}} , \itrsbirpcond{\vec {\sbirc{t}}}} $ which we attach it to the symbolic state at the entry node $ a $. 
Afterward, we proceed to execute the program from the exit point $ c $.

\subsection{\imlbir{} Semantics}
\label{app:imlbirsemantics}
\autoref{fig:mixedbiriml} presents the $ \imlsymb $-$\birsymb$ mixed execution semantics.

\begin{figure*}
  \adjustbox{varwidth=\linewidth,scale=1}{%
   \[
	\begin{array}{c}
	\begin{prooftree}
				\hypo{ \begin{matrix}
				\birpc \in \lbls_{\birc{\mathcal{N}}} \uplus \entr_{\birc{\mathcal{C}}} \uplus \entr_{\birc{\mathcal{R}}} \uplus \entr_{\birc{\mathcal{E}}} 
		\end{matrix}}
		\hypo{ \begin{matrix}
				\birtrans{\birevent}{}{\tuple{ {\birenv}, \birpc}}{\tuple{ {\birenv}', \birpc'}} 
		\end{matrix}}
		\infer2[normal]{\birprog \vdash \imltrans{\birevent}{1}{\tuple{\birenv, \birpc}, \mixmultisetprocsimlbir}{\tuple{ {\birenv}', \birpc'}, {\mixmultisetprocsimlbir}}}
	\end{prooftree}
     
    \\[20pt]
    \begin{prooftree}
	   \hypo{ \begin{matrix}
			\forall j \leq m : \imleval{\imlvar{y_j}}_{\imlenv} = \imlvar{b_j} \neq \imlc\bot  \quad
			\forall j \leq m : (\birenv_{\birc{j+1}} , \var{a_j})= \update(\birenv_{\birc{j}}, \heap, \memsymb, \imlvar{b_j},128)\\
	\end{matrix}}
	   \infer1[run]{\birprog \vdash \imltrans{}{1}{(\imlenv,  \imlc{run}( \birpc, (\imlvar{y_1},\dots,\imlvar{y_m}))), {\mixmultisetprocsimlbir}}{\tuple{\birenv_{\birc{m+1}}[\reg{0} \mapsto \var{a_1},\dots,\reg{m-1} \mapsto \var{a_m}], \birpc}, {\mixmultisetprocsimlbir} }}
    \end{prooftree}
   
   \\[25pt]
   \begin{prooftree}
	\hypo{ \begin{matrix}
			\imleval{\imlvar{e}}_{\imlenv} = \imlvar{b} \neq \imlc\bot  \quad \imlvar{b}' = \mathit{truncate}(\imlvar{b}, \mathit{maxlen}(\imlchannel)) \quad \forall j \leq m : \imleval{\imlvar{e_j}}_{\imlenv} = \imlvar{b_j} \neq \imlc\bot  \quad {\mixmultisetprocsimlbir}' =  \mathit{reduce}(\{ (\imlenv,  \imlprocessq) \}) \\[5pt] 
			\exists ! (\tuple{\birenv, \birpc}, \imlchannel[{\var{e_1}}',\dots,{\var{e_m}}']) \in {\mixmultisetprocsimlbir}	 : \birpc \in \entr_{\birc{\attacker_r}} \ \land	\ \forall j \leq m : \birenv[{\var{e_j}}']\! =\! \imlvar{b_j} \neq \imlc\bot   \\[5pt] 
			(\birenv{'},\var{a})\!=\! \update(\birenv, \heap_{\birc{\mathcal{A}}}, \advmem, \imlvar{b}', 128) \quad
			\birpc'\! =\! \pcinc{\birenv, \birpc} \quad {\mixmultisetprocsimlbir}''\! =\!  \{ (\tuple{\birenv, \birpc}, \imlchannel[{\var{e_1}}',\dots,{\var{e_m}}']) \}
	\end{matrix}}
	   \infer1[$ \imlc{I}to\birc{B}$]{\birprog \vdash \imltrans{\Input{\imlvar{b}',({\var{e_1'}},\dots,{\var{e_m'}})}}{1}{(\imlenv,  \imlc{out}(\imlchannel[\imlvar{e_1},\dots,\imlvar{e_m}], \imlvar{e});\imlprocessq), {\mixmultisetprocsimlbir}}{(\tuple{ \birenv{'}[\reg{0} \mapsto \var{a}], \birpc'}, \imlchannel[{\var{e_m}}',\dots,{\var{e_m}}']), {\mixmultisetprocsimlbir} \uplus {\mixmultisetprocsimlbir}' \setminus {\mixmultisetprocsimlbir}'' }}
    \end{prooftree}

    \\[40pt]
    \begin{prooftree}
	   \hypo{ \begin{matrix}
			\birpc \in \entr_{\birc{\attacker_s}} \quad \birpc' = \pcinc{\birenv, \birpc} \quad  {\mixmultisetprocsimlbir}' =  \{ (\birenv, \birpc') \} \quad \var{a} = \birenv[\reg{0}]	\quad
	\imlvar{b} = \memload(\birenv,\var{a}) \quad \forall j \leq m : \birenv[\var{e_j}] = \imlvar{b_j} \neq \imlc\bot   \\[5pt] 
			\exists ! (\imlenv,  \imlprocessq)\! \in\! {\mixmultisetprocsimlbir}	 \!:\! \imlprocessq\! =\! \imlc{in}(\imlchannel[{\imlvar{e_1}}',\dots,{\imlvar{e_m}}'], \imlvar{x});\imlprocessp  \land \forall j \leq m : \imleval{{\imlvar{e_j}}'}_{\imlenv} = \imlvar{b_j} \neq \imlc\bot \quad \imlvar{b}' = \mathit{truncate}(\imlvar{b}, \mathit{maxlen}(\imlchannel))
	\end{matrix}}
	   \infer1[$ \birc{B}to\imlc{I}$]{\birprog \vdash \imltrans{\Output{\imlvar{b},(\var{e_1},\dots,\var{e_m})}}{1}{(\tuple{\birenv, \birpc}, \imlchannel[\var{e_1},\dots,\var{e_m}]), {\mixmultisetprocsimlbir}}{(\imlenv[\imlvar{x} \mapsto \imlvar{b}'],  \imlprocessp), {\mixmultisetprocsimlbir} \uplus {\mixmultisetprocsimlbir}' \setminus \{ (\imlenv,   \imlprocessq) \} }}
    \end{prooftree}
   \end{array}
\]
}
\caption{The mixed semantics of $\birsymb$ and $ \imlsymb $ shown by \imlbir{}.}
\label{fig:mixedbiriml}
\end{figure*}

\subsection{Soundness of Translation into  $ \imlsymb $ }
\label{sound-iml-trans}
\mixsbirstateeventeqstepzero*

\begin{proof}
	By construction of our symbolic tree we know that  for $ \birpc_{\birc{0}} $, we have $ \sbirtoiml{\tree} $. Thus, from  translation rule $ \sbirtoiml{\tree} = \imlprocessq^{\mathit{full}}$ in~\autoref{fig:sbirtoiml} we get that $ \imlstate{0}.\imlprocessp = \imlprocessq^{\mathit{full}} $. 
	Then, we can conclude that $ \simreltra{ \tuple{True, \sbirenv_{\sbirc{0}}[ \randommem \mapsto \randomsymvals_{\sbirc{k}} , \randommemidx \mapsto 0], \birpc_{\birc{0}}}}{\sbirvaluation}{\sbirtoiml{.}}{(\imlenv_{\imlc{0}}, \imlprocessq^{\mathit{full}})}$. 
\end{proof}

\mixsbirstateeventeqstepone*

Proof of lemma~\ref{lem:iml:stateeventeq} is done by a case split on the type of label sets in the $\birsymb$ program $ \birprog $.
\begin{proof}
	We prove the statement by a case split based on the type of the program counter $ \mixedexec{\mixediml}{\sbirstate{i}}.\birpc $. 
	\begin{itemize}
			\item $	\birpc \in \lbls_{\birc{\mathcal{E}}}$ \\
			Since $\simreltra{\mixedexec{\mixediml}{\sbirstate{i}}}{\sbirvaluation}{\sbirtoiml{.}}{\imlstate{i}}$,
			we know that $\imlstate{i}.\imlprocessp = \sbirtoiml{\tree[\sbirstate{i}.\birpc]}$.
			Based on the translation rules in~\autoref{fig:sbirtoiml}, we get that $ \imlstate{i}.\imlprocessp = \imlc{event}(\symvar{d_1},$  $\dots,\symvar{d_m});\sbirtoiml{\tree'} $. 
			Also, by construction of the symbolic execution tree we know that there exist $\tree_\sbirc{0}$ and $\tree_\sbirc{1}$ such that $\tree = \tree_\sbirc{0} \bnfconcat(\birpc, \symevent{\symvar{d_1},\dots,$  $\symvar{d_m}})\bnfconcat(\birpc', \nodeevent)\bnfconcat\tree_\sbirc{1}$.
			Therefore, based on the operational semantics of $ \imlsymb $ processes~\cite[p.~23]{aizatulin2015verifying} and translation rules in~\autoref{fig:sbirtoiml}, we get $ \imltrans{\imlev{\imlvar{b_1},\ldots,\imlvar{b_{m}}}}{1}{\imlstate{i}}{\imlstate{j}}$ and $\imlstate{j} = (\imlenv,  \sbirtoiml{\tree[\sbirstate{j}.\birpc']}), \imlmultisetprocs $.
				
			Based on the semantics of \imlsbir{},~\autoref{fig:mixedsymbiml}, we find $ \mixedexec{\mixediml}{\sbirstate{j}} = \tuple{\sbirpcond,\sbirenv, \birpc'}, \mixmultisetprocsimlsbir$ that is in the relation $\miximltrans{\miximlsbirevent}{1}{\sbirvaluation}{\mixedexec{\mixediml}{\sbirstate{i}}}{\mixedexec{\mixediml}{\sbirstate{j}}}$ such that $\miximlsbirevent = \symevent{\symvar{d_1},\dots,\symvar{d_m}} $.
			
			Finally, we choose $\sbirvaluation'$ by extending $\sbirvaluation$ according to the executed operation and such that $  \sbirvaluation'(\symvar{d_j})=\!\imlvar{b_j} $ for $ 0 < j \leq m $. 
			Therefore, we can conclude that $ \simreltra{\tuple{\sbirpcond,\sbirenv, \birpc'}, \mixmultisetprocsimlsbir}{\sbirvaluation'}{\sbirtoiml{.}}{(\imlenv,  \sbirtoiml{\tree[\sbirstate{i}.\birpc']}), \imlmultisetprocs}$ and  $\symevent{\symvar{d_1},\dots,\symvar{d_m}} =_{\sbirvaluation'} \imlev{\imlvar{b_1},\ldots,\imlvar{b_{m}}}$.\\
		
			\item $	\birpc \in \lbls_{\birc{\mathcal{C}}}$ \\
			Since $\simreltra{\mixedexec{\mixediml}{\sbirstate{i}}}{\sbirvaluation}{\sbirtoiml{.}}{\imlstate{i}}$,
			we know that $\imlstate{i}.\imlprocessp = \sbirtoiml{\tree[\sbirstate{i}.\birpc]}$.
			Based on the translation rules in~\autoref{fig:sbirtoiml}, we get that $ \imlstate{i}.\imlprocessp =\imlc{let}\ \imlvar{x} \ \imlc= \ \symvar{v}\ $  $\imlc{in}\ \sbirtoiml{\tree'} $. Also, by construction of the symbolic execution tree we know that there exist $\tree_\sbirc{0}$ and $\tree_\sbirc{1}$ such that $\tree = \tree_\sbirc{0}\bnfconcat(\birpc, \symcrypto{\symvar{v}})\bnfconcat(\birpc', \nodeevent)\bnfconcat\tree_\sbirc{1}$.
			
			Therefore, based on the operational semantics of $ \imlsymb $ processes~\cite[p.~23]{aizatulin2015verifying} and translation rules in~\autoref{fig:sbirtoiml}, we get that $ \imltrans{}{1}{\imlstate{i}}{\imlstate{j}}$ and $\imlstate{j} = (\imlenv[\imlvar{x} \mapsto \imlvar{b}],  \sbirtoiml{\tree[\sbirstate{j}.\birpc']}), \imlmultisetprocs $.
			Based on the semantics of \imlsbir{},~\autoref{fig:mixedsymbiml}, we find $ \mixedexec{\mixediml}{\sbirstate{j}} = \tuple{\sbirpcond, {\sbirenv}{''}, \birpc'}, $  $\mixmultisetprocsimlsbir$ that is in the transition relation $\miximltrans{\miximlsbirevent}{1}{\sbirvaluation}{\mixedexec{\mixediml}{\sbirstate{i}}}{\mixedexec{\mixediml}{\sbirstate{j}}}$ such that $\miximlsbirevent = \symcrypto{\symvar{v}}$ and $ \sbirenv{''} = \sbirenv{'}[\reg{0} \mapsto \var{a}] $ where $(\sbirenv{'},\var{a})=\update(\sbirenv, \heap_{\mathit{\cryptoOp}}, \memsymb_{\mathit{\cryptoOp}}, \symvar{v}, 128)$.
			
			Finally, we choose $\sbirvaluation'$ by extending $\sbirvaluation$ according to the executed operation and such that $  \sbirvaluation'(\symvar{v})=\!\imlvar{b}$.
			Therefore, we can conclude that $ \simreltra{\tuple{\sbirpcond, \sbirenv{'}[\reg{0} \mapsto \var{a}], \birpc'}, \mixmultisetprocsimlsbir}{\sbirvaluation'}{\sbirtoiml{.}}{(\imlenv[\imlvar{x} \mapsto \imlvar{b}],  \sbirtoiml{\tree[\sbirstate{i}.\birpc']}), \imlmultisetprocs}$. \\

			\item $	\birpc \in \lbls_{\birc{\mathcal{R}}}$ \\
			Since $\simreltra{\mixedexec{\mixediml}{\sbirstate{i}}}{\sbirvaluation}{\sbirtoiml{.}}{\imlstate{i}}$,
			we know that $\imlstate{i}.\imlprocessp = \sbirtoiml{\tree[\sbirstate{i}.\birpc]}$.
			Based on the translation rules in~\autoref{fig:sbirtoiml}, we get that $ \imlstate{i}.\imlprocessp = \imlc{new}\ \symvar{x};\ $  $ \sbirtoiml{\tree'} $. Also, by construction of the symbolic execution tree we know that there exist $\tree_\sbirc{0}$ and $\tree_\sbirc{1}$ such that $\tree = \tree_\sbirc{0}\bnfconcat(\birpc, \symfreshv{\symvar{x}})\bnfconcat(\birpc', \nodeevent)\bnfconcat\tree_\sbirc{1}$.
			
			Therefore, based on the operational semantics of $ \imlsymb $ processes~\cite[p.~23]{aizatulin2015verifying} and translation rules in~\autoref{fig:sbirtoiml}, we get that $ \imltrans{\imlfreshev{\imlvar{b}}}{\probdist{n}}{\imlstate{i}}{\imlstate{j}}$ and $\imlstate{j} = (\imlenv[\symvar{x} \mapsto \imlvar{b}],  \sbirtoiml{\tree[\sbirstate{j}.\birpc']}), \imlmultisetprocs $.
			Based on the semantics of \imlsbir{},~\autoref{fig:mixedsymbiml}, we find $ \mixedexec{\mixediml}{\sbirstate{j}} = \tuple{\sbirpcond, {\sbirenv}{''}, \birpc'}, $  $\mixmultisetprocsimlsbir$ that is in the relation $\miximltrans{\miximlsbirevent}{1}{\sbirvaluation}{\mixedexec{\mixediml}{\sbirstate{i}}}{\mixedexec{\mixediml}{\sbirstate{j}}}$ such that $\miximlsbirevent = \symfreshv{\symvar{x}} $ and $ \sbirenv{''} = {\sbirenv{'}}[\reg{0} \mapsto \var{a}; \randommemidx \mapsto \sbirenv[\randommemidx] + {l}]  $  where $ (\sbirenv{'},\var{a}) = \update(\sbirenv, \heap, \memsymb, \symvar{x},128) $ and $ \symvar{x} = \sbirrand(\sbirenv, n) $. 
			
			Finally, we choose $\sbirvaluation'$ by extending $\sbirvaluation$ according to the executed operation and such that $  \sbirvaluation'(\symvar{x})=\!\imlvar{b}$.
			Therefore, we can conclude that $ \simreltra{\tuple{\sbirpcond, \sbirenv{'}[\reg{0} \mapsto \var{a}; \randommemidx \mapsto \sbirenv[\randommemidx] + {l}], \birpc'}, \mixmultisetprocsimlsbir}{\sbirvaluation'}{\sbirtoiml{.}}{(\imlenv[\symvar{x} \mapsto \imlvar{b}],  \sbirtoiml{\tree[\sbirstate{i}.\birpc']}), \imlmultisetprocs}$ and \\
			 $\imlfreshev{\imlvar{b}} =_{\sbirvaluation'} \symfreshv{\symvar{x}}$. \\
			
			\item $	\birpc_\birc{1} \in \lbls_{\birc{\attacker_s}} \ \land \ \birpc_\birc{2} \in \lbls_{\birc{\attacker_r}} $ \\
			In this case, since $\birpc_\birc{1} \in \lbls_{\birc{\attacker_s}}$, we have an $ \imlsymb $ output process $\imlc{out}$ in the $ \imlsymb $ state $ \imlstate{i} $. 
			Based on IOut rule in the operational semantics of $ \imlsymb $ processes~\cite[p.~23]{aizatulin2015verifying}, we get that there is an $ \imlsymb $ input process $\imlc{in}$ in the multiset of executing input processes that is ready to receive input from the channel. Therefore, there exists a node in our symbolic execution tree such that its program counter is in the label set $ \lbls_{\birc{\attacker_r}} $.
			
			More formally, since $\simreltra{\mixedexec{\mixediml}{\sbirstate{i}}}{\sbirvaluation}{\sbirtoiml{.}}{\imlstate{i}}$,
			we know that $\imlstate{i}.\imlprocessp = \sbirtoiml{\tree[\sbirstate{i}.\birpc_\birc{1}]}$.
			Based on the translation rules in~\autoref{fig:sbirtoiml}, we get that $ \imlstate{i}.\imlprocessp = \imlc{out}(\imlchannel[\symvar{e_1},\dots,\symvar{e_m}], \symvar{e});$ $\sbirtoiml{\tree'}$. 
			Based on the operational semantics of $ \imlsymb $ processes~\cite[p.~23]{aizatulin2015verifying}, we get that $ \imltrans{}{1}{\imlstate{i}}{\imlstate{j}}$ and $\imlstate{j} = (\nextimlenv[\symvar{v} \mapsto \imlvar{b}'],  \imlprocessp'), \imlmultisetprocs \uplus \imlmultisetprocs' \setminus \imlmultisetprocs'' $ such that $\imlmultisetprocs'' = \{ (\nextimlenv,  \imlc{in}(\imlchannel[{\symvar{e_1}}',\dots,{\symvar{e_m}}'], \symvar{v});\imlprocessp') \} $.

			Therefore, there exist $\tree_\sbirc{0}$ and $\tree_\sbirc{1}$ such that $\tree' = \tree_\sbirc{0}\bnfconcat(\birpc_\birc{2},\syminput{\symvar{v},$  $({\symvar{e_1}}',\dots, {\symvar{e_m}}')} )\bnfconcat(\birpc', \nodeevent)\bnfconcat\tree_\sbirc{1}$ and it holds that $\imlc{in}(\imlchannel[{\symvar{e_1}}',\dots,$  ${\symvar{e_m}}'], \symvar{v});\imlprocessp' = \sbirtoiml{\syminput{\symvar{v},({\symvar{e_1}}',\dots,{\symvar{e_m}}')}}; \sbirtoiml{\nodeevent}$.
			
			Since we have an event node $(\birpc_\birc{2},\syminput{\symvar{v},({\symvar{e_1}}',\dots,{\symvar{e_m}}')} )$, we get that there exist an intermediate state $\mixedexec{\mixediml}{\sbirstate{l}} \in \mixmultisetprocsimlsbir $ such that $\mixedexec{\mixediml}{\sbirstate{l}} = \{(\tuple{\sbirpcond, {\sbirenv}', \birpc_\birc{2}}, \imlchannel[{\symvar{e_1}}'$  $,\dots,{\symvar{e_m}}'])\}$ and $ \birpc_\birc{2} \in \lbls_{\birc{\attacker_r}}$.
			
			Based on the semantics of \imlsbir{},~\autoref{fig:mixedsymbiml}, we find $ \mixedexec{\mixediml}{\sbirstate{j}} = (\tuple{\sbirpcond, {\sbirenv}{'''}, \birpc'}, \imlchannel[{\symvar{e_1}}',\dots,{\symvar{e_m}}']), \mixmultisetprocsimlsbir  \uplus {\mixmultisetprocsimlsbir}' \setminus {\mixmultisetprocsimlsbir}''$ that is in the transition relation $\miximltransmulti{\miximlsbirevent}{1}{\sbirvaluation}{\mixedexec{\mixediml}{\sbirstate{i}}}{\mixedexec{\mixediml}{\sbirstate{j}}}$ such that $\miximlsbirevent = \symoutput{\symvar{e}, (\symvar{e_1},\dots,\symvar{e_m})};\dots;\syminput{\symvar{v},({\symvar{e_1}}',\dots,{\symvar{e_m}}')}$ and $ \sbirenv{'''} = \sbirenv{''}$  $[\reg{0} \mapsto \var{a}] $ where $(\sbirenv{''},\var{a})=\update(\sbirenv{'}, \heap_{\birc{\mathcal{A}}}, \advmem, \symvar{v}, 128)$ and $ \birpc' = \pcinc{\sbirpcond, \sbirenv{'}, \birpc_\birc{2}}$
			
			Finally, we choose $\sbirvaluation'$ by extending $\sbirvaluation$ according to the executed operation and such that $  \sbirvaluation'(\symvar{v})=\!\imlvar{b}'$.
			Thus, we can conclude that $ \simreltra{(\tuple{\sbirpcond, {\sbirenv}''[\reg{0} \mapsto \var{a}], \birpc'}, \imlchannel[{\symvar{e_1}}',\dots,{\symvar{e_m}}']), \mixmultisetprocsimlsbir  \uplus {\mixmultisetprocsimlsbir}' \setminus {\mixmultisetprocsimlsbir}''}{\sbirvaluation'}{\sbirtoiml{.}}{(\nextimlenv[\symvar{v} \mapsto \imlvar{b}'],  \sbirtoiml{\tree[\sbirstate{i}.\birpc']}), \imlmultisetprocs \uplus \imlmultisetprocs' \setminus \imlmultisetprocs'' }$. \\
			
			\item $	\birpc \in \lbls_{\birc{\mathcal{N}}} \ \land \ \birpc \notin \lbls_{\birc{\mathcal{L}}}$ \\
			We show this case by a case split based on $ \birprog(\mixedexec{\mixediml}{\sbirstate{i}}.\birpc) $. 
			\begin{itemize}
				\item $ \birprog(\mixedexec{\mixediml}{\sbirstate{i}}.\birpc) \neq \branchingnode(\nodecond, \tree_\sbirc{1}, \tree_\sbirc{2})$\\
				Since $\simreltra{\mixedexec{\mixediml}{\sbirstate{i}}}{\sbirvaluation}{\sbirtoiml{.}}{\imlstate{i}}$,
				we know that $\imlstate{i}.\imlprocessp = \sbirtoiml{\tree[\sbirstate{i}.\birpc]}$.
				By construction of the symbolic execution tree we know that there exist $\tree_\sbirc{0}$ and ${\tree_\sbirc{1}}'$ such that $\tree = \tree_\sbirc{0}\bnfconcat(\birpc, \sbirc\tau)\bnfconcat(\birpc', \nodeevent)\bnfconcat{\tree_\sbirc{1}}'$.
				
				Therefore, based on the operational semantics of $ \imlsymb $ processes~\cite[p.~23]{aizatulin2015verifying} and translation rules in~\autoref{fig:sbirtoiml}, we get that $\imlstate{j} = ({\imlenv}',  \sbirtoiml{\tree[\sbirstate{j}.\birpc']}), \imlmultisetprocs $ and $ \imltrans{}{1}{\imlstate{i}}{\imlstate{j}}$.
				Based on the semantics of \imlsbir{},~\autoref{fig:mixedsymbiml}, we find $ \mixedexec{\mixediml}{\sbirstate{j}} = \tuple{ \sbirpcond', {\sbirenv}', \birpc'}, \mixmultisetprocsimlsbir$ that is in the relation $\miximltrans{}{1}{\sbirvaluation}{\mixedexec{\mixediml}{\sbirstate{i}}}{\mixedexec{\mixediml}{\sbirstate{j}}}$.

				Finally, we choose $\sbirvaluation'$ by extending $\sbirvaluation$ according to the executed operation and such that $\sbirvaluation'(\mixedexec{\mixediml}{\sbirstate{j}}.\sbirenv)= \imlstate{j}.\imlenv$.
				Therefore, we can conclude that $ \simreltra{\tuple{ \sbirpcond', {\sbirenv}', \birpc'}, \mixmultisetprocsimlsbir}{\sbirvaluation'}{\sbirtoiml{.}}{({\imlenv}',  $  $\sbirtoiml{\tree[\sbirstate{i}.\birpc']}), \imlmultisetprocs }$. \\
				
				\item $ \birprog(\mixedexec{\mixediml}{\sbirstate{i}}.\birpc) = \branchingnode(\nodecond, \tree_\sbirc{1}, \tree_\sbirc{2})$ case $ True $ (case $False$)\\
				Since $\simreltra{\mixedexec{\mixediml}{\sbirstate{i}}}{\sbirvaluation}{\sbirtoiml{.}}{\imlstate{i}}$,
				we know that $\imlstate{i}.\imlprocessp = \sbirtoiml{\tree[\sbirstate{i}.\birpc]}$.
				Based on the translation rules in~\autoref{fig:sbirtoiml}, we get that $ \imlstate{i}.\imlprocessp = \imlc{if} \ \sbirtoiml{\nodecond} \ \imlc{then} \ \sbirtoiml{\tree_\sbirc{1}} \ \imlc{else}\ \sbirtoiml{\tree_\sbirc{2}}$. 
				Also, by construction of the symbolic execution
				tree we know that there exist $\tree_\sbirc{0}$, ${\tree_\sbirc{1}}'$, and
				${\tree_\sbirc{2}}'$ s.t.
				$\tree = \tree_\sbirc{0}\bnfconcat
				(\birpc, \branchingnode(\nodecond, ((\birpc', -)\bnfconcat{\tree_\sbirc{1}}'),
				\tree_\sbirc{2}))$ 
				($\tree = \tree_\sbirc{0}\bnfconcat
				(\birpc, \branchingnode(\nodecond, \tree_\sbirc{1},
				((\birpc'',-)\bnfconcat{\tree_\sbirc{2}}')))$).
				
				Therefore, based on the operational semantics of $ \imlsymb $ processes~\cite[p.~23]{aizatulin2015verifying} and translation rules in~\autoref{fig:sbirtoiml}, we get that $ \imlstate{j} = (\imlenv,  \sbirtoiml{\tree[\sbirstate{j}.\birpc']}), \imlmultisetprocs$ 
				$ \left( \imlstate{j} = (\imlenv,  \sbirtoiml{\tree[\sbirstate{j}.\birpc'']}), \imlmultisetprocs \right)$ and $ \imltrans{}{1}{\imlstate{i}}{\imlstate{j}}$.
				
				Based on the semantics of \imlsbir{},~\autoref{fig:mixedsymbiml}, we find 
				$ \sbirstate{j} = \tuple{\sbirpcond, {\sbirenv}, \birpc'}$
				$ \left(\sbirstate{j} = \tuple{\sbirpcond, {\sbirenv}, \birpc''}\right)$ that is in the transition relation $\miximltrans{}{1}{\sbirvaluation}{\mixedexec{\mixediml}{\sbirstate{i}}}{\mixedexec{\mixediml}{\sbirstate{j}}}$.
				Therefore, we can conclude that $ \simreltra{\tuple{ \sbirpcond, {\sbirenv}, \birpc'}, \mixmultisetprocsimlsbir}{\sbirvaluation'}{\sbirtoiml{.}}{(\imlenv,  \sbirtoiml{\tree[\sbirstate{i}.\birpc']}), \imlmultisetprocs }\\
				\left( \simreltra{\tuple{ \sbirpcond, {\sbirenv}, \birpc''}, \mixmultisetprocsimlsbir}{\sbirvaluation'}{\sbirtoiml{.}}{(\imlenv,  \sbirtoiml{\tree[\sbirstate{i}.\birpc'']}), \imlmultisetprocs } \right)$. \\
			\end{itemize}

	\item $ \birprog(\mixedexec{\mixediml}{\sbirstate{i}}.\birpc) \in \lbls_{\mathcal{L}}$\\
	Since $\simreltra{\mixedexec{\mixediml}{\sbirstate{i}}}{\sbirvaluation}{\sbirtoiml{.}}{\imlstate{i}}$,
	we know that $\imlstate{i}.\imlprocessp = \sbirtoiml{\tree[\sbirstate{i}.\birpc]}$.
	Based on the translation rules in~\autoref{fig:sbirtoiml}, we get that $ \imlstate{i}.\imlprocessp = \imlc{!}^{\symvar{t}\leq \imlc{m}} $  $\loopbody{\sbirstate{i}.\birpc}$. 
	Also, by construction of the symbolic execution
	tree we know that there exist $\tree_\sbirc{0}$, and
	$\tree_\sbirc{1}$ such that
	$\tree = \tree_\sbirc{0}\bnfconcat 
	(\birpc, \symloop{\symvar{t}})\bnfconcat(\birpc', \nodeevent)\bnfconcat\tree_\sbirc{1}$ such that $ \birpc' = \pcexit{\birpc} $. 
	
	Therefore, based on the operational semantics of $ \imlsymb $ processes~\cite[p.~23]{aizatulin2015verifying} and translation rules in~\autoref{fig:sbirtoiml}, we get  
	$ \imlstate{j} = (\nextimlenv,  \sbirtoiml{\tree[\sbirstate{j}.\birpc']}),\imlmultisetprocs $ and $ \imltrans{}{1}{\imlstate{i}}{\imlstate{j}}$ such that $ \nextimlenv[\symvar{t}] = \imlvar{b} $.
	Based on the semantics of \imlsbir{},~\autoref{fig:mixedsymbiml} and the computed loop summary, we find 
	$ \mixedexec{\mixediml}{\sbirstate{j}} = \tuple{{\sbirpcond}', {\sbirenv}', \birpc'}, \mixmultisetprocsimlsbir$ that is in the transition relation $\miximltrans{\miximlsbirevent}{1}{\sbirvaluation}{\mixedexec{\mixediml}{\sbirstate{i}}}{\mixedexec{\mixediml}{\sbirstate{j}}}$, $\miximlsbirevent = \symloop{\symvar{t}} $, and $ \birpc' = \pcexit{\birpc} $.

	Finally, for every $ \imlsymb $ variable $ \imlvar{x} $ and symbolic $\birsymb$ variable $ \symvar{x} $ that $ \mixedexec{\mixediml}{\!\sbirstate{j}\!}.\sbirenv[\symvar{x}] \neq \mixedexec{\mixediml}{\!\sbirstate{i}\!}.\sbirenv[\symvar{x}]$ and $ \imlstate{j}.\imlenv[\imlvar{x}] \neq \imlstate{i}.\imlenv[\imlvar{x}]$, we choose $\sbirvaluation'$ by extending $\sbirvaluation$ according to the executed operation such that $  \sbirvaluation'(\mixedexec{\mixediml}{\!\sbirstate{j}\!}.\sbirenv[\symvar{x}])= \imlstate{j}.\imlenv[\imlvar{x}]$. 
	Moreover, for the symbolic counter $ \symvar{t} $ which represents the number of iterations of the loop in symbolic execution and the number of the process replication $ \imlvar{b} $, we have $ \sbirvaluation' (\symvar{t}) = \imlvar{b}$.
	Therefore, we can conclude that $ \simreltra{\tuple{{\sbirpcond}', {\sbirenv}', \birpc'}, \mixmultisetprocsimlsbir}{\sbirvaluation'}{\sbirtoiml{.}}{(\nextimlenv,  \sbirtoiml{\tree[\sbirstate{j}.\birpc']}), \imlmultisetprocs }$. 
	\end{itemize}
\end{proof}

\subsection{Soundness of Symbolic Execution}
\label{mixedbirtomixedsbirsoundness}
 
\mixbirstateeventeqstepzero*
 
 \begin{proof}
 	We choose $\sbirvaluation$ such that for $ 0 < i \leq k $, $ \sbirvaluation( \mixedexec{\mixediml}{\sbirstate{0}}.\sbirenv.\randommem[\symvar{x_i}]) $  $
 	= \mixedexec{\mixediml}{\birstate{0}}.\birenv.\randommem[\var{x_i}] $.
 	Finally, we can conclude that $ \simrel{ \tuple{\birenv_\birc{0}[ \randommem \mapsto \randomvals_\imlc{k} , \randommemidx $  $\mapsto 0], \birpc_\birc{0}}}{\sbirvaluation}{ \tuple{True, \sbirenv_\sbirc{0}[ \randommem \mapsto \randomsymvals_\sbirc{k} , \randommemidx \mapsto 0], \birpc_\birc{0}}}$. 
 \end{proof}

\mixbirstateeventeqstepone*

Proof of lemma~\ref{lem:bir:stateeventeq} is done by a case split on the type of label sets in the $\birsymb$ program $ \birprog $.
\begin{proof}
	We prove the statement by a case split based on the type of the program counter $ \mixedexec{\mixediml}{\birstate{i}}.\birpc $. 
	\begin{itemize}
		\item $	\birpc \in \lbls_{\birc{\mathcal{E}}}$ \\
		Based on the mixed semantics of $ \imlsymb $ and $ \birsymb $,~\autoref{fig:mixedbiriml}, we get that $ \mixedexec{\mixediml}{\birstate{j}} = \tuple{{\birenv}', \birpc'}, {\mixmultisetprocsimlbir}$ and $\imltrans{\miximlbirevent}{1}{\mixedexec{\mixediml}{\birstate{i}}}{\mixedexec{\mixediml}{\birstate{j}}}$.
		Based on the definition of event functions in~\autoref{ev-fun}, we have $ \birpc' =  \pcinc{{\birenv}, \birpc}  $, $ \miximlbirevent = \event{\imlvar{b_1},\ldots,\imlvar{b_m}} $ and $\birenv{'} = \birenv$.
		Since $ \simrel{\mixedexec{\mixediml}{\birstate{i}}}{\sbirvaluation}{\mixedexec{\mixediml}{\sbirstate{i}}} $,
		we get that $ \mixedexec{\mixediml}{\sbirstate{i}}.\birpc \in \lbls_{\birc{\mathcal{E}}}$.
		Based on the semantics of \imlsbir{},~\autoref{fig:mixedsymbiml}, we find $ \mixedexec{\mixediml}{\sbirstate{j}} = \tuple{\sbirpcond,\sbirenv, \birpc'}, {\mixmultisetprocsimlsbir}$ that is in the transition relation $\miximltrans{\miximlsbirevent}{1}{\sbirvaluation}{\mixedexec{\mixediml}{\sbirstate{i}}}{\mixedexec{\mixediml}{\sbirstate{j}}}$ such that $\miximlsbirevent = \symevent{\symvar{d_1},\dots,\symvar{d_m}} $.
		
		Therefore, we choose $\sbirvaluation'$ by extending $\sbirvaluation$ according to the executed operation and such that $  \sbirvaluation'(\symvar{d_j})=\!\imlvar{b_j} $ for $ 0 < j \leq m $. 
		Moreover, since the environment of both $\birstate{j}$ and $\sbirstate{j}$ are the same as $\birstate{i} $ and $ \sbirstate{i}$, respectively, we get that $ \mixedexec{\mixediml}{\sbirstate{j}}.\sbirenv =_{\sbirvaluation'} \mixedexec{\mixediml}{\birstate{j}}.\birenv $.
		Finally, we can conclude that $ \simrel{\tuple{\birenv, \birpc'}, {\mixmultisetprocsimlbir}}{\sbirvaluation'}{\tuple{\sbirpcond,\sbirenv, \birpc'}, $  ${\mixmultisetprocsimlsbir}}$ and $\event{\imlvar{b_1},\ldots,\imlvar{b_m}} =_{\sbirvaluation'} \symevent{\symvar{d_1},\dots,\symvar{d_m}}$. \\
	
		\item $	\birpc \in \lbls_{\birc{\mathcal{C}}}$ \\
		Based on the mixed semantics of $ \imlsymb $ and $ \birsymb $,~\autoref{fig:mixedbiriml}, we get that $ \mixedexec{\mixediml}{\birstate{j}} = \tuple{{\birenv}{''}, \birpc'}, {\mixmultisetprocsimlbir}$ and $\imltrans{\miximlbirevent}{1}{\mixedexec{\mixediml}{\birstate{i}}}{\mixedexec{\mixediml}{\birstate{j}}}$.
		Based on the definition of crypto calls in~\autoref{lib-fun}, we have $ \miximlbirevent = \crypto{\imlvar{v}} $, $ \birpc' =  \pcinc{{\birenv}, \birpc}  $ and $ \birenv{''} = \birenv{'}[\reg{0} \mapsto \var{a}] $ where $(\birenv{'},\var{a})=\update(\birenv, \heap_{\cryptoOp}, \memsymb_{\cryptoOp},$   $\imlvar{v}, 128)$.
		
		Since $ \simrel{\mixedexec{\mixediml}{\birstate{i}}}{\sbirvaluation}{\mixedexec{\mixediml}{\sbirstate{i}}} $,
		we get that $ \mixedexec{\mixediml}{\sbirstate{i}}.\birpc \in \lbls_{\birc{\mathcal{C}}}$.
		Based on the semantics of \imlsbir{},~\autoref{fig:mixedsymbiml}, we find $ \mixedexec{\mixediml}{\sbirstate{j}} = \tuple{\sbirpcond, {\sbirenv}{''}, \birpc'}$  $, {\mixmultisetprocsimlsbir}$ that is in the transition relation $\miximltrans{\miximlsbirevent}{1}{\sbirvaluation}{\mixedexec{\mixediml}{\sbirstate{i}}}{\mixedexec{\mixediml}{\sbirstate{j}}}$ such that $\miximlsbirevent = \symcrypto{\symvar{v}}$ and $ \sbirenv{''} = \sbirenv{'}[\reg{0} \mapsto \var{a}] $ where $(\sbirenv{'},\var{a})=\update(\sbirenv, \heap_{\cryptoOp}, \memsymb_{\cryptoOp}, $  $\symvar{v}, 128)$.
		
		Therefore, we choose $\sbirvaluation'$ by extending $\sbirvaluation$ according to the executed operation and such that $  \sbirvaluation'(\symvar{v})=\!\imlvar{v}$.
		Hence, we get that $ \mixedexec{\mixediml}{\sbirstate{j}}.\sbirenv =_{\sbirvaluation'} \mixedexec{\mixediml}{\birstate{j}}.\birenv $.
		Finally, we can conclude that $ \simrel{\tuple{\birenv{''}, \birpc'}, {\mixmultisetprocsimlbir}}{\sbirvaluation'}{ \tuple{\sbirpcond, {\sbirenv}'', \birpc'},  {\mixmultisetprocsimlsbir}}$ and $\crypto{\imlvar{v}} =_{\sbirvaluation'} \symcrypto{\symvar{v}}$. \\
	
		\item $	\birpc \in \lbls_{\birc{\mathcal{R}}}$ \\
		Based on the mixed semantics of $ \imlsymb $ and $ \birsymb $,~\autoref{fig:mixedbiriml}, we get that $ \mixedexec{\mixediml}{\birstate{j}} = \tuple{{\birenv}'', \birpc'}, {\mixmultisetprocsimlbir}$ and $\imltrans{\miximlbirevent}{1}{\mixedexec{\mixediml}{\birstate{i}}}{\mixedexec{\mixediml}{\birstate{j}}}$.
		Based on the definition of RNG in~\autoref{rng-fun}, we have $ \miximlbirevent = \freshv{\imlvar{x_d}} $, $ d = \left\lfloor \frac{\birenv[\randommemidx]}{l}  \right\rfloor +1$, $\birenv{''} = \birenv{'}[\reg{0} \mapsto \var{a}; \randommemidx \mapsto \birenv[\randommemidx] + {l}]$ where $(\birenv{'},\var{a}) = \update(\birenv, \heap, \memsymb, \imlvar{x_d},128)$, $ \imlvar{x_d} = \birrand(\birenv, n) $ and $ \birpc' =  \pcinc{{\birenv}, \birpc}$.

		Since $ \simrel{\mixedexec{\mixediml}{\birstate{i}}}{\sbirvaluation}{\mixedexec{\mixediml}{\sbirstate{i}}} $,
		we get that $ \mixedexec{\mixediml}{\sbirstate{i}}.\birpc \in \lbls_{\birc{\mathcal{R}}}$.
		Based on the semantics of \imlsbir{},~\autoref{fig:mixedsymbiml}, we find $ \mixedexec{\mixediml}{\sbirstate{j}} = \tuple{\sbirpcond, {\sbirenv}'', \birpc'}$  $, {\mixmultisetprocsimlsbir}$ that is in the relation $\miximltrans{\miximlsbirevent}{1}{\sbirvaluation}{\mixedexec{\mixediml}{\sbirstate{i}}}{\mixedexec{\mixediml}{\sbirstate{j}}}$ such that $\miximlsbirevent = \symfreshv{\symvar{x}_{\sbirc{d'}}}$, $ {d}' = \left\lfloor \frac{\sbirenv[\randommemidx]}{l}  \right\rfloor +1$ and $ \sbirenv{''} = {\sbirenv{'}}[\reg{0} \mapsto \var{a}; \randommemidx \mapsto \sbirenv[\randommemidx] + {l}]  $  where $ (\sbirenv{'},\var{a}) = \update(\sbirenv, \heap, \memsymb, \symvar{x}_{\sbirc{d'}},128) $ and $ \symvar{x}_{\sbirc{d'}} = \sbirrand(\sbirenv, n) $.
		
		Since $ \simrel{\mixedexec{\mixediml}{\birstate{i}}}{\sbirvaluation}{\mixedexec{\mixediml}{\sbirstate{i}}} $ and $ \mixedexec{\mixediml}{\sbirstate{i}}.\sbirenv =_{\sbirvaluation} \mixedexec{\mixediml}{\birstate{i}}.\birenv $, we get that $ d' = d $.
		Finally, we choose $\sbirvaluation'$ by extending $\sbirvaluation$ according to the executed operation and such that $  \sbirvaluation'(\symvar{x_d})=\!\imlvar{x_d}$.
		Hence, we get that $ \mixedexec{\mixediml}{\sbirstate{j}}.\sbirenv =_{\sbirvaluation'} \mixedexec{\mixediml}{\birstate{j}}.\birenv $.
		Therefore, we can conclude that $ \simrel{\tuple{\birenv{''}, \birpc'}, {\mixmultisetprocsimlbir}}{\sbirvaluation'}{ \tuple{\sbirpcond, {\sbirenv}'', \birpc'},  {\mixmultisetprocsimlsbir}}$ and $\freshv{\imlvar{x_d}} =_{\sbirvaluation'} \symfreshv{\symvar{x_d}}$. \\
		
		\item $	\birpc_\birc{1} \in \lbls_{\birc{\attacker_s}} \ \land \ \birpc_\birc{2} \in \lbls_{\birc{\attacker_r}} $ \\
		In this case, $ \imltransmulti{\miximlbirevent}{p}{\mixedexec{\mixediml}{\birstate{i}}}{\mixedexec{\mixediml}{\birstate{j}}} $ amounts to 3 sub-transitions as follows.
		First, using the $ \birc{B}to\imlc{I}$ rule in~\autoref{fig:mixedbiriml}, we have the transition relation $\imltrans{\Output{\imlvar{b},(\var{e_1},\dots,\var{e_m})}}{1}{\mixedexec{\mixediml}{\birstate{i}}}{\mixedexec{\mixediml}{\imlstate{x}}}$  such that $ \mixedexec{\mixediml}{\birstate{i}}.\birpc_\birc{1} \in \lbls_{\birc{\attacker_s}}$.
		
		Then, based on the operational semantics of $ \imlsymb $ output processes~\cite[p.~23]{aizatulin2015verifying}, the transition $ \imltransmulti{}{p}{\mixedexec{\mixediml}{\imlstate{x}}}{\mixedexec{\mixediml}{\imlstate{y}}} $ takes place.
		Finally, we have $\imltrans{\Input{\imlvar{b}',({\var{e'_1}},\dots,{\var{e'_m}})}}{1}{\mixedexec{\mixediml}{\imlstate{y}}}{\mixedexec{\mixediml}{\birstate{j}}}$ using $ \imlc{I}to\birc{B}$ rule in~\autoref{fig:mixedbiriml} such that there exist a $\birsymb$ state $ (\tuple{\birenv, \birpc_\birc{2}}, \imlchannel[{\var{e'_1}},$  $\dots,{\var{e'_m}}])$ in the multiset of executing states $ \mixedexec{\mixediml}{\imlstate{y}}.{\mixmultisetprocsimlbir} $ such that $ \birpc_\birc{2} \in \entr_\birc{\attacker_r} $. 
		
		Since $ \simrel{\mixedexec{\mixediml}{\birstate{i}}}{\sbirvaluation}{\mixedexec{\mixediml}{\sbirstate{i}}} $, we need to find a mixed symbolic and $ \imlsymb $ state $ \mixedexec{\mixediml}{\sbirstate{j}} $ such that $ \simrel{\mixedexec{\mixediml}{\birstate{j}}}{\sbirvaluation}{\mixedexec{\mixediml}{\sbirstate{j}}} $ and it is reachable from intermediate states $ \mixedexec{\mixediml}{\imlstate{x}} $ and $ \mixedexec{\mixediml}{\imlstate{y}} $ using $ \sbirc{SB}to\imlc{I}$ and $ \imlc{I}to\sbirc{SB}$ rules in~\autoref{fig:mixedsymbiml}, respectively.
		
		Therefore, since $ \birpc_\birc{1} \in \lbls_\birc{\attacker_s}  $, based on the mixed semantics of $ \imlsymb $ and $ \birsymb $,~\autoref{fig:mixedbiriml}, we get that there exist an $ \imlsymb $ state $ \mixedexec{\mixediml}{\imlstate{x}} $ such that $\imltrans{\Output{\imlvar{b},(\var{e_1},\dots,\var{e_m})}}{1}{\mixedexec{\mixediml}{\birstate{i}}}{\mixedexec{\mixediml}{\imlstate{x}}}$ and $ \mixedexec{\mixediml}{\imlstate{x}} = (\imlenv[\imlvar{x} \mapsto \imlvar{b}'], \imlprocessp), {\mixmultisetprocsimlbir} \setminus {\mixmultisetprocsimlbir}' $ such that $ \imlvar{b}' = \mathit{truncate}(\imlvar{b}, $  $\mathit{maxlen}(\imlchannel)) $.
		
		Since $ \simrel{\mixedexec{\mixediml}{\birstate{i}}}{\sbirvaluation}{\mixedexec{\mixediml}{\sbirstate{i}}} $,
		we get that $ \mixedexec{\mixediml}{\sbirstate{i}}.\birpc_\birc{1} \in \lbls_\birc{\attacker_s}$.
		Based on the semantics of \imlsbir{},~\autoref{fig:mixedsymbiml}, we find $ \mixedexec{\mixediml}{\imlstate{x}} =(\imlenv[\imlvar{x} \mapsto \imlvar{b}'], \imlprocessp), {\mixmultisetprocsimlsbir} \setminus  {\mixmultisetprocsimlsbir}'$that is in the transition relation $\mixedexec{\mixediml}{\sbirstate{i}}$  $\miximltrans{\miximlsbirevent}{1}{\sbirvaluation}{}{\mixedexec{\mixediml}{\imlstate{x}}}$ such that $\miximlsbirevent = \symoutput{\symvar{e},(\symvar{e_1},\dots,\symvar{e_m})}$.
		
		Since $ \birpc_\birc{2} \in \lbls_{\attacker_r}  $, based on the mixed semantics of $ \imlsymb $ and $ \birsymb $,~\autoref{fig:mixedbiriml}, we get that there exist an $ \imlsymb $ state $ \imlstate{y} $ such that $\imltrans{\Input{\imlvar{b}',({\var{e'_1}},\dots,{\var{e'_m}})}}{1}{\mixedexec{\mixediml}{\imlstate{y}}}{\mixedexec{\mixediml}{\birstate{j}}}$ and
		$ \mixedexec{\mixediml}{\birstate{j}} = (\tuple{ \birenv{'}[\reg{0} \mapsto \var{a}], \birpc'}, \imlchannel[{\var{e'_1}},\dots,{\var{e'_m}}]), {\mixmultisetprocsimlbir} \uplus {\mixmultisetprocsimlbir}' \setminus {\mixmultisetprocsimlbir}'' $  such that $ (\birenv{'},\var{a}) = \update(\birenv, \heap_{\birc{\mathcal{A}}}, \advmem, \imlvar{b}', 128) $ and $ \birpc' = \pcinc{\birenv, \birpc_\birc{2}} $.
		
		Since $ \simrel{\mixedexec{\mixediml}{\birstate{i}}}{\sbirvaluation}{\mixedexec{\mixediml}{\sbirstate{i}}} $,
		we get that $ \mixedexec{\mixediml}{\sbirstate{i}}.\birpc_\birc2 \in \lbls_\birc{\attacker_r}$.
		Based on the semantics of \imlsbir{},~\autoref{fig:mixedsymbiml}, we find $ \mixedexec{\mixediml}{\sbirstate{j}} =(\tuple{\sbirpcond, \sbirenv{'}[\reg{0} $  $ \mapsto \var{a}], \birpc'}, \imlchannel[{\symvar{e_1}'},\dots,{\symvar{e_m}'}]), {\mixmultisetprocsimlsbir} \uplus {\mixmultisetprocsimlsbir}' \setminus {\mixmultisetprocsimlsbir}''$ that is in the relation $\miximltrans{\miximlsbirevent}{1}{\sbirvaluation}{\mixedexec{\mixediml}{\imlstate{y}}}{\mixedexec{\mixediml}{\sbirstate{j}}}$ such that $ \birpc' = \pcinc{\sbirpcond, \sbirenv, \birpc_\birc{2}}$, $\miximlsbirevent = \syminput{{\symvar{e}}',({\symvar{e_1}'},\dots,{\symvar{e_m}'})}$ and 
        $ (\sbirenv{'},\var{a}) = \update(\sbirenv, \heap_{\birc{\mathcal{A}}}, $  $\advmem, {\symvar{e}}', 128) $.

		Moreover, we choose $\sbirvaluation'$ by extending $\sbirvaluation$ according to the executed operation and such that $  \sbirvaluation'(\symvar{e})=\!\imlvar{b}$, $  \sbirvaluation'({\symvar{e}}')=\!\imlvar{b}'$, $ \sbirvaluation'(\symvar{e_j})=\!\var{e_j} $ and $ \sbirvaluation'({\symvar{e_j}}')=\!{\var{e_j}}'$ for $ 1 \leq j \leq m $.
		Hence, we get that $ \mixedexec{\mixediml}{\sbirstate{j}}.\sbirenv =_{\sbirvaluation'} \mixedexec{\mixediml}{\birstate{j}}.\birenv $.
		Therefore, we can conclude that $ \simrel{(\tuple{ \birenv{'}[\reg{0} \mapsto \var{a}], \birpc'}, \imlchannel[{\var{e'_1}},\dots,{\var{e'_m}}]), {\mixmultisetprocsimlbir} \uplus {\mixmultisetprocsimlbir}' \setminus {\mixmultisetprocsimlbir}''}{\sbirvaluation'}{(\tuple{\sbirpcond, \sbirenv{'}[\reg{0} \mapsto \var{a}], \birpc'}, \imlchannel[{\symvar{e_1}'},\dots,{\symvar{e_m}'}]), {\mixmultisetprocsimlsbir} \uplus {\mixmultisetprocsimlsbir}' \setminus {\mixmultisetprocsimlsbir}''}$, $\Output{\imlvar{b},(\var{e_1},\dots,\var{e_m})} =_{\sbirvaluation'} \symoutput{\symvar{e},(\symvar{e_1},\dots,\symvar{e_m})}$, and $ \Input{\imlvar{b}',({\var{e'_1}},$  $\dots,{\var{e'_m}})} =_{\sbirvaluation'} \syminput{{\symvar{e}}',({\symvar{e_1}'},\dots,{\symvar{e_m}'})} $. \\

		\item $	\birpc \in \lbls_{\birc{\mathcal{N}}} \ \land \ \birpc \notin \lbls_{\birc{\mathcal{L}}}$ \\
		Based on the mixed semantics of $ \imlsymb $ and $ \birsymb $,~\autoref{fig:mixedbiriml}, we get that $ \mixedexec{\mixediml}{\birstate{j}} = \tuple{{\birenv}', \birpc'}, {\mixmultisetprocsimlbir}$ and $\imltrans{\miximlbirevent}{1}{\mixedexec{\mixediml}{\birstate{i}}}{\mixedexec{\mixediml}{\birstate{j}}}$.
		Since $ \simrel{\mixedexec{\mixediml}{\birstate{i}}}{\sbirvaluation}{\mixedexec{\mixediml}{\sbirstate{i}}} $,
		we get that $ \mixedexec{\mixediml}{\sbirstate{i}}.\birpc \in \lbls_{\birc{\mathcal{N}}}$.
		Based on the semantics of \imlsbir{},~\autoref{fig:mixedsymbiml}, we find $ \mixedexec{\mixediml}{\sbirstate{j}} =  \tuple{ \sbirpcond', {\sbirenv}', \birpc'}, {\mixmultisetprocsimlsbir}$ that is in the transition relation $\miximltrans{}{1}{\sbirvaluation}{\mixedexec{\mixediml}{\sbirstate{i}}}{\mixedexec{\mixediml}{\sbirstate{j}}}$.
		
		Based on Lindner's work~\cite{DBLP:journals/scp/LindnerGM19} (see ~\autoref{single-step-vanila}), there exist an interpretation  $\sbirvaluation' \supseteq \sbirvaluation$ such that $ \mixedexec{\mixediml}{\sbirstate{j}}.\sbirenv =_{\sbirvaluation'} \mixedexec{\mixediml}{\birstate{j}}.\birenv $.
		Therefore, we can conclude that $ \simrel{\tuple{ {\birenv}', \birpc'}, {\mixmultisetprocsimlbir}}{\sbirvaluation'}{\tuple{ \sbirpcond', {\sbirenv}', \birpc'}, $  $ {\mixmultisetprocsimlsbir}}$. \\
		
		\item $	\birpc \in \lbls_{\birc{\mathcal{L}}}$ \\
		Based on the mixed semantics of $ \imlsymb $ and $ \birsymb $,~\autoref{fig:mixedbiriml}, we get that $ \mixedexec{\mixediml}{\birstate{j}} = \tuple{{\birenv}', \birpc'}, {\mixmultisetprocsimlbir}$ and $\imltransmulti{\Loop{\var{t}}}{1}{\mixedexec{\mixediml}{\birstate{i}}}{\mixedexec{\mixediml}{\birstate{j}}}$ and $ \birpc' = \pcexit{\birpc} $.
		Since $ \simrel{\mixedexec{\mixediml}{\birstate{i}}}{\sbirvaluation}{\mixedexec{\mixediml}{\sbirstate{i}}} $,
		we get that $ \mixedexec{\mixediml}{\sbirstate{i}}.\birpc \in \lbls_{\birc{\mathcal{L}}}$.
		Based on the semantics of \imlsbir{},~\autoref{fig:mixedsymbiml} and the computed loop summary\footnote{Note that the concrete $\birsymb$ uses unrolling for loops, while our symbolic execution, $ \sbirsymb $, uses the loop summary to handle bounded loops.}, we find $ \mixedexec{\mixediml}{\sbirstate{j}} =  \tuple{ \sbirpcond', {\sbirenv}', \birpc'},$  ${\mixmultisetprocsimlsbir}$ that is in the transition relation $\miximltransmulti{\symloop{\symvar{t}}}{1}{\sbirvaluation}{\mixedexec{\mixediml}{\sbirstate{i}}}{\mixedexec{\mixediml}{\sbirstate{j}}}$ such that $ \birpc' = \pcexit{\birpc} $. 
		Therefore, we get that both loops in $ \mixedexec{\mixediml}{\birstate{j}} $ and $ \mixedexec{\mixediml}{\sbirstate{j}} $ terminate at the same position $ \birpc' $. 
		
		Finally, for every $\birsymb$ variable $ \var{x} $ and symbolic $\birsymb$ variable $ \symvar{x} $ that $\mixedexec{\mixediml}{\!\sbirstate{j}\!}.\sbirenv[\symvar{x}] \neq \mixedexec{\mixediml}{\!\sbirstate{i}\!}.\sbirenv[\symvar{x}]$ and $ \mixedexec{\mixediml}{\!\birstate{j}\!}.\birenv[\var{x}] \neq \mixedexec{\mixediml}{\!\birstate{i}\!}.\birenv[\var{x}]$, we choose $\sbirvaluation'$ by extending $\sbirvaluation$ according to the executed operation such that $  \sbirvaluation'(\mixedexec{\mixediml}{\!\sbirstate{j}\!}.\sbirenv[\symvar{x}])= \mixedexec{\mixediml}{\!\birstate{j}\!}.\birenv[\var{x}]$. 
		Moreover, for the symbolic counter $ \symvar{t} $ which represents the number of iterations of the loop in symbolic execution and the concrete counter $ \var{t} $, we have $ \sbirvaluation' (\symvar{t}) = \var{t}$.
		Therefore, we can conclude that $ \simrel{\tuple{ {\birenv}', \birpc'}, {\mixmultisetprocsimlbir}}{\sbirvaluation'}{\tuple{ \sbirpcond', {\sbirenv}', \birpc'},  {\mixmultisetprocsimlsbir}}$ and $ \symloop{\symvar{t}} =_{\sbirvaluation'} \Loop{\var{t}} $. \\

	\end{itemize}
\end{proof}

\subsection{Security properties}\label{sec:properties}

\subsubsection*{Preliminaries}

Trace property $ \traceproperty $ is a polynomially decidable prefix-closed set of event traces.

\begin{math}
	\begin{array}{l}
		\{ \eventtrace  \in \traceproperty \ | \ \forall i \in \naturalnum : \eventtrace[..i] \in \traceproperty\\ 
		\hspace{30mm}\begin{array}{l}	
			\implies \left( \exists j \in \naturalnum: j < i \ \land \ \eventtrace[..j] \in \traceproperty \right)
		\end{array} \}
	\end{array}
\end{math}\\

Here, we use the following notation.
Given a trace
$ \eventtrace = \eventsym_0 \eventsym_1 \dots $
and index $ i $,
we define $ \eventtrace[i] = \eventsym_i$ and
$ \eventtrace[..i] = \eventsym_0 \eventsym_1 \dots \eventsym_i$.

\thmsoundness*
\begin{proof}
    \newcommand{\mDef}[1]{\text{Def.~\ref{#1}}}
    \newcommand{\mLem}[1]{\text{Lem.~\ref{#1}}}
    \newcommand{\mThm}[1]{\text{Thm.~\ref{#1}}}

Throughout the following proof, we use the following property of sums:
if $\zeta: A \to B$ is an injection,
then $
\sum_{a\in A} f(\zeta(a))
\le 
\sum_{b\in B} f(b)
$. 
Recall also that, if $\zeta: A \to B$ is injective,
then $\zeta|_{A'}$ is injective for any $ A' $.
Also, for brevity, let $\mixedsbirtraces=\sbirtracesproof$ and
$\mixedbirtraces = \birtracesproof$
in the follow up. \\

$\insecurity{\imlsystemproof, n}{ \traceproperty}$
\begin{align*}
    \stackrel{\mDef{def:IML-insec}}{=} &
    \sum_{\imltrace \in \imltraces(\imlsystemproof,n)\cap \tracepropertyneg } \imlpr(\imltrace)
    \\
    \\
    \intertext{Note that $ \simreltratraces{\mixedsbirexecution}{\sbirvaluation}{k}{\sbirtoiml{\cdot}}{\imlexecution}$
    implies $\sbirvaluation(\mixedsbirtrace) \in \tracepropertyneg \iff \imltrace \in \tracepropertyneg$}
\stackrel{\mLem{lem:traceinclusion},\mLem{lem:same-random-steps}}{\ge} &
    \sum_{
        \substack{
            \mixedsbirtrace\in\mixedsbirtraces
            \\
            \sbirvaluation: \randomsymvals_{\sbirc{l}} \to \imlvals^k_n
            \\
        \text{s.t.~}\exists \sbirvaluation'.\
        \sbirvaluation'|_{\randomsymvals_{\sbirc{l}}} = \sbirvaluation 
        \\
        \land \sbirvaluation'(\mixedsbirtrace) \in \tracepropertyneg 
        \\
        \land \rng{\mixedsbirtrace} = l
    }
}  2^{-n l} \imlpr(\mixedsbirtrace)
    \intertext{Let $k$ be the maximal $l$ in the sum. We extend the
    range of $\sbirvaluation$, which increases the range of the sum by
$2^{n(k-l)}$:}
    = &
    \sum_{
        \substack{
            \mixedsbirtrace\in\mixedsbirtraces
            \\
            \sbirvaluation: \randomsymvals_{\sbirc{k}} \to \imlvals^k_n
            \\
        \text{s.t.~}\exists \sbirvaluation'.\
        \sbirvaluation'|_{\randomsymvals_{\sbirc{l}}} = \sbirvaluation |_{\randomsymvals_{\sbirc{l}}}
        \\
        \land \sbirvaluation'(\mixedsbirtrace) \in \tracepropertyneg 
        \\
        \land \rng{\mixedsbirtrace} = l
    }
}  \underbrace{\frac{2^{-n l}}{2^{n(k - l)}}}_{= 2^{-nk}}\cdot \imlpr(\mixedsbirtrace)
\\
    \stackrel{\mLem{lem:extra-randomness}}{=} &
    \sum_{
        \substack{
            \mixedsbirtrace\in\mixedsbirtraces
            \\
            \sbirvaluation: \randomsymvals_{\sbirc{k}} \to \imlvals^k_n
            \\
        \text{s.t.~}\exists \sbirvaluation'.
        \sbirvaluation'|_{\randomsymvals_{\sbirc{k}}} = \sbirvaluation
        \\
        \land \sbirvaluation'(\mixedsbirtrace) \in \tracepropertyneg 
    }
}  2^{-nk} \imlpr(\mixedsbirtrace)\\
    = & 
2^{-nk} \cdot
    \sum_{
        \substack{
            \mixedsbirtrace\in\mixedsbirtraces
            \\
            \sbirvaluation: \randomsymvals_{\sbirc{k}} \to \imlvals^k_n
            \\
        \text{s.t.~}\exists \sbirvaluation'.
        \sbirvaluation'|_{\randomsymvals_{\sbirc{k}}} = \sbirvaluation
        \\
        \land \sbirvaluation'(\mixedsbirtrace) \in \tracepropertyneg 
    }
}  \imlpr(\mixedsbirtrace)
    \\
    \intertext{Note that $ \simreltraces{\mixedbirexecution}{\sbirvaluation}{k}{\mixedsbirexecution}$
    	implies $\sbirvaluation(\mixedsbirtrace) \in \tracepropertyneg \iff \mixedbirtrace \in \tracepropertyneg$}
\stackrel{\mLem{lem:traceinclusion-bir}}{\ge} &
2^{-nk} \cdot
    \sum_{\substack{
       \randomvals_{\imlc{k}} \in \imlvals_n^k \\
       \mixedbirtrace\in\mixedbirtraces \cap \tracepropertyneg
   }} \imlpr(\mixedbirtrace)
            \\
    \stackrel{\mDef{def:BIR-insec}}{=} &
\insecurity{\mixedbirprog,n,k}{\traceproperty}
\end{align*}
\end{proof}

\subsubsection{Lemmas required to prove \autoref{thm:soundness}}\label{lem:thm4lems}

\begin{lemma}[trace contains
    randomness]\label{lem:trace-contains-randomness}
    For a $\birsymb$ program $\birprog$,
    an $ \imlsymb $ process $\imlprog$,
    any upper bound $ k \in \naturalnum $ on the number of RNG steps, 
    all $ \imlsymb $ traces $ \imlexecution \in \imlexecutions(\imltransprog)$,
    all interpretations $ \sbirvaluation $ and $ \sbirvaluation'$,
    and all mixed $ \imlsymb $ and symbolic execution traces $\mixedsbirexecution\in\mixedsbirexecutions(\mixedbirprog, \sbirenv_{\sbirc{0}}[ \randommem \mapsto \randomsymvals_{\sbirc{k}}  , \randommemidx \mapsto 0])$  with $d$ the number of $\textsc{RNG}$ steps in $\mixedsbirexecution$ such that $\randomsymvals_\sbirc{d}$ be the first $d$ random symbolic values in
    $\randommem$ and $ d \leq k $,
    then,
    \[
        \simreltratraces{\mixedsbirexecution}{\sbirvaluation}{k}{\sbirtoiml{\cdot}}{\imlexecution}
        \land
    \simreltratraces{\mixedsbirexecution}{\sbirvaluation'}{k}{\sbirtoiml{\cdot}}{\imlexecution}
        \implies
        \sbirvaluation|_{\randomsymvals_\sbirc{d}}
        = \sbirvaluation'|_{\randomsymvals_\sbirc{d}}
    \]
\end{lemma}

\begin{proof}
	Since  the relation $ \simreltratraces{\mixedsbirexecution}{\sbirvaluation}{k}{\sbirtoiml{\cdot}}{\imlexecution} $ holds, we know that 
	$ \mixedsbirexecution = \miximltrans{\miximlsbirevent_1}{p_1}{\sbirvaluation}{\mixedexec{\mixediml}{\sbirstate{0}}}{} \dots  \miximltrans{\miximlsbirevent_{i}}{p_{i}}{\sbirvaluation}{}{\mixedexec{\mixediml}{\sbirstate{i}}} \miximltrans{\symfreshv{\symvar{x_1}}}{1}{\sbirvaluation}{}{\mixedexec{\mixediml}{\sbirstate{i+1}}} \\
	\miximltrans{\miximlsbirevent_{i+2}}{p_{i+2}}{\sbirvaluation}{}{} \dots \miximltrans{\miximlsbirevent_{j}}{p_{j}}{\sbirvaluation}{}{\mixedexec{\mixediml}{\sbirstate{j}}} 
	\miximltrans{\symfreshv{\symvar{x_d}}}{1}{\sbirvaluation}{}{\mixedexec{\mixediml}{\sbirstate{j+1}}}
	\miximltrans{\miximlsbirevent_{j+2}}{p_{j+2}}{\sbirvaluation}{}{} \dots \\
	\miximltrans{}{p_m}{\sbirvaluation}{}{\mixedexec{\mixediml}{\sbirstate{m}}} $ 
	where $ \miximlsbirevent_i $ does not come from a RNG($n$) call and $ \symvar{x_1} , \dots , \symvar{x_d} = \randomsymvals_\sbirc{d} $
	and
	$\imlexecution = \imltrans{\imlevent_\imlc1}{p_1}{\imlstate{0}}{} \dots \imltrans{\imlevent_\imlc{i}}{p_{i}}{}{\imlstate{i}} \imltrans{\imlfreshev{\imlvar{b_1}}}{\probdist{n}}{}{\imlstate{i+1}} \dots \\
	 \imltrans{\imlevent_\imlc{j}}{p_{j}}{}{\imlstate{j}} \imltrans{\imlfreshev{\imlvar{b_d}}}{\probdist{n}}{}{\imlstate{j+1}}  \dots 
	\imltrans{}{p_m}{}{\imlstate{m}}$.
	For all $ 0 < i \leq d $, then, $ \imlvar{b_i} = \sbirvaluation(\symvar{x_i}) $ (1).
	Since  $ \simreltratraces{\mixedsbirexecution}{\sbirvaluation'}{k}{\sbirtoiml{\cdot}}{\imlexecution} $, for the random symbolic values $ \symvar{x_1} , \dots , \symvar{x_d} $, we get that $ \imlvar{b_i} = \sbirvaluation'(\symvar{x_i}) $ for $ 0 < i \leq d $ (2).
	From (1) and (2), we conclude that $ \sbirvaluation|_{\randomsymvals_\sbirc{d}} = \sbirvaluation'|_{\randomsymvals_\sbirc{d}} $.
	
\end{proof}

\begin{lemma}\label{lem:same-random-steps}
       For a $\birsymb$ program $\birprog$,
       an $ \imlsymb $ process $\imlprog$,
       any upper bound $ k \in \naturalnum $ on the number of RNG steps, 
       a security parameter $\secparam\in\naturalnum$,
       all $ \imlsymb $ traces $ \imlexecution \in \imlexecutions(\imltransprog)$,
       all interpretations $ \sbirvaluation $,
       all mixed $ \imlsymb $ and symbolic execution traces $\mixedsbirexecution\in\mixedsbirexecutions(\mixedbirprog, \sbirenv_{\sbirc{0}}[ \randommem \mapsto \randomsymvals_{\sbirc{k}}  , \randommemidx \mapsto 0])$,
       if $\simreltratraces{\mixedsbirexecution}{\sbirvaluation}{k}{\sbirtoiml{\cdot}}{\imlexecution}$
       , then, the number of \textsc{Inew} steps in $\imlexecution$
       and the number of \textsc{RNG} plus the number of \textsc{INew} steps in $\mixedsbirexecution$ are equal. 
       Moreover, $\imlpr(\trace{\imlexecution}) = \imlpr(\trace{\mixedsbirexecution}) \cdot 2^{-n \cdot \rng{\mixedsbirexecution}}$. 
\end{lemma}

\begin{proof}
	Since  $ \simreltratraces{\mixedsbirexecution}{\sbirvaluation}{k}{\sbirtoiml{\cdot}}{\imlexecution} $, we know that
	$ \symfreshv{\symvar{x_i}} $ or $ \imlfreshev{\imlvar{b}} $ steps in $ \mixedsbirexecution $ maps to $ \imlfreshev{\sbirvaluation(\symvar{x_i})} $ or $ \imlfreshev{\imlvar{b}} $ steps in $ \imlexecution $, respectively, for $ i \leq \rng{\mixedsbirexecution} $. The first statement is as follows:
	
	Based on the rule \textsc{INew} in the operational semantics of $ \imlsymb $ output processes~\cite[p.~23]{aizatulin2015verifying}, we get that
	the probability of generating the random number $ \imlvar{b} $ based on the $ \imlsymb $ transition relation (i.e. $\imltrans{}{\probdist{n}}{}{} $ semantics) with respect to security parameter $ n $ is $ \probdist{n} $. 
	Therefore, the $ \imlfreshev{\imlvar{b}} $ step has the probability $ \probdist{n} $ in both $ \mixedsbirexecution $ and $ \imlexecution $.
	Hence, for $ \imlfreshev{\imlvar{b}} $ steps, we have $ \imlpr(\trace{\imlexecution}) = \imlpr(\trace{\mixedsbirexecution}) $.
	
	Based on the rule RNG(${n}$) in the semantics of \imlsbir{},~\autoref{fig:mixedsymbiml}, we get that the probability of generating the symbolic random number $ \symvar{x} $ based on the mixed $ \imlsymb $ and symbolic transition relation (i.e., $\miximltrans{}{1}{\sbirvaluation}{}{}$ semantics) with respect to security parameter $ n $ is $ 1 $ but we translate $ \symfreshv{\symvar{x}} $ to $ \imlfreshev{\sbirvaluation(\symvar{x})} $ using interpretation $ \sbirvaluation $ which have the probability $ 1 \cdot \probdist{n} $.
	The number of $ \symfreshv{\symvar{x}} $ steps in $ \mixedsbirexecution $ are $ \rng{\mixedsbirexecution} $, hence, the probability of $ \symfreshv{\symvar{x}} $ steps is $ \imlpr(\trace{\mixedsbirexecution}) \cdot 2^{-n \cdot \rng{\mixedsbirexecution}}$.

	Therefore, we can conclude that $\imlpr(\trace{\imlexecution}) = \imlpr(\trace{\mixedsbirexecution}) \cdot 2^{-n \cdot \rng{\mixedsbirexecution}}$.
	
\end{proof}

\begin{lemma}[extra randomness]\label{lem:extra-randomness}
    For a $\birsymb$ program $\birprog$,
    an $ \imlsymb $ process $\imlprog$,
    any upper bound $ k \in \naturalnum $ on the number of RNG steps, 
    a security parameter $\secparam\in\naturalnum$,
    all $ \imlsymb $ traces $ \imlexecution \in \imlexecutions(\imltransprog)$,
    all interpretations $ \sbirvaluation $,
    all mixed $ \imlsymb $ and symbolic execution traces $\mixedsbirexecution\in\mixedsbirexecutions(\mixedbirprog, \sbirenv_{\sbirc{0}}[ \randommem \mapsto \randomsymvals_{\sbirc{k}}  , \randommemidx \mapsto 0])$
    with $ k = \rng{\mixedsbirexecution}$,
    and for any $l \ge k$,
    if
    $
        \simreltratraces{\mixedsbirexecution}{\sbirvaluation}{k}{\sbirtoiml{\cdot}}{\imlexecution}
        $
    , then, for all $\sbirvaluation'$ with
    $\domain{\sbirvaluation'}\subseteq \randomsymvals_{\sbirc{l}} \setminus
    \randomsymvals_{\sbirc{k}}$
    , we have, 
$
\simreltratraces{\mixedsbirexecution}{\sbirvaluation'|_{\randomsymvals_{\sbirc{k}}}}{k}{\sbirtoiml{\cdot}}{\imlexecution}
        $.
\end{lemma}

\begin{proof}
	Since $ \simreltratraces{\mixedsbirexecution}{\sbirvaluation}{k}{\sbirtoiml{\cdot}}{\imlexecution} $, we know that the number of RNG steps in $\mixedsbirexecution$ is $ \rng{\mixedsbirexecution} \leq k $.
	Based on the semantics of \imlsbir{},
	\autoref{fig:mixedsymbiml}, we know that random symbolic values $ \symvar{x_i} \in  \randomsymvals_{\sbirc{k}} $ for $ 0 < i \leq k $ are used in RNG($n$) rule by order of $ i $. 
	The random symbolic values $ \symvar{x_j} \in  \randomsymvals_{\sbirc{l}} $ for $ l \ge j \ge k $ are not generated by RNG($n$) rule in~\autoref{fig:mixedsymbiml}.
	Therefore, the random symbolic values $ \symvar{x_j} \in  \randomsymvals_{\sbirc{l}} $ for $ l \ge j \ge k $ are not part of the mixed $ \imlsymb $ and symbolic state after RNG($n$) call.
	Hence, we can conclude that modification of $ \sbirvaluation $ does not affect trace equivalence and we have $ \simreltratraces{\mixedsbirexecution}{\sbirvaluation'}{k}{\sbirtoiml{\cdot}}{\imlexecution} $ such that $ \sbirvaluation \subseteq \sbirvaluation' $ and $\domain{\sbirvaluation'}\subseteq \randomsymvals_{\sbirc{l}} \setminus
	\randomsymvals_{\sbirc{k}}$.
	
\end{proof}

\begin{lemma}[Injective Event Trace Inclusion ($ \imlsymb $)]\label{lem:traceinclusion}
	For a $\birsymb$ program $\birprog$,
	an $ \imlsymb $ process $\imlprog$,
	any security parameter $n \in \naturalnum$,
	an upper bound on the number of nonces $k\in \naturalnum$,
    and for $\mixedsbirtraces= \sbirtracesproof$,
    there is an injective function $\zeta$ from
    \[
     \left\{
    \left(
        \begin{aligned}
    \mixedsbirtrace\in\mixedsbirtraces,
    \\
    \sbirvaluation: \randomsymvals_{\sbirc{k}} \to \imlvals^k_n 
    \end{aligned}
    \right)
    \suchthat
    \begin{aligned}
        & \rng{\mixedsbirtrace} = k \ \land  \\
        & \exists \sbirvaluation'.\
        \sbirvaluation'|_{\randomsymvals_{\sbirc{k}}} = \sbirvaluation
    \end{aligned}
\right\}        
    \]
    to $ \imltracesproof$
    such that
    \[
    \zeta(\mixedsbirtrace, \sbirvaluation)=\imltrace
    \implies
    \left(
     \begin{aligned}
    	& \exists \sbirvaluation', \mixedsbirexecution, \imlexecution. \\
    	& \quad \simreltratraces{\mixedsbirexecution}{\sbirvaluation'}{k}{\sbirtoiml{\cdot}}{\imlexecution} \ \land \\
    	& \trace{\mixedsbirexecution} = \mixedsbirtrace \ \land \ \trace{\imlexecution} = \imltrace
    \end{aligned}
	\right)
    \]
\end{lemma}
\begin{proof}
	 From the skolemization of~\autoref{thm:iml:traceeq}, we define $ \zeta' : \mixedsbirexecutions \times \sbirvaluations \to \imlexecutions $ such that $ \zeta'(\mixedsbirexecution,\sbirvaluation')=\imlexecution $ implies $ \simreltratraces{\mixedsbirexecution}{\sbirvaluation'}{k}{\sbirtoiml{\cdot}}{\imlexecution} $.
	We define $ \zeta : \mixedsbirtraces \times \sbirvaluations \to \imltraces $ such that  $ \zeta(\mixedsbirtrace, \sbirvaluation)=\imltrace $.
	Hence, we choose arbitrary $ \mixedsbirexecution $ and $ \sbirvaluation' $ such that $ \trace{\mixedsbirexecution} = \mixedsbirtrace $
	and $ \trace{\zeta'(\mixedsbirexecution,\sbirvaluation')} = \trace{\imlexecution} = \imltrace $.
	
	Then, we prove the  function $\zeta$ is injective by contradiction and assume for arbitrary $ \mixedsbirtrace $, $ {\mixedsbirtrace}' $, $ \sbirvaluation $, and $ \sbirvaluation' $ such that $ \mixedsbirtrace \neq {\mixedsbirtrace}' \lor \sbirvaluation \neq \sbirvaluation' $, we have $ \zeta(\mixedsbirtrace,\sbirvaluation)=\imltrace = \zeta({\mixedsbirtrace}',\sbirvaluation') $.
	Therefore, $ \sbirvaluation(\mixedsbirtrace) = \sbirvaluation'({\mixedsbirtrace}') $ and for any event $ \miximlsbirevent $ where $ \mixedsbirtrace $ and $ {\mixedsbirtrace}' $ are different, we have $ \sbirvaluation(\miximlsbirevent) = \sbirvaluation'({\miximlsbirevent}') $.
	Hence, there exist two cases:
	\begin{enumerate}
		\item $ \miximlsbirevent \neq {\miximlsbirevent}' $
		\begin{enumerate}
		\item $ \sbirvaluation = \sbirvaluation' $:
		We assume that there exists a symbolic value $ \symvar{d}$ such that $ \sbirvaluation(\symevent{\symvar{d}}) \neq \sbirvaluation(\symfreshv{\symvar{d}})$.
		Based on $ \simreltratraces{\mixedsbirexecution}{\sbirvaluation}{k}{\sbirtoiml{\cdot}}{\imlexecution} $, we have $ 	\imlev{\sbirvaluation(\symvar{d})} \in \imltrace $ at the position $ j $ and $\symevent{\symvar{d}} \in$ $\mixedsbirtrace $ at the same position $ j $. 
		Similarly, we have $\imlfreshev{\sbirvaluation(\symvar{d})} \in \imltrace $ at the position $ j $ and $ \symfreshv{\symvar{d}} \in$ ${\mixedsbirtrace}' $ at the same position $ j $. 
		Hence, we get that we have two different $ \imltrace $ (i.e., at the position $ j $). 
		Since $ \miximlsbirevent \neq {\miximlsbirevent}' $ gives a contradiction, we deduce that $  \miximlsbirevent = {\miximlsbirevent}' $.\\
		\item $ \sbirvaluation \neq \sbirvaluation' $:
		We assume that there exists a symbolic value $ \symvar{d}$ such that $ \sbirvaluation(\symevent{\symvar{d}}) \neq \sbirvaluation'(\symfreshv{\symvar{d}})$.
		Based on $ \simreltratraces{\mixedsbirexecution}{\sbirvaluation}{k}{\sbirtoiml{\cdot}}{\imlexecution} $, we have $ 	\imlev{\sbirvaluation(\symvar{d})} $  $ \in \imltrace $ at the position $ j $ and $ \symevent{\symvar{d}} \in$ $\mixedsbirtrace $ at the same position $ j $. 
		Based on $ \simreltratraces{{\mixedsbirexecution}'}{\sbirvaluation'}{k}{\sbirtoiml{\cdot}}{\imlexecution} $, we have $\imlfreshev{\sbirvaluation'(\symvar{d})} \in \imltrace $ at the position $ j $ and $ \symfreshv{\symvar{d}} \in$ ${\mixedsbirtrace}' $ at the same position $ j $. 
		Hence, we get that we have two different $ \imltrace $ (i.e., at the position $ j $). 
		Since $ \miximlsbirevent \neq {\miximlsbirevent}' $ gives a contradiction, we deduce that $  \miximlsbirevent = {\miximlsbirevent}' $.\\
		\end{enumerate}
		\item $ \sbirvaluation \neq \sbirvaluation' \ \land \  \miximlsbirevent = {\miximlsbirevent}' $
		\begin{enumerate}
			\item $ \miximlsbirevent $ is a $ \symfreshv{\cdot}  $ event:
			We assume that there exists a random symbolic value $ \symvar{x_i} \in  \randomsymvals_{\sbirc{k}} $ for $ 0 < i \leq k $ such that $ \sbirvaluation(\symfreshv{\symvar{x_i}}) \neq \sbirvaluation'(\symfreshv{\symvar{x_i}})$.
			Based on $ \simreltratraces{\mixedsbirexecution}{\sbirvaluation}{k}{\sbirtoiml{\cdot}}{\imlexecution} $, we have $ 	\imlfreshev{\sbirvaluation(\symvar{x_i})} \in \imltrace $ at the position $ j $ and $ \symfreshv{\symvar{x_i}} \in$ $\mixedsbirtrace $ at the same position $ j $. Similarly, for $ {\mixedsbirtrace}' $ and $ \sbirvaluation' $ and the same position $ j $. Therefore, $ \sbirvaluation(\symvar{x_i}) = \sbirvaluation'(\symvar{x_i}) $. Hence, from $ \symvar{x_i} \in  \randomsymvals_{\sbirc{k}} $, we get that $ \sbirvaluation(\symfreshv{\symvar{x_i}}) = \sbirvaluation'(\symfreshv{\symvar{x_i}})$. Since $ \sbirvaluation \neq \sbirvaluation' $ gives a contradiction, we deduce that $ \sbirvaluation = \sbirvaluation' $.\\
			\item $ \miximlsbirevent $ is a $ \symevent{\cdot} $ event:
			We assume that there exists a symbolic value $ \symvar{d}$ such that $ \sbirvaluation(\symevent{\symvar{d}}) \neq \sbirvaluation'(\symevent{\symvar{d}})$.
			Based on $ \simreltratraces{\mixedsbirexecution}{\sbirvaluation}{k}{\sbirtoiml{\cdot}}{\imlexecution} $, we have $ 	\imlev{\sbirvaluation(\symvar{d})} \in \imltrace $ at the position $j$ and $ \symevent{\symvar{d}} \in$ $\mixedsbirtrace $ at the same position $ j $. Similarly, for $ {\mixedsbirtrace}' $ and $ \sbirvaluation' $ and the same position $ j $. Because we generate symbolic values in a canonical form, $ \sbirvaluation(\symvar{d}) = \sbirvaluation'(\symvar{d}) $. Therefore, we get that $ \sbirvaluation(\symevent{\symvar{d}}) = \sbirvaluation'(\symevent{\symvar{d}})$. Since $ \sbirvaluation \neq \sbirvaluation' $ gives a contradiction, we deduce that $ \sbirvaluation = \sbirvaluation' $.
		\end{enumerate}	
	\end{enumerate}
	Therefore, we can conclude that the function $\zeta$ is injective.

\end{proof}

\begin{lemma}[Injective Event Trace Inclusion (symbolic
    execution)]\label{lem:traceinclusion-bir}
    For a $\birsymb$ program $\birprog$,
    an $ \imlsymb $ process $\imlprog$,
    any security parameter $n \in \naturalnum$,
    an upper bound on the number of nonces $k\in \naturalnum$,
    and for $\mixedbirtraces= \birtracesproof$,
    there is an injective function $\zeta$ from

$ \{ ( \mixedbirtrace, \randomvals_{\imlc{k}} ) \mid  \randomvals_{\imlc{k}} \in \imlvals_n^k \ \land $ \\
$ \mixedbirtrace \in \birtracesproof \}  $

    to 
    \[
     \left\{
    \left(
        \begin{aligned}
    \mixedsbirtrace\in\mixedsbirtraces,
    \\
    \sbirvaluation: \randomsymvals_{\sbirc{k}} \to \imlvals^k_n 
    \end{aligned}
    \right)
    \suchthat
    \begin{aligned}
        \text{s.t.~}\exists \sbirvaluation'. &
        \sbirvaluation'|_{\randomsymvals_{\sbirc{k}}} = \sbirvaluation
    \end{aligned}
\right\}
    \]
    such that 
  
    $ \zeta( \mixedbirtrace, \randomvals_{\imlc{k}} ) = (\mixedsbirtrace,\sbirvaluation) $\\
     
    $ \implies \left(
    \begin{aligned}
    	& \imlpr(\mixedbirtrace) \cdot 2^{-n \cdot k } = \imlpr(\mixedsbirtrace) \cdot 2^{-n \cdot k } \ \land \\
    	& \trace{\mixedbirexecution} = \mixedbirtrace \ \land \ \trace{\mixedsbirexecution} = \mixedsbirtrace
    \end{aligned}
    \right) $
    
\end{lemma}
\begin{proof}
	From the skolemization of~\autoref{thm:bir:traceeq}, we define $ \zeta' : \mixedbirexecutions \times \imlvals^k_n \to \mixedsbirexecutions \times \sbirvaluations $ such that $ \zeta'( \mixedbirexecution, \randomvals_{\imlc{k}} )
	= (\mixedsbirexecution,\sbirvaluation) $ implies $ \simreltraces{\mixedbirexecution}{\sbirvaluation}{k}{\mixedsbirexecution} $.
	We define $ \zeta : \mixedbirtraces \times \imlvals^k_n \to \mixedsbirtraces \times \sbirvaluations $ such that  $ \zeta( \mixedbirtrace, \randomvals_{\imlc{k}} )
	= (\mixedsbirtrace,\sbirvaluation) $.
	Hence, we choose arbitrary $ \mixedbirexecution $ and $ \randomvals_{\imlc{k}} $ such that $ \trace{\mixedbirexecution} = \mixedbirtrace $
	and $ \trace{\zeta'( \mixedbirexecution, \randomvals_{\imlc{k}} )} = \trace{\mixedsbirexecution,\sbirvaluation} = \mixedsbirtrace $.
	
	Then, we prove the  function $\zeta$ is injective by contradiction and assume for arbitrary $ \mixedbirtrace $, $ {\mixedbirtrace}' $, $ \randomvals_{\imlc{k}} $, and $ {\randomvals}'_{\imlc{k}} $ such that $ \mixedbirtrace \neq {\mixedbirtrace}' \lor \randomvals_{\imlc{k}} \neq {\randomvals}'_{\imlc{k}} $, we have $ \zeta( \mixedbirtrace, \randomvals_{\imlc{k}} )
	= (\mixedsbirtrace,\sbirvaluation) = \zeta( {\mixedbirtrace}', {\randomvals}'_{\imlc{k}} )$.
	Therefore, we get that $ ( \mixedbirtrace, \randomvals_{\imlc{k}} ) = \sbirvaluation(\mixedsbirtrace) = ( {\mixedbirtrace}', {\randomvals}'_{\imlc{k}} )$.
	Hence, there exist two cases:

	\begin{enumerate}
		\item $ \randomvals_{\imlc{k}} \neq {\randomvals}'_{\imlc{k}} $ \\
		We assume that there exists a random value $ \imlvar{x_i} \in  \randomvals_{\imlc{k}} $ and a random value $ \imlvar{x_i}' \in  {\randomvals}'_{\imlc{k}} $ for $ 0 < i \leq k $ such that $ \imlvar{x_i} \neq \imlvar{x_i}' $.
		Based on $ \simreltraces{\mixedbirexecution}{\sbirvaluation}{k}{\mixedsbirexecution} $, we have $	\symfreshv{\symvar{x_i}} \in$ $ \mixedsbirtrace $ at the position $ j $ and $ \freshv{\imlvar{x_i}} \in \mixedbirtrace$ at the same position $ j $ such that $ \sbirvaluation(\symvar{x_i}) = \imlvar{x_i}$.
		Based on $ \simreltraces{{\mixedbirexecution}'}{\sbirvaluation}{k}{\mixedsbirexecution} $, we have $	\symfreshv{\symvar{x_i}} \in$ $ \mixedsbirtrace $ at the position $ j $ and $ \freshv{\imlvar{x_i}'} \in {\mixedbirtrace}'$ at the same position $ j $ such that $ \sbirvaluation(\symvar{x_i}) = \imlvar{x_i}'$.
		Therefore, we get that $ \imlvar{x_i} = \sbirvaluation(\symvar{x_i}) = \imlvar{x_i}' $. 
		Since $ \randomvals_{\imlc{k}} \neq {\randomvals}'_{\imlc{k}}  $ gives a contradiction, we deduce that $ \randomvals_{\imlc{k}} = {\randomvals}'_{\imlc{k}}  $.\\
		
		\item $ \mixedbirtrace \neq {\mixedbirtrace}' \land \randomvals_{\imlc{k}} = {\randomvals}'_{\imlc{k}} $\\
		Let $ \miximlbirevent $ and $ {\miximlbirevent}' $ be the earliest mixed $ \imlsymb $ and $\birsymb$ events where $ \mixedbirtrace $ and $ {\mixedbirtrace}' $ are different. 
		Based on $ \simreltraces{\mixedbirexecution}{\sbirvaluation}{k}{\mixedsbirexecution} $, we have $	\miximlsbirevent \in \mixedsbirtrace $ at the position $ j $ and $ \miximlbirevent \in \mixedbirtrace$ at the same position $ j $ such that $ \sbirvaluation(\miximlsbirevent) = \miximlbirevent$.
		Based on $ \simreltraces{{\mixedbirexecution}'}{\sbirvaluation}{k}{\mixedsbirexecution} $, we have $	\miximlsbirevent \in \mixedsbirtrace $ at the position $ j $ and $ {\miximlbirevent}' \in {\mixedbirtrace}'$ at the same position $ j $ such that $ \sbirvaluation(\miximlsbirevent) = {\miximlbirevent}'$.
		Therefore, we get that $ \miximlbirevent = \sbirvaluation(\miximlsbirevent) = {\miximlbirevent}'$. 
		Since $\mixedbirtrace \neq {\mixedbirtrace}'$ gives a contradiction, we deduce that $ \mixedbirtrace = {\mixedbirtrace}'  $.
	\end{enumerate}

		Therefore, we can conclude that the function $\zeta$ is injective.
		Since  $ \simreltraces{\mixedbirexecution}{\sbirvaluation}{k}{\mixedsbirexecution} $, we know that
		$ \symfreshv{\symvar{x_i}} $ or $ \imlfreshev{\imlvar{b}} $ steps in $ \mixedsbirexecution $ maps to $ \freshv{\sbirvaluation(\symvar{x_i})} $ or $ \imlfreshev{\imlvar{b}} $ steps in $ \mixedbirexecution $, respectively, for $ i \leq k $. 
		The first statement is as follows:
		
		Based on the rule \textsc{INew} in the operational semantics of $ \imlsymb $ output processes~\cite[p.~23]{aizatulin2015verifying}, we get that the probability of generating the random number $ \imlvar{b} $ based on the $ \imlsymb $ transition relation (i.e. $\imltrans{}{\probdist{n}}{}{} $ semantics) with respect to security parameter $ n $ is $ \probdist{n} $. 
		Therefore, the $ \imlfreshev{\imlvar{b}} $ step has the probability $ \probdist{n} $ in both $ \mixedbirtrace $ and $ \mixedsbirtrace $.
		Hence, for $ \imlfreshev{\cdot} $ steps, we have $ \imlpr(\mixedbirtrace) = \imlpr(\mixedsbirtrace) $.
		
		Based on the rule RNG($n$) in the semantics of \imlsbir{},~\autoref{fig:mixedsymbiml}, we get that the probability of generating the symbolic random number $ \symvar{x} $ based on the mixed $ \imlsymb $ and symbolic transition relation (i.e., $\miximltrans{}{1}{\sbirvaluation}{}{}$ semantics) with respect to security parameter $ n $ is $ 1 $ but we map $ \symfreshv{\symvar{x}} $ to $ \freshv{\sbirvaluation(\symvar{x})} $ using interpretation $ \sbirvaluation $ which have the probability $ 1 \cdot \probdist{n} $.
		The number of $ \symfreshv{\symvar{x}} $ steps in $ \mixedsbirtrace $ are $ k $, hence, the probability of $ \symfreshv{\symvar{x}} $ steps is $ \imlpr(\mixedsbirtrace) \cdot 2^{-n \cdot k }$.

		Based on the rule normal in the mixed semantics of $ \imlsymb $ and $ \birsymb $~\autoref{fig:mixedbiriml} and ~\autoref{rng-fun}, we get that the probability of generating the random number $ \var{x} $ based on the mixed $ \imlsymb $ and $\birsymb$ transition relation (i.e., $\imltrans{}{1}{}{}$ semantics) with respect to security parameter $ n $ is $ 1 $ but we extracting $ \var{x} $ from the random memory $\randommem$ with length $ n $ which have the probability $ 1 \cdot \probdist{n} $.
		The number of $ \freshv{\imlvar{x}} $ steps in $ \mixedbirtrace $ are $ k $, hence, the probability of $ \freshv{\imlvar{x}} $ steps is $ \imlpr(\mixedbirtrace) \cdot 2^{-n \cdot k }$.
		
		Therefore, we can conclude that $ \imlpr(\mixedbirtrace) \cdot 2^{-n \cdot k } = \imlpr(\mixedsbirtrace) \cdot 2^{-n \cdot k }$.

\end{proof}

\subsection{Extension for Multi-Programs}
\label{multi-programs}

In this section, we extend our results for multiple programs by establishing  theorems~\ref{thm:iml:traceeqm} to~\ref{thm:soundnessm}.
Given an $\imlsymb$ process $ \imlprog $ and implementations $\birprogset$ of protocol participants in $ \birsymb $, we denote $ \mixedbirprogs $ which parties running in parallel with an $\imlsymb$ attacker and communicate through a channel.
For $\birprogset$ which are symbolically executed and translated into $\imlsymb$ processes $ \sbirtoimlprogset $, the $\imlsymb$ process $ \imltransprogs $ describes the parallel composition of $ m $ parties in the presence of an attacker.
The following theorem indicates that  \imlsbir{} and $\imlsymb$ preserve the simulation relation for $ m $ programs. 

\begin{theorem}[\imlsbir{}-$\imlsymb$ Trace Inclusion$ ^* $]
	\label{thm:iml:traceeqm}
	Let $\birprogset$ be $\birsymb$ programs,
	$\imlprog$ be an $\imlsymb$ process and
	$ k \in \naturalnum $ is any upper bound on RNG steps, then,
	for all mixed $\imlsymb$ and symb. execution traces
	$\mixedsbirexecution\in\mixedsbirexecutions(\mixedbirprogs, \sbirenv_{\sbirc{0}}[ \randommem \mapsto \randomsymvals_{\sbirc{k}}  , \randommemidx \mapsto 0])$
	such that 
	$ \rng{\mixedsbirexecution} \leq k $,
	there exist an $\imlsymb$ trace 
	$\imlexecution \in \imlexecutions(\imltransprogs) $
	and an interpretation $\sbirvaluation$ such that $\simreltratraces{\mixedsbirexecution}{\sbirvaluation}{\sbirtoiml{.}}{k}{\imlexecution}$.
\end{theorem}

\begin{proof}
	By~\autoref{thm:iml:traceeq}, for each $ 0 < i \leq m $, we have that for all mixed $\imlsymb$ and symb. execution traces
	$\mixedsbirexecution_i \in\mixedsbirexecutions_i(\mixedbirprogidx{\birprog_{\birc{i}}}, \sbirenv_{\sbirc{0}}[ \randommem \mapsto \randomsymvals_{\sbirc{k}} , \randommemidx \mapsto 0])$, exist an $\imlsymb$ trace 
	$\imlexecution_{\imlc{i}} \in \imlexecutions_{\imlc{i}}(\imltransprogidx{\birprog_{\birc{i}}}) $
	and an interpretation $\sbirvaluation_i$ s.t. $\simreltratraces{\mixedsbirexecution_i}{\sbirvaluation_i}{\sbirtoiml{.}}{k}{\imlexecution_{\imlc{i}}}$.
	Then, we can conclude that for all mixed $\imlsymb$ and symb. execution traces
	$\mixedsbirexecution\in\mixedsbirexecutions(\mixedbirprogs, $  $\sbirenv_{\sbirc{0}}[ \randommem \mapsto \randomsymvals_{\sbirc{k}}  , \randommemidx \mapsto 0])$,
	there exist an $\imlsymb$ trace $\imlexecution \in \imlexecutions(\imlprog \{\sbirtoiml{\birprog_{\birc{1}}} , ... , $  $\sbirtoiml{\birprog_{\birc{m}}} 
	\}) $
	and a $\sbirvaluation$ such that $\simreltratraces{\mixedsbirexecution}{\sbirvaluation}{\sbirtoiml{.}}{k}{\imlexecution}$ and $ \sbirvaluation_i \subseteq \sbirvaluation $ for $0 < i \leq m $.
\end{proof}

\autoref{thm:iml:traceeqm} proves that the $\imlsymb$ model resulting from the translation of $ m $ programs covers all behaviors in the mixed $\imlsymb$ and symbolic execution semantics.
To ensure that the extracted $\imlsymb$ model for $ m $ protocol parties preserves the semantics of their implementations in binary, we have to show~\autoref{thm:bir:traceeqm}.

\begin{theorem}[\imlbir{}-\imlsbir{} Trace Inclusion$ ^* $]
	\label{thm:bir:traceeqm}
	Let $\birprogset$ be $\birsymb$ programs,
	$\imlprog$ be an $\imlsymb$ process and
	$ k \in \naturalnum $ is any upper bound on RNG steps, then,
    for all mixed $ \imlsymb $ and $ \birsymb $ traces $\mixedbirexecution\in\mixedbirexecutions(\mixedbirprogs, \birenv_{\birc{0}}[ \randommem \mapsto \randomvals_{\imlc{k}} , \randommemidx \mapsto 0])$ such that $ \rng{\mixedbirexecution} \leq k $, there exist a mixed $ \imlsymb $ and $ \sbirsymb $ trace $\mixedsbirexecution\in\mixedsbirexecutions(\mixedbirprogs,$ $ \sbirenv_{\sbirc{0}}[ \randommem \mapsto \randomsymvals_{\sbirc{k}}  , \randommemidx \mapsto 0])$ and an  $\sbirvaluation$ such that $	\simreltraces{\mixedbirexecution}{\sbirvaluation}{k}{\mixedsbirexecution}$.
\end{theorem}

\begin{proof}
	By~\autoref{thm:bir:traceeq}, for each $0 < i \leq m $, we have for all mixed $\imlsymb$ and $\birsymb$ traces
	$\mixedbirexecution_i \in\mixedbirexecutions_i(\mixedbirprogidx{\birprog_{\birc{i}}}, \birenv_{\birc{0}}[ \randommem \mapsto \randomvals_{\imlc{k}}  , \randommemidx \mapsto 0])$, there is a mixed $\imlsymb$ and symbolic execution trace
	$\mixedsbirexecution_i \in\mixedsbirexecutions_i(\mixedbirprogidx{\birprog_{\birc{i}}}, \sbirenv_{\sbirc{0}}$  $[ \randommem \mapsto \randomsymvals_{\sbirc{k}}  , \randommemidx \mapsto 0])$,
	and an interpretation $\sbirvaluation_i$ such that $\simreltraces{\mixedbirexecution_i}{\sbirvaluation_i}{k}{\mixedsbirexecution_i}$.
	Then, we can conclude that for all mixed $\imlsymb$ and $\birsymb$ traces
	$\mixedbirexecution\in\mixedbirexecutions(\mixedbirprogs, \birenv_{\birc{0}}[ \randommem \mapsto \randomvals_{\imlc{k}}  , \randommemidx \mapsto 0])$,
	there exist a mixed $\imlsymb$ and symbolic execution trace
	$\mixedsbirexecution\in\mixedsbirexecutions(\mixedbirprogs, $  $\sbirenv_{\sbirc{0}}[ \randommem \mapsto \randomsymvals_{\sbirc{k}}  , \randommemidx \mapsto 0])$ 
	and an interpretation $\sbirvaluation$ such that $	\simreltraces{\mixedbirexecution}{\sbirvaluation}{k}{\mixedsbirexecution}$ and $ \sbirvaluation_i \subseteq \sbirvaluation $ for $0 < i \leq m $.
\end{proof}

\autoref{thm:bir:traceeqm} states that symbolically executed $\birsymb$ programs $\birc{\birprog_1} , ... ,$  $\birc{\birprog_m}$ preserve all behaviors of the same programs in the concrete execution for an appropriately chosen interpretation and random memory.
In the following, we measure the success probability of the attacker for $\birsymb$ programs in \imlbir{} execution semantics $ \mixedbirprogs $ and extracted $\imlsymb$ process $ \imltransprogs $.
Then we show by~\autoref{thm:soundnessm} that $ \mixedbirprogs $ is at least as secure as $ \imlprog \{\sbirtoiml{\birprog_{\birc{1}}} , ... , $  $\sbirtoiml{\birprog_{\birc{m}}} 
\} $ with respect to any trace property $ \traceproperty$, security parameter $n$ and upper bound $k$ on the number of \textsc{RNG} steps.

\begin{restatable}[Translation preserves attacks$ ^* $]{theorem}{thmsoundnessm}\label{thm:soundnessm}
	Given an $\imlsymb$ process $\imlprog$,
	$\birsymb$ programs $\birprogset$,
	a security parameter $n \in \naturalnum$,
	a trace property $ \traceproperty$
	and
	an upper bound $k\in \naturalnum$
	on the number of \textsc{RNG} steps in $\mixedbirtraces(\mixedbirprogs, \birenv_{\birc{0}}[ \randommem \mapsto  \randomvals_{\imlc{k}}  , \randommemidx \mapsto 0])$, 
	we get that
	$
	\insecurity{\mixedbirprogs,n,k}{\traceproperty} \le
	\insecurity{\imltransprogs, n}{ \traceproperty}
	$
\end{restatable}

\begin{proof}
	\newcommand{\mDef}[1]{\text{Def.~\ref{#1}}}
	\newcommand{\mLem}[1]{\text{Lem.~\ref{#1}}}
	\newcommand{\mThm}[1]{\text{Thm.~\ref{#1}}}
	
	The insecurity of $\imltransprogs$ w.r.t. $\traceproperty$ is as follows:
	
	$\insecurity{\imltransprogs, n}{ \traceproperty}$
	\begin{align*}
		\stackrel{}{=} &
		\sum_{\imltrace \in \imltraces(\imltransprogs,n)\cap \tracepropertyneg } \imlpr(\imltrace)
		\\
		\intertext{Note that by~\autoref{thm:iml:traceeqm}, we have $ \simreltratraces{\mixedsbirexecution}{\sbirvaluation}{k}{\sbirtoiml{\cdot}}{\imlexecution}$
			which implies $\sbirvaluation(\mixedsbirtrace) \in \tracepropertyneg \iff \imltrace \in \tracepropertyneg$}
		\stackrel{}{=} &
		\sum_{\substack{0 < i \leq m \\}}
		\sum_{
			\substack{ \\ \\
				\imltrace_\imlc{i} \in  \\ \imltraces_\imlc{i}((\imltransprogidx{\birprog_{\birc{i}}},n) \cap \tracepropertyneg\\
			}
		}  \imlpr(\imltrace_\imlc{i})
		\\ \\
		\stackrel{\mDef{def:IML-insec}}{=} &
		\sum_{\substack{0 < i \leq m \\}} \insecurity{\imltransprogidx{\birprog_{\birc{i}}}, n}{ \traceproperty}
		\\ \\
		\stackrel{\mThm{thm:soundness}}{\ge} &
		\sum_{\substack{0 <  i \leq m \\}}
		\insecurity{\mixedbirprogidx{\birprog_{\birc{i}}},n,k}{\traceproperty}\\ \\
		\stackrel{\mDef{def:BIR-insec}}{=} &
		\sum_{\substack{0 <  i \leq m \\}}
		2^{-n \cdot k} \cdot \sum_{\substack{
				\randomvals_{\imlc{k}} \in \imlvals_n^k \\
				\mixedbirtrace_i \in  \tracepropertyneg \cap \\
				\mixedbirtraces_i(\mixedbirprogidx{\birprog_{\birc{i}}}, \birenv_{\birc{0}}[ \randommem \mapsto  \randomvals_{\imlc{k}}  , \randommemidx \mapsto 0]) }}
		\imlpr(\mixedbirtrace_i)
		\\
		\intertext{Note that by~\autoref{thm:bir:traceeqm}, we have $ \simreltraces{\mixedbirexecution}{\sbirvaluation}{k}{\mixedsbirexecution}$
			which implies $\sbirvaluation(\mixedsbirtrace) \in \tracepropertyneg \iff \mixedbirtrace \in \tracepropertyneg$}
		= \ &
		2^{-nk} \cdot
		\sum_{\substack{
				\randomvals_{\imlc{k}} \in \imlvals_n^k \\
				\mixedbirtrace\in\mixedbirtraces \cap \tracepropertyneg
		}} \imlpr(\mixedbirtrace)
		\\ \\
		= \ &
		\insecurity{\mixedbirprogs,n,k}{\traceproperty}
	\end{align*}
\end{proof}
\end{full}
\end{document}